\documentclass[11pt]{article}
\usepackage{amsmath, amssymb, amsfonts, amsthm}
\usepackage{epsfig}

\newtheorem{theorem}{Theorem}[section]
\newtheorem{proposition}[theorem]{Proposition}
\newtheorem{corollary}[theorem]{Corollary}
\newtheorem{lemma}[theorem]{Lemma}

\newcommand{\rd}{{\rm d}}
\newcommand{\be}{\begin{equation}}
\newcommand{\ee}{\end{equation}}
\newcommand{\bey}{\begin{eqnarray}}
\newcommand{\eey}{\end{eqnarray}}

\newcommand{\eps}{\varepsilon}

\newcommand{\bx}{{\bf x}}

\newcommand{\ph}{\varphi}

\newcommand{\la}{\langle}
\newcommand{\ra}{\rangle}

\renewcommand{\a}{\alpha}

\newcommand{\cU}{{\cal U}}

\newcommand{\bR}{{\mathbb R}}
\newcommand{\bC}{{\mathbb C}}
\newcommand{\bN}{{\mathbb N}}

\newcommand{\tr}{\mbox{Tr}}

\newcommand{\wt}{\widetilde}
\newcommand{\wh}{\widehat}

\newcommand{\const}{\mathrm{const}}

\newcommand{\cF}{{\cal F}}

\newcommand{\cE}{{\cal E}}

\newcommand{\cW}{{\cal W}}
\newcommand{\cK}{{\cal K}}
\newcommand{\cH}{{\cal H}}
\newcommand{\cL}{{\cal L}}
\newcommand{\cN}{{\cal N}}

\input epsf

\newcommand{\fh}{{\frak h}}

\newcommand{\donothing}[1]{}

\begin{document}

\title{Derivation of Effective Evolution Equations from \\ Microscopic Quantum Dynamics}

\author{Benjamin Schlein\thanks{On leave from Cambridge University. Supported by a Kovalevskaja Award from the Humboldt Foundation.}\\
\\
Institute of Mathematics, University of Munich, \\
Theresienstr.~39, D-80333 Munich, Germany}

\maketitle

\tableofcontents

\section{Introduction}
\setcounter{equation}{0}

A quantum mechanical system of $N$ particles in $d$ dimensions can be described by a complex valued wave function $\psi_N \in L^2 (\bR^{d N}, \rd x_1 \dots \rd x_N)$. The variables $x_1, \dots ,x_N \in \bR^d$ represent the position of the $N$ particles. Physically, the absolute value squared of $\psi_N (x_1, x_2, \dots ,x_N)$ is interpreted as the probability density for finding particle one at $x_1$, particle two at $x_2$, and so on. Because of this probabilistic interpretation, we will always consider wave functions $\psi_N$ with $L^2$-norm equal to one.

\medskip

In Nature there exist two different types of particles; bosons and fermions. Bosonic systems are described by wave functions which are symmetric with respect to permutations, in the sense that
\begin{equation}\label{eq:perm} \psi_N (x_{\pi 1} , x_{\pi 2}, \dots , x_{\pi N}) = \psi_N (x_1, \dots , x_N) \end{equation}
for every permutation $\pi \in S_N$. Fermionic systems, on the other hand, are described by antisymmetric wave functions satisfying \[ \psi_N (x_{\pi 1} , x_{\pi 2}, \dots , x_{\pi N}) = \sigma_{\pi} \psi_N (x_1, \dots , x_N) \qquad \text{for all $\pi \in S_N$,} \] where $\sigma_{\pi}$ is the sign of the permutation $\pi$; $\sigma_{\pi} = + 1$ if $\pi$ is even (in the sense that it can be written as the composition of an even number of transpositions) and $\sigma_{\pi} = -1$ if it is odd. In these notes we are only going to consider bosonic systems; the wave function $\psi_N$ will always be taken from the Hilbert space $L^2_s (\bR^{d N})$, the subspace of $L^2 (\bR^{d N})$ consisting of all functions satisfying (\ref{eq:perm}).

\medskip

Observables of the $N$-particle system are self adjoint operators over $L^2_s (\bR^{dN})$. The expected value of an observable $A$ in a state described by the wave function $\psi_N$ is given by the inner product \[ \langle \psi_N, A \psi_N \rangle = \int \rd x_1 \dots \rd x_N \, \overline{\psi}_N (x_1, \dots , x_N) \, (A \psi_N ) (x_1 , \dots , x_N) .\] The multiplication operator $x_j$ is the observable measuring the position of the $j$-th particle. The differential operator $p_j = -i \nabla_j$ is the observable measuring the momentum of the $j$-th particle ($p_j$ is called the momentum operator of the $j$-th particle).

\medskip

The time evolution of an $N$-particle wave function $\psi_N \in L^2_s (\bR^{d N})$ is governed by the Schr\"odinger equation
\begin{equation}\label{eq:schr}
i \partial_t \psi_{N,t} = H_N \psi_{N,t}\,.
\end{equation}
Here, and in the rest of these notes, the subscript $t$ indicates the time dependence of the wave function; all time-derivatives will be explicitly written as $\partial_t$. On the right hand side of (\ref{eq:schr}), $H_N$ is a self-adjoint operator acting on the Hilbert space $L^2_s (\bR^{d N})$, usually known as the Hamilton operator (or Hamiltonian) of the system. We will consider only time-independent Hamilton operators with two body interactions, which have the form
\begin{equation*}
H_N = \sum_{j=1}^N \left( -\Delta_{x_j} + V_{\text{ext}} (x_j) \right) + \lambda \sum_{i<j}^N V (x_i -x_j)\,.
\end{equation*}
The first part of the Hamiltonian is a sum of one-body operators (operators acting on one particle only); the sum of the Laplacians is the kinetic part of the Hamiltonian. The function $V_{\text{ext}}$ describes an { external potential} which acts in the same way on all $N$ particles. For example, $V_{\text{ext}}$ may describe a confining potential which traps the particles in a certain region. The second part of the Hamiltonian, given by a sum over all pairs of particles, describes the interactions among the particles ($\lambda \in \bR$ is a coupling constant). The Hamilton operator is the observable associated with the energy of the system. In other words, the expectation \[ \langle \psi_N , H_N \psi_N \rangle = \int \rd x_1 \dots \rd x_N \; \overline{\psi}_N (x_1, \dots ,x_N) (H_N \psi_N ) (x_1 , \dots ,x_N) \] gives the energy of the system in the state described by the wave function $\psi_N$.

\medskip

Note that the Schr\"odinger equation (\ref{eq:schr}) is linear and, since $H_N$ is a self-adjoint operator, it preserves the $L^2$-norm of the wave function (moreover, since $H_N$ is invariant with respect to permutations, it also preserves the permutation symmetry of the wave function).
In fact, the solution to (\ref{eq:schr}), with initial condition $\psi_{N,t=0} = \psi_N$, can be written by means of the unitary group generated by $H_N$ as \begin{equation}\label{eq:schr-sol} \psi_{N,t} = e^{-iH_N t} \psi_N \qquad \text{for all } t \in \bR \, .\end{equation} The global well-posedness of (\ref{eq:schr}) is not an issue here. The study of (\ref{eq:schr}) is focused, therefore, on other questions concerning the qualitative and quantitative behavior
of the solution $\psi_{N,t}$. Despite the linearity of the equation, these questions are usually quite hard to answer, because in physically interesting situation the number of particles $N$ is very large; it varies between $N \simeq 10^3$ for very dilute Bose-Einstein samples, up to values of the order $N \simeq 10^{30}$ in stars. For such huge values of $N$, it is of course impossible to compute the solution (\ref{eq:schr-sol}) explicitly; numerical methods are completely useless as well (unless the interaction among the particles is switched off).

\medskip

Fortunately, also from the point of view of physics, it is not so important to know the precise solution to (\ref{eq:schr}); it is much more important, for physicists performing experiments, to have information about the macroscopic properties of the system, which describe the typical behavior of the particles, and result from averaging over a large number of particles. Restricting the attention to macroscopic quantities simplifies the study of the solution $\psi_{N,t}$, but it still does not make it accessible to mathematical analysis. To further simplify matters, we are going to let the number of particles $N$ tend to infinity. The macroscopic properties of the system, computed in the limiting regime $N \to \infty$, are then expected to be a good approximation for the macroscopic properties observed in experiments, where the number of particles $N$ is very large, but finite (in some cases it is possible to obtain explicit bounds on the difference between the limiting behavior as $N \to \infty$ and the behavior for large but finite $N$; see Section \ref{sec:coh}).

\medskip

To consider the limit of large $N$, we are going to make use of the marginal or reduced density matrices associated with an $N$ particle wave function $\psi_N \in L^2_s (\bR^{d N})$. First of all, we define the density matrix $\gamma_N = |\psi_N \rangle \langle \psi_N|$ associated with $\psi_N$ as the orthogonal projection onto $\psi_N$; we use here and henceforth the notation $|\psi \rangle \langle \psi|$ to indicate the orthogonal projection onto $\psi$ (Dirac bracket notation). Note that, expressed through the density matrix $\gamma_N$, the expectation $\langle\psi_N , A \psi_N \rangle$ of the observable $A$ can be written as $\tr \, A \gamma_N$. The kernel of the operator $\gamma_N$ is then given by \[ \gamma_N (\bx ; \bx') = \psi (\bx) \overline{\psi} (\bx') \, ,\] where we introduced the notation $\bx = (x_1, \dots ,x_N), \bx'= (x'_1,\dots, x'_N) \in \bR^{dN}$. Note that the $L^2$-normalization of $\psi_N$ implies that $\tr \, \gamma_N =1$. For $k =1 ,\dots ,N$, we define the $k$-particle marginal density $\gamma^{(k)}_{N}$ associated with $\psi_N$ as the partial trace of $\gamma_N$ over the degrees of freedom of the last $(N-k)$ particles: \[ \gamma^{(k)}_{N} = \tr_{k+1, k+2, \dots, N} \; |\psi_{N} \rangle \langle \psi_N | \,  \] where $\tr_{k+1, \dots, N}$ denotes the partial trace over the particle $k+1, k+2, \dots ,N$. In other words, $\gamma_{N}^{(k)}$ is defined as the non-negative trace class operator on $L^2_s (\bR^{d k})$ with kernel given by
\begin{equation}\label{eq:gammakN}
\begin{split}
\gamma^{(k)}_{N} (\bx_k ; \bx'_k) &= \int \rd \bx_{N-k} \, \gamma_{N} (\bx_k , \bx_{N-k} ; \bx'_k, \bx_{N-k} ) \\ &= \int \rd \bx_{N-k} \, \overline{\psi}_{N} (\bx_k , \bx_{N-k}) \, \psi_{N} (\bx'_k, \bx_{N-k})\,.
\end{split}
\end{equation}
The last equation can be considered as the definition of partial trace. Here we used the notation $\bx_k = (x_1, \dots ,x_k), \bx'_k = (x'_1 , \dots ,x'_k) \in \bR^{dk}$ and $\bx_{N-k} = (x_{k+1}, \dots ,x_N) \in \bR^{d(N-k)}$. By definition, $\tr \, \gamma^{(k)}_{N} = 1$ for all $N$ and for all $k= 1, \dots ,N$ (note that, in the physics literature, one normally uses a different normalization of the reduced density matrices). For fixed $k <N$, the $k$-particle density matrix does not contain the full information about the state described by $\psi_N$. Knowledge of the $k$-particle marginal $\gamma^{(k)}_N$, however, is sufficient to compute the expectation of every $k$-particle observable in the state described by the wave function $\psi_N$. In fact, if $A^{(k)}$ denotes an arbitrary
bounded operator on $L^2 (\bR^{d k})$, and if $A^{(k)} \otimes 1^{(N-k)}$ denotes the operator on $L^2 (\bR^{d N})$ which acts as $A^{(k)}$ on the first $k$ particles, and as the identity on the last $(N-k)$ particles, we have
\begin{equation}\label{eq:kobs} \left\langle \psi_{N} , \left(A^{(k)} \otimes 1^{(N-k)}\right)  \psi_{N} \right\rangle = \tr \;\left(A^{(k)} \otimes 1^{(N-k)}\right) \gamma_{N} = \tr \; A^{(k)} \gamma^{(k)}_{N} .\end{equation}
Thus, $\gamma^{(k)}_{N}$ is sufficient to compute the expectation of arbitrary observables which depend non-trivially on at most $k$ particles (because of the permutation symmetry, it is not important on which particles it acts, just that it acts at most on $k$ particles).

\medskip

Marginal densities play an important role in the analysis of the $N \to \infty$ limit because, in contrast to the wave function $\psi_N$ and to the density matrix $\gamma_N$, the $k$-particle marginal $\gamma^{(k)}_N$ can have, for every fixed $k \in \bN$, a well-defined limit as $N \to \infty$ (because, if we fix $k \in \bN$, $\{ \gamma^{(k)}_N \}$ defines a sequence of operators all acting on the same space $L^2 (\bR^{d k})$).

\medskip

In these notes we are going to study macroscopic properties of the dynamics of bosonic $N$-particle systems, in the limit $N \to \infty$. We are interested in the time-evolution of marginal densities $\gamma^{(k)}_{N,t}$ associated with the solution $\psi_{N,t}$ to the Schr\"odinger equation (\ref{eq:schr}), for fixed $k$, and as $N \to \infty$. Unfortunately, it is not so simple to describe the time-dependence of $\gamma_{N,t}^{(k)}$ in the limit of large $N$; in fact it is in general impossible to obtain closed equations for the evolution of the limiting $k$-particle density $\gamma_{\infty,t}^{(k)} = \lim_{N \to \infty} \gamma_{N,t}^{(k)}$ (in general it is not even clear that this limit exists).
Nevertheless, there are some physically interesting situations for which it is indeed possible to prove the existence of $\gamma^{(k)}_{\infty,t}$ and to derive closed equations to describe its dynamics. In Section \ref{sec:mf}, we are going to study the time evolution of factorized initial wave functions of the form $\psi_N (\bx) = \prod_{j=1}^N \ph (x_j)$ in the so-called mean field limit. We will show that in this case, for every fixed $k \in \bN$, the $k$-particle marginal associated with the solution to the Schr\"odinger equation $\psi_{N,t}$ converges, as $N \to \infty$, to the limiting $k$-particle density $\gamma^{(k)}_{\infty,t} = |\ph_t \rangle \langle \ph_t|^{\otimes k}$, where $\ph_t$ is the solution of a certain one-particle nonlinear Schr\"odinger equation, known as the Hartree equation. In Section \ref{sec:GP}, we are going to study the time-evolution of Bose-Einstein condensates, in the so-called Gross-Pitaevskii limit. As we will see, although it describes a very different physical situation, the Gross-Pitaevskii scaling can be formally interpreted as a mean-field limit, with a very singular interaction potentials. Also in this case we will prove that the time evolution of the marginal densities can be described through a one-particle nonlinear Schr\"odinger equation (known as the time-dependent Gross-Pitaevskii equation); however, because of the singularity of the interaction potential, the analysis in this case is going to be much more involved. In Section \ref{sec:RS}, we will come back to the study of mean field models, and we will discuss how to prove quantitative estimates on the rate of convergence towards the Hartree dynamics.

\medskip

{\it Notation.} Throughout these notes, we will make use of the notation $\bx = (x_1, \dots ,x_N), \bx' = (x_1 , \dots x'_N) \in \bR^{dN}$, and for $k =1, \dots ,N$, $\bx_k = (x_1 , \dots ,x_k), \bx'_k = (x'_1, \dots , x'_k) \in \bR^{dk}$, and $\bx_{N-k} = (x_{k+1}, \dots ,x_N) \in \bR^{d(N-k)}$ (starting from Section \ref{sec:GP}, we will fix $d=3$). We will also use the shorthand notation $\nabla_j = \nabla_{x_j}$ and $\Delta_j = \Delta_{x_j}$.

\section{Derivation of the Hartree Equation in the Mean Field Limit}\label{sec:mf}
\setcounter{equation}{0}

\subsection{Mean Field Systems}

A mean-field system is described by an $N$-body Hamilton operator of the form
\begin{equation}\label{eq:ham-mf}
H_N = \sum_{j=1}^N \left( -\Delta_j + V_{\text{ext}} (x_j) \right)+ \frac{1}{N} \sum_{i<j}^N V (x_i -x_j)
\end{equation}
acting on the Hilbert space $L^2_s (\bR^{d N})$, the subspace of $L^2 (\bR^{dN})$ consisting of permutation symmetric functions. In these notes we will only discuss bosonic systems, which are described by symmetric wave functions. Note, however, that the mean field limit for fermionic system has also been considered in the literature; see, for example, \cite{NS,S,EESY0}. In (\ref{eq:ham-mf}) and henceforth, we use the notation $\Delta_j = \Delta_{x_j}$ (similarly, we will use the notation $\nabla_j = \nabla_{x_j}$). The mean-field character of the Hamiltonian is expressed by the factor $1/N$ in front of the interaction; this factor guarantees that the kinetic and potential energies are typically both of order $N$.

\medskip

We are going to be interested in the solution $\psi_{N,t} = e^{-i H_N t} \psi_N$ of the Schr\"odinger equation (\ref{eq:schr}) with Hamiltonian $H_N$ given by (\ref{eq:ham-mf}) and with factorized initial data
\begin{equation}\label{eq:fact}
\psi_N (\bx) = \prod_{j=1}^N \ph (x_j) \, ,
\end{equation}
for some $\ph \in L^2 (\bR^{d})$. The physical motivation for studying the evolution of factorized wave functions is that states close to the ground state of $H_N$ (the eigenvector associated with the lowest eigenvalue), which are the most accessible and thus the most interesting states, can be approximately described by wave functions like (\ref{eq:fact}) (the results which we are going to discuss in this section do not require strict factorization as in (\ref{eq:fact}); instead condensation of the initial wave function in the sense of (\ref{eq:BEC}) would be sufficient).

\medskip

Because of the interaction among the particles, the factorization (\ref{eq:fact}) is not preserved by the time evolution; in other words, the evolved wave function $\psi_{N,t}$ is not given by the product of one-particle wave functions, if $t \neq 0$. However, due to the mean-field character of the interaction each particle interacts very weakly (the strength of the interaction is of the order $1/N$) with all other $(N-1)$ particles (at least in the initial state, every particle is described by the same one-particle orbital; every particles therefore ``sees'' all other particles). For this reason, we may expect that, in the limit of large $N$, the total interaction potential experienced by a typical particle in the system can be effectively replaced by an averaged, mean-field, potential, and therefore that factorization is approximately, and in an appropriate sense, preserved by the time evolution. In other words, we may expect that, in a sense to be made precise,
\begin{equation}\label{eq:factt}
\psi_{N,t} (x_1, \dots ,x_N) \simeq \prod_{j=1}^N \ph_t (x_j) \qquad \text{as } N \to \infty
\end{equation}
for an evolved one-particle orbital $\ph_t$. Assuming (\ref{eq:factt}), it is simple to derive a self-consistent equation for the time-evolution of the one-particle orbital $\ph_t$. In fact, (\ref{eq:factt}) states that, for every fixed time $t$, the $N$ particles are independently distributed in space with density $|\ph_t (x)|^2$. If this is true, the total potential experienced, for example, by the first particle can be approximated by \[ \frac{1}{N} \sum_{j \geq 2} V (x_1 -x_j) \simeq \frac{1}{N} \sum_{j \geq 2} \int \rd y \, V (x_1 - y) |\ph_t (y)|^2 = \frac{N-1}{N} (V * |\ph_t|^2) \simeq (V * |\ph_t|^2) \, \] as $N \to \infty$. It follows that, if (\ref{eq:factt}) holds true, the one-particle orbital $\ph_t$ must satisfy the self-consistent equation \begin{equation}\label{eq:hartree1}
i\partial_t \ph_t  = - \Delta \ph_t + ( V * |\ph_t|^2 ) \ph_t
\end{equation}
with initial data $\ph_{t=0} = \ph$ given by (\ref{eq:fact}). Equation (\ref{eq:hartree1}) is known as the Hartree equation; it is an example of a cubic nonlinear Schr\"odinger equation on $\bR^{d}$. Starting from the linear Schr\"odinger equation (\ref{eq:schr}) on $\bR^{d N}$, we obtain, for the evolution of factorized wave functions, a nonlinear Schr\"odinger equation on $\bR^{d}$; the nonlinearity in the Hartree equation is a consequence of the many-body effects in the linear dynamics.

\medskip

The convergence of $\psi_{N,t}$ to the factorized wave function on the r.h.s. of (\ref{eq:factt}) as $N \to \infty$ cannot hold in the $L^2$-sense; we cannot expect, in other words, that $\| \psi_{N,t} - \ph_t^{\otimes N} \| \to 0$ as $N \to \infty$ (we use here the notation $\ph^{\otimes N} (\bx) = \prod_{j=1}^N \ph (x_j)$). Instead, (\ref{eq:factt}) has to be understood as convergence of marginal densities. Recall that, for $k =1,
\dots ,N$, the $k$-particle marginal $\gamma_{N,t}^{(k)}$ associated with $\psi_{N,t}$ is defined as the non-negative trace class operator on $L^2 (\bR^{d k})$ with kernel given by
\begin{equation}\label{eq:gammakNt}
\begin{split}
\gamma^{(k)}_{N,t} (\bx_k ; \bx'_k) &= \int \rd \bx_{N-k} \, \gamma_{N,t} (\bx_k , \bx_{N-k} ; \bx'_k, \bx_{N-k} ) \\ &= \int \rd \bx_{N-k} \, \overline{\psi}_{N,t} (\bx_k , \bx_{N-k}) \, \psi_{N,t} (\bx'_k, \bx_{N-k})\,.
\end{split}
\end{equation}
It turns out that (\ref{eq:factt}) holds in the sense that, for every fixed $k \in \bN$, the $k$-particle marginal density associated with the left hand side converges, as $N \to \infty$, to the $k$-particle marginal density associated with the right hand side (which is independent of $N$, if $N \geq k$). In other words, assuming (\ref{eq:fact}), one can show that, for a large class of interaction potentials, and for every fixed $t \in \bR$ and $k \in \bN$,
\begin{equation}\label{eq:conv} \tr \; \left| \gamma^{(k)}_{N,t} - |\ph_t \rangle \langle \ph_t|^{\otimes k} \right| \, \to \, 0 \qquad \text{as } N \to \infty \, ,
\end{equation}
where $\ph_t$ is the solution to the Hartree equation (\ref{eq:hartree1}) with initial data $\ph_{t=0} = \ph$. It is important here that the time $t \in \bR$ and the integer $k \geq 1$ are fixed; the convergence is not uniform in these two parameters. {F}rom (\ref{eq:kobs}), we observe that (\ref{eq:conv}) implies (and is actually equivalent to the condition) that, for every fixed $t\in \bR$ and $k\in\bN$, and for every fixed compact operator $J^{(k)}$ on $L^2_s (\bR^{d k})$,
\begin{equation}\label{eq:conv2}
\left\langle \psi_{N,t} , \left( J^{(k)} \otimes 1^{(N-k)} \right) \psi_{N,t} \right\rangle  \to \langle \ph_t^{\otimes k} , J^{(k)} \ph_t^{\otimes k} \rangle
\end{equation}
as $N \to \infty$. This means that, if we are interested in the expectation of observables which depend non-trivially on a fixed number $k$ of particles, then we can approximate, as $N\to \infty$, the true solution $\psi_{N,t}$ to the $N$-body Schr\"odinger equation by the product of $N$ copies of the solution $\ph_t$ to the Hartree equation (\ref{eq:hartree1}). This approximation, however, is in general not valid if we are interested in the expectation of observables depending on a macroscopic (that is, proportional to $N$) number of particles.


\medskip

The first rigorous results establishing a relation between the many body Schr\"odinger evolution and the nonlinear Hartree dynamics were obtained by Hepp in \cite{Hepp} (for smooth interaction potentials) and then generalized by Ginibre and Velo to singular potentials in \cite{GV1}. These works were inspired by techniques used in quantum field theory. We will discuss this method in Section \ref{sec:RS}, where we present a recent proof of  (\ref{eq:conv}), obtained in collaboration with I. Rodnianski in \cite{RS}, which provides a quantitative control of the rate of convergence and makes use of the original idea of Hepp.

\medskip

The first proof of the convergence (\ref{eq:conv}) was obtained by Spohn in \cite{Sp}, for bounded potentials. The method introduced by Spohn was then extended to singular potentials. In \cite{EY}, Erd\H os and Yau proved (\ref{eq:conv}) for a Coulomb potential $V(x) = \pm 1/|x|$; partial results for the Coulomb potential were also obtained by Bardos, Golse and Mauser in \cite{BGM} (note that recently a new proof of (\ref{eq:conv}) for the case of a Coulomb interaction has been proposed by Fr\"ohlich, Knowles, and Schwarz in \cite{FKS}). In \cite{ES}, a joint work with A.~Elgart, we considered again the Coulomb potential, but this time assuming a relativistic dispersion for the bosons. Recently, in a series of papers \cite{ESY0,ESY2,ESY3,ESY4,ESY5} in collaboration with L. Erd\H os and H.-T. Yau (and also in \cite{EESY}, a collaboration with A. Elgart, L. Erd\H os and H.-T. Yau) the strategy of \cite{Sp} was applied to systems with an $N$-dependent interaction potential, which converges, in the limit $N \to \infty$, to a delta-function. Note that in the one-dimensional setting, potentials converging to a delta-interaction have been considered by Adami, Golse and Teta in \cite{AGT} (making use of previous results obtained by the same authors in collaboration with Bardos in \cite{ABGT}) . We will discuss these systems in Section \ref{sec:GP}.

\medskip

Recently, a different approach to the proof of (\ref{eq:conv}) has been proposed by Fr\"ohlich, Schwarz and Graffi in \cite{FGS}. For smooth potentials, they can consider the mean-field limit uniformly in Planck's constant $\hbar$ (up to errors exponentially small in time); this allows them to combine the semiclassical limit and the mean field limit. It is also interesting to remark that the mean-field limit (\ref{eq:conv}) can be interpreted as a Egorov-type theorem; this was observed by Fr\"ohlich, Knowles, and Pizzo in \cite{FKP}.

\subsection{Derivation of the Hartree Equation for Bounded Potentials}
\label{sec:bd}

We consider, in this section, the dynamics generated by the mean field Hamiltonian (\ref{eq:ham-mf}) under the assumption that the interaction potential is a bounded operator.  We will assume, in other words, that $V \in L^{\infty} (\bR^{d})$ (recall that the operator norm of the multiplication operator $V(x_i -x_j)$ is given by the $L^{\infty}$-norm of the function $V$). To simplify a little bit the analysis we will also assume the external potential $V_{\text{ext}}$ in the Hamiltonian (\ref{eq:ham-mf}) to vanish; the techniques discussed here can however be easily extended to $V_{\text{ext}} \neq 0$.

\begin{theorem}[Spohn, \cite{Sp}]\label{thm:bd}
Suppose that \[ H_N = \sum_{j=1}^N -\Delta_j + \frac{1}{N} \sum_{i<j} V (x_i -x_j) \] with $V \in L^{\infty} (\bR^d)$. Let $\psi_N = \ph^{\otimes N} \in L^2 (\bR^{dN})$ for some $\ph \in L^2(\bR^d)$ with $\| \ph \| =1$. Let $\psi_{N,t} = e^{-iH_N t} \psi_N$, and denote by $\gamma^{(k)}_{N,t}$ the $k$-particle marginal density associated with $\psi_{N,t}$. Then, for every fixed $t \in \bR$, and for every fixed $k \geq 1$, we have
\begin{equation}\label{eq:claimbd}
\tr \; \left| \gamma^{(k)}_{N,t} - |\ph_t \rangle \langle \ph_t|^{\otimes k} \right| \, \to \, 0 \end{equation}
as $N \to \infty$. Here $\ph_t$ denotes the solution to the Hartree equation
\begin{equation}\label{eq:hartree}
i\partial_t \ph_t =-\Delta \ph_t + (V * |\ph_t|^2) \ph_t
\end{equation}
with initial data $\ph_{t=0} = \ph$.
\end{theorem}

\begin{proof}
The proof is based on the study of the time evolution of the marginal densities $\gamma^{(k)}_{N,t}$ in the limit $N \to \infty$. {F}rom (\ref{eq:schr}), it is simple to show that the dynamics of the marginals is governed by a hierarchy of $N$ coupled equation, commonly known as the BBGKY hierarchy:
\begin{equation}\label{eq:BBGKY}
\begin{split}
i\partial_t \gamma^{(k)}_{N,t} =\; & \sum_{j=1}^k \left[ -\Delta_{j} , \gamma^{(k)}_{N,t} \right] + \frac{1}{N} \sum_{i<j}^k \left[ V (x_i -x_j) , \gamma^{(k)}_{N,t} \right]
\\ &+ \frac{(N-k)}{N} \sum_{j=1}^k \tr_{k+1} \, \left[ V (x_j -x_{k+1}), \gamma^{(k+1)}_{N,t} \right] \,.
\end{split}
\end{equation}
We use here the convention that $\gamma^{(k)}_{N,t} = 0$ if $k >N$. Moreover $[A,B]= AB -BA$ denotes the commutator of the two operators $A$ and $B$. The symbol $\tr_{k+1}$ denotes the partial trace over the $(k+1)$-th particle; the kernel of the $k$-particle operator $\tr_{k+1} \, [ V (x_j -x_{k+1}) , \gamma^{(k+1)}_{N,t} ]$ is given by
\begin{equation}\label{eq:trk+1}
\begin{split}
\Big(\tr_{k+1} &\left[ V (x_j - x_{k+1}) ,  \gamma^{(k+1)}_{N,t} \right] \Big) (\bx_k ; \bx'_k) \\ &= \int \rd x_{k+1} \, \left(V (x_j -x_{k+1}) - V (x'_j - x_{k+1}) \right) \gamma^{(k+1)} (\bx_k ,x_{k+1}; \bx'_k, x_{k+1}) \, . \end{split}
\end{equation}

Rewriting the BBGKY hierarchy (\ref{eq:BBGKY}) in integral form, we find
\begin{equation}\label{eq:bd1}
\gamma_{N,t}^{(k)} = \cU^{(k)} (t) \gamma^{(k)} + \frac{1}{N} \int_0^t \rd s \, \cU^{(k)} (t-s) A^{(k)} \gamma^{(k)}_{N,s} + \left( 1-\frac{k}{N}\right) \int_0^t \rd s \, \cU^{(k)} (t-s) B^{(k)} \gamma^{(k+1)}_{N,s}
\end{equation}
where $\cU^{(k)} (t)$ denotes the free evolution of $k$ particles, defined  by
\begin{equation}\label{eq:free}
\cU^{(k)} (t) \gamma^{(k)} = e^{it\sum_{j=1}^k \Delta_j} \gamma^{(k)} e^{-it\sum_{j=1}^k \Delta_j} \end{equation}
and the maps $A^{(k)}$ and $B^{(k)}$ are defined by
\begin{equation}\label{eq:Ak} A^{(k)} \gamma^{(k)} = -i \sum_{j=1}^k \left[ V(x_i -x_j), \gamma^{(k)} \right] \end{equation}
and, respectively, by
\begin{equation}\label{eq:Bk} B^{(k)} \gamma^{(k+1)} = -i\sum_{j=1}^k \tr_{k+1} \, \left[ V (x_j -x_{k+1}) , \gamma^{(k+1)} \right] \,. \end{equation} Note that $B^{(k)}$ maps $(k+1)$-particle operators into $k$-particle operators (while $A^{(k)}$ maps $k$-particle operators into $k$-particle operators). Since we are interested in the limit $N \to \infty$ with fixed $k \geq 1$, it is clear that the second term on the r.h.s. of (\ref{eq:bd1}), as well as the contribution proportional to $k/N$ to the third term on the r.h.s. of (\ref{eq:bd1}) should be considered as small perturbations. Iterating the integral equation (\ref{eq:bd1}) for $n$ times, and stopping the iteration every time we hit a perturbation, we obtain the Duhamel type series
\begin{equation}\label{eq:duhbd1}
\begin{split}
\gamma_{N,t}^{(k)} = \; &\cU^{(k)} (t) \gamma^{(k)}_0 \\ &+ \sum_{m=1}^{n-1} \int_0^t \rd s_1 \dots \int_0^{s_{m-1}} \rd s_m \, \cU^{(k)} (t-s_1) B^{(k)} \cU^{(k+1)} (s_1 -s_2) \dots B^{(k+m-1)} \cU^{(k+m)} (s_m) \gamma_0^{(k+m)} \\ &+\int_0^t \rd s_1 \dots \int_0^{s_{n-1}} \rd s_n \, \cU^{(k)} (t-s_1) B^{(k)} \cU^{(k+1)} (s_1 -s_2) \dots B^{(k+n-1)} \gamma_{N,s_n}^{(k+n)} \\ &+ \frac{1}{N} \sum_{m=1}^N \int_0^t \rd s_1 \dots \int_0^{s_{m-1}} \rd s_m \cU^{(k)} (t-s_1) B^{(k)} \dots \cU^{(k+m-1)} (s_{m-1} -s_m) A^{(k+m-1)} \gamma^{(k+m-1)}_{N,s_m} \\
&- \sum_{m=1}^n \frac{k+m-1}{N} \int_0^t \rd s_1 \dots \int_0^{s_{m-1}} \rd s_m \, \cU^{(k)} (t-s_1) B^{(k)} \dots B^{(k+m-1)} \gamma^{(k+m)}_{N,s_m}\,.
\end{split}
\end{equation}
To show (\ref{eq:claimbd}), we need to compare $\gamma_{N,t}^{(k)}$ with $\gamma_{\infty,t}^{(k)} = |\ph_t \rangle \langle \ph_t|^{\otimes k}$, where $\ph_t$ is the solution to the Hartree equation (\ref{eq:hartree}). It is simple to check that the family $\{ \gamma^{(k)}_{\infty,t} \}_{k \geq 1}$ solves the infinite hierarchy (written directly in integral form)
\begin{equation}\label{eq:infhierbd}
\gamma^{(k)}_{\infty,t} = \cU^{(k)} (t) \gamma^{(k)}_0 + \int_0^t \rd s \, \cU^{(k)} (t-s) B^{(k)} \gamma^{(k+1)}_{\infty,s}
\end{equation}
which leads, after iteration, to the expansion
\begin{equation}\label{eq:duhbd2}
\begin{split}
\gamma^{(k)}_{\infty,t} = \; &\cU^{(k)} (t) \gamma^{(k)}_0 \\ &+ \sum_{m=1}^{n-1} \int_0^t \rd s_1 \dots \int_0^{s_{m-1}} \rd s_m \, \cU^{(k)} (t-s_1) B^{(k)} \cU^{(k+1)} (s_1 -s_2) \dots B^{(k+m-1)} \cU^{(k+m)} (s_m) \gamma_0^{(k+m)} \\ &+\int_0^t \rd s_1 \dots \int_0^{s_{n-1}} \rd s_n \, \cU^{(k)} (t-s_1) B^{(k)} \cU^{(k+1)} (s_1 -s_2) \dots B^{(k+n-1)} \gamma_{\infty,s_n}^{(k+n)} \,.
\end{split}
\end{equation}
The difference between $\gamma_{N,t}^{(k)}$ and $\gamma^{(k)}_{\infty,t}$ can thus be bounded by
\begin{equation}\label{eq:duhbd3}
\begin{split}
\tr \; \Big| &\gamma^{(k)}_{N,t} - \gamma^{(k)}_{\infty,t} \Big| \\ \leq \; &  \int_0^t \rd s_1 \dots \int_0^{s_{n-1}} \rd s_n \, \tr \; \left| \cU^{(k)} (t-s_1) B^{(k)} \cU^{(k+1)} (s_1 -s_2) \dots B^{(k+n-1)} \left(\gamma_{N,s_n}^{(k+n)} - \gamma^{(k+n)}_{\infty,s_n} \right) \right|  \\ &+ \frac{1}{N} \sum_{m=1}^N \int_0^t \rd s_1 \dots \int_0^{s_{m-1}} \rd s_m \tr \; \left| \cU^{(k)} (t-s_1) B^{(k)} \dots \cU^{(k+m-1)} (s_{m-1} -s_m) A^{(k+m-1)} \gamma^{(k+m-1)}_{N,s_m} \right| \\ &+ \sum_{m=1}^n \frac{k+m-1}{N} \int_0^t \rd s_1 \dots \int_0^{s_{m-1}} \rd s_m \,\tr \; \left| \cU^{(k)} (t-s_1) B^{(k)} \dots B^{(k+m-1)} \gamma^{(k+m)}_{N,s_m} \right|\,.
\end{split}
\end{equation}
Next we observe that, since $\cU^{(k)} (t)$ is a unitary operator,
\begin{equation}\label{eq:cU} \tr \; \left| \cU^{(k)} (t) \gamma^{(k)} \right| = \tr \; \left| \gamma^{(k)} \right| \, .\end{equation}
Moreover, since $V$ is a bounded potential, we have
\begin{equation}\label{eq:Akbd}
\tr \; \left| A^{(k)} \gamma^{(k)} \right| \leq k^2 \| V \| \tr \left| \gamma^{(k)} \right|
\end{equation}
and
\begin{equation}\label{eq:Bkbd}
\tr \; \left| B^{(k)} \gamma^{(k+1)} \right| \leq 2k \| V \| \, \tr \left| \gamma^{(k+1)} \right|
\end{equation}
where we used the fact that
\begin{equation}\label{eq:223} \tr \left| \tr_{k+1} \gamma^{(k+1)} \right| \leq \tr \left| \gamma^{(k+1)} \right| \,. \end{equation}
(Here the trace on the r.h.s. is a trace over $(k+1)$ particles.) Applying these bounds iteratively to the terms on the r.h.s. of (\ref{eq:duhbd3}), and using the a-priori information $\tr \, \left| \gamma^{(k+n)}_{N,t} \right| = \tr \, \gamma^{(k+n)}_{N,t} = 1$ (and analogously for $\gamma^{(k+n)}_{\infty,t}$), we obtain
\begin{equation}\label{eq:diffbd1}
\begin{split}
\tr \; \Big| \gamma^{(k)}_{N,t} - \gamma^{(k)}_{\infty,t} \Big| \leq \; &2 \, ( 2 \| V \| t)^n \frac{(k+n-1)!}{(k-1)! n!} + \frac{2}{N} \sum_{m=1}^{n} (2\| V \| t)^m (k+m-1) \frac{(k+m-1)!}{m! (k-1)!} \\
\leq \; &2^k \, (4 \| V \| t)^n + \frac{k \, 2^{k+1}}{N} \sum_{m=1}^N (4 \| V \| t)^m\,.
\end{split}
\end{equation}
If $0 < t \leq t_0$, with $t_0 = 1/ (8 \|V \|)$, it follows that
\begin{equation*}
\tr \; \Big| \gamma^{(k)}_{N,t} - \gamma^{(k)}_{\infty,t} \Big| \leq \frac{2^k}{2^n} + \frac{k 2^{k+1}}{N}\,.
\end{equation*}
Since the l.h.s. is independent of the order $n$ of the expansion, it follows that
\begin{equation}\label{eq:diffbd2}
\tr\; \Big| \gamma^{(k)}_{N,t} - \gamma^{(k)}_{\infty,t} \Big| \leq \frac{k 2^{k+1}}{N}
\end{equation}
and thus that
\begin{equation}\label{eq:t0}
\tr \; \Big| \gamma^{(k)}_{N,t} - \gamma^{(k)}_{\infty,t} \Big| \to 0 \qquad \text{as } N \to \infty
\end{equation}
for all $0 \leq t \leq t_0$ and for all $k \geq 1$. Next, set \[ t_1 := \sup \left\{ t > 0 : \lim_{N \to \infty} \, \tr \; \Big| \gamma^{(k)}_{N,s} - \gamma^{(k)}_{\infty,s} \Big| = 0 \quad \text{for all fixed $0 \leq s \leq t$ and $k \geq 1$} \right\} \,.\] {F}rom (\ref{eq:t0}), it follows that $t_1 \geq t_0$. We show that $t_1 = \infty$ by contradiction. Suppose that $t_1 < \infty$. Then, if $t_2 = t_1 - (t_0/2)$, we have, by definition,
\begin{equation}\label{eq:t2} \lim_{N\to \infty}  \tr\; \left|  \gamma^{(k)}_{N, t_2} - \gamma^{(k)}_{\infty, t_2} \right| = 0 \qquad \text{for all $k \geq 1$}. \end{equation} Starting from (\ref{eq:t2}), we are going to prove that
\begin{equation}\label{eq:t1} \lim_{N \to \infty} \tr \; \left| \gamma^{(k)}_{N, t} - \gamma^{(k)}_{\infty, t} \right| = 0 \end{equation} for all $k \geq 1$ and for all $0\leq t \leq t_1 + (t_0/2)$; this contradicts the definition of $t_1$. To show (\ref{eq:t1}), we expand $\gamma^{(k)}_{N, t}$ and $\gamma^{(k)}_{\infty, t}$ in Duhamel series similar to (\ref{eq:duhbd1}) and (\ref{eq:duhbd2}), but starting at time $t_2 = t_1 -(t_0/2)$. Analogously to (\ref{eq:duhbd3}), we obtain, for $t = t_2 + \tau$,
\begin{equation}\label{eq:duhbd4}
\begin{split}
\tr \; &\Big| \gamma^{(k)}_{N,t} - \gamma^{(k)}_{\infty,t} \Big| \\ \leq \; &\tr \; \left| \cU^{(k)} (\tau) \left( \gamma^{(k)}_{N, t_2} - \gamma^{(k)}_{\infty, t_2} \right) \right| \\ &+ \sum_{m=1}^{n-1} \int_0^{\tau} \rd s_1 \dots \int_0^{s_{m-1}} \rd s_m \\ &\hspace{1cm} \times \tr \; \left|  \cU^{(k)} (\tau- s_1) B^{(k)} \cU^{(k+1)} (s_1 -s_2) \dots B^{(k+m-1)} \cU^{(k+m)} (s_m)  \left(\gamma_{N,t_2}^{(k+n)} - \gamma^{(k+n)}_{\infty,t_2} \right) \right| \\ &+  \int_0^{\tau} \rd s_1 \dots \int_0^{s_{n-1}} \rd s_n \, \tr \; \left| \cU^{(k)} (\tau-s_1) B^{(k)}  \dots B^{(k+n-1)} \left(\gamma_{N,s_n}^{(k+n)} - \gamma^{(k+n)}_{\infty,s_n} \right) \right|  \\ &+ \frac{1}{N} \sum_{m=1}^N \int_0^{\tau} \rd s_1 \dots \int_0^{s_{m-1}} \rd s_m \tr \; \left| \cU^{(k)} (\tau-s_1) B^{(k)} \dots \cU^{(k+m-1)} (s_{m-1} -s_m) A^{(k+m-1)} \gamma^{(k+m-1)}_{N,s_m} \right| \\ &+ \sum_{m=1}^n \frac{k+m-1}{N} \int_0^{\tau} \rd s_1 \dots \int_0^{s_{m-1}} \rd s_m \,\tr \; \left| \cU^{(k)} (\tau-s_1) B^{(k)} \dots B^{(k+m-1)} \gamma^{(k+m)}_{N,s_m} \right|\,.
\end{split}
\end{equation}
With respect to (\ref{eq:duhbd3}), we have one more term on the r.h.s. of the last equation, due to the fact that at time $t=t_2$ the densities do not coincide (while they do at time $t=0$).
Analogously to (\ref{eq:diffbd1}) we find
\begin{equation*}
\begin{split}
\tr \; \Big| \gamma^{(k)}_{N,t} - \gamma^{(k)}_{\infty,t} \Big| \leq \; & 2^k \sum_{m=0}^{n-1} \frac{1}{2^m} \, \tr \; \left| \left(\gamma_{N,t_2}^{(k+m)} - \gamma^{(k+m)}_{\infty,t_2} \right) \right| +\frac{2^k}{2^n} + \frac{k 2^{k+1}}{N}\, ,
\end{split}
\end{equation*}
if $t_1 - (t_0/2) \leq t \leq t_1 + (t_0/2)$ (that is, if $0 \leq \tau \leq t_0$). Choosing first $n >0$ sufficiently large (to make the second term on the r.h.s. smaller than $\eps/3$), and then $N > 0$ sufficiently large (this guarantees that the third term, and, by (\ref{eq:t2}), also the first term, are smaller than $\eps /3$), the quantity on the l.h.s. can be made smaller than any $\eps >0$ (for arbitrary $k \geq 1$ and $t_1 - (t_0/2) \leq t \leq t_1 + (t_0/2)$). This shows (\ref{eq:t1}) and completes the proof of the theorem.
\end{proof}

\subsection{Another Proof of Theorem \ref{thm:bd}}
\label{sec:3steps}

{F}rom the proof of Theorem \ref{thm:bd} presented above, we notice that the expansion of the BBGKY hierarchy in (\ref{eq:duhbd1}) is much more involved than the corresponding expansion (\ref{eq:duhbd2}) of the infinite hierarchy (\ref{eq:infhierbd}). It turns out that it is possible to avoid the expansion of the BBGKY hierarchy making use of a simple compactness argument; this will be especially important when dealing with singular potentials. In the following we explain the main steps of this alternative proof to Theorem~\ref{thm:bd}. Then, in the next section, we will illustrate how to extend it to potentials with a Coulomb singularity.

\medskip

The idea, which was first presented in \cite{BGM,BGMEY,EY}, consists in characterizing the limit of the densities $\gamma^{(k)}_{N,t}$ as the unique solution to the infinite hierarchy of equations (\ref{eq:infhierbd}); combined with the compactness, this information provides a proof of Theorem~\ref{thm:bd}. More precisely, the proof is divided into three main steps. First of all, one shows the compactness of the sequence $\{ \gamma^{(k)}_{N,t} \}_{k \geq 1}$ with respect to an appropriate weak topology. Then, one proves that an arbitrary limit point $\{ \gamma^{(k)}_{\infty,t} \}_{k \geq 1}$ of the sequence $\{ \gamma^{(k)}_{N,t} \}_{k=1}^N$ is a solution to the infinite hierarchy (\ref{eq:infhierbd}) (one proves, in other words, the convergence to the infinite hierarchy). Finally, one shows the uniqueness of the solution to the infinite hierarchy (\ref{eq:infhierbd}). Since it is simple to verify that the factorized family $\{ \gamma^{(k)}_{\infty,t} \}_{k\geq 1}$, with $\gamma^{(k)}_t = |\ph_t \rangle \langle \ph_t|^{\otimes k}$ for all $k \geq 1$, is a solution to the infinite hierarchy, it follows immediately that $\gamma^{(k)}_{N,t} \to |\ph_t \rangle \langle \ph_t|^{\otimes k}$ as $N \to \infty$ (at first only in the weak topology with respect to which we have compactness; since the limit is an orthogonal rank one projection, it is however simple to check that weak convergence implies strong convergence, in the sense (\ref{eq:claimbd})). Next, we discuss these three main steps (compactness, convergence, and uniqueness) in some more details.

\medskip

{\bf Compactness:} Let $\cL^1_k \equiv \cL^1 (L^2 (\bR^{dk}))$ denote the space of trace class operators on $L^2 (\bR^{dk})$, equipped with the trace norm
\[ \| A \|_1 = \tr\, | A | = \tr \, \left( A^* A \right)^{1/2} \, \quad \text{for all } A \in \cL_k^1 \, . \] Moreover, let $\cK_k \equiv \cK (L^2 (\bR^{dk}))$ be the space of compact operators on $L^2 (\bR^{dk})$, equipped with the operator norm. Then $\cL^1_k$ and $\cK_k$ are Banach spaces and $\cL^1_k = \cK_k^*$ (see, for example, \cite{RSi}[Theorem VI.26]). By definition, the $k$-particle marginal density $\gamma^{(k)}_{N,t}$ is a non-negative operator in $\cL^1_k$, with
\[ \| \gamma^{(k)}_{N,t} \|_1 = \tr \; | \gamma_{N,t}^{(k)} | = \tr \; \gamma^{(k)}_{N,t} = 1 \] for all $N \geq k$. For fixed $t \in \bR$ and $k \geq 1$, it follows from the Banach-Alaouglu Theorem that the sequence $\{\gamma^{(k)}_{N,t} \}_{N \geq k}$ is compact with respect to the weak* topology of $\cL^1_k$.

\medskip

Since we want to identify limit points of the sequence $\gamma^{(k)}_{N,t}$ as solutions to the system of integral equations (\ref{eq:infhierbd}), compactness for fixed $t \in \bR$ is not enough. To make sure that there are subsequences of $\gamma_{N,t}^{(k)}$ which converge for all times in a certain interval, we use the fact that, since $\cK_k$ is separable, the weak* topology on the unit ball of $\cL_k^1$ is metrizable. It is possible, in other words, to introduce a metric $\eta_k$ on $\cL^1_k$ such that a uniformly bounded sequence $\{ A_n \}_{n \in \bN} \in \cL^1_k$ converges to $A \in \cL^1_k$ as $n \to \infty$ with respect to the weak* topology of $\cL^1_k$ if and only if $\eta_k (A_n,A) \to 0$ (see \cite{Ru}[Theorem~3.16], for the explicit construction of the metric $\eta_k$). For arbitrary $T >0$ let $C([0,T],\cL_1^k)$ be the space of functions of $t \in [0,T]$ with values in $\cL^1_k$ which are continuous with respect to the metric $\eta_k$; on $C([0,T],\cL^1_k)$ we can define the metric
\begin{equation}\label{eq:whetak}
\widehat \eta_k (\gamma^{(k)} (\cdot ) , \bar \gamma^{(k)} (\cdot ))
:= \sup_{t \in [0,T]} \eta_k (\gamma^{(k)} (t) , \bar \gamma^{(k)}
(t))\,.
\end{equation}
Finally, we denote by $\tau_{\text{prod}}$ the topology on the
space $\bigoplus_{k \geq 1} C([0,T], \cL^1_k)$ given by the product
of the topologies generated by the metrics $\wh \eta_k$ on $C([0,T],
\cL^1_k)$.

\medskip

The metric structure introduced on the space $\bigoplus_{k \geq 1} C([0,T], \cL^1_k)$ allows us to invoke the Arzela-Ascoli Theorem to prove the compactness of the sequence $\Gamma_{N,t} = \{ \gamma_{N,t}^{(k)} \}_{k=1}^N$. We obtain the following proposition (for the detailed proof, see, for example, \cite[Section 6]{ESY3}).
\begin{proposition}\label{prop:compactness}
Fix an arbitrary $T>0$. Then the sequence
$\Gamma_{N,t} = \{ \gamma_{N,t}^{(k)} \}_{k=1}^N \in \bigoplus_{k \geq 1} C([0,T], \cL_k^1)$ is compact with respect to the product topology $\tau_{\text{prod}}$ defined above.
For any limit point $ \Gamma_{\infty,t} = \{ \gamma_{\infty,t}^{(k)}
\}_{k \geq 1}$, $ \gamma^{(k)}_{\infty,t}$ is symmetric w.r.t. permutations, non-negative and such that   \begin{equation}\label{eq:bou} \tr \; \gamma^{(k)}_{\infty,t} \leq 1
\,\end{equation} for every $k \geq 1$.
\end{proposition}

{\it Remark.} Convergence of $\Gamma_{N,t} = \{ \gamma^{(k)}_{N,t} \}_{k= 1}^N$ to $\Gamma_{\infty,t} = \{ \gamma^{(k)}_{\infty,t} \}_{k\geq 1}$ with respect to the topology $\tau_{\text{prod}}$ is equivalent to the statement that, for every fixed $k \geq 1$, and for every fixed compact operator $J^{(k)} \in \cK_k$, \begin{equation}\label{eq:prodtop} \tr J^{(k)} \left( \gamma^{(k)}_{N,t} - \gamma^{(k)}_{\infty,t} \right) \to 0 \end{equation} as $N \to \infty$, uniformly in $t$ for $t \in [0,T]$. Compactness of $\Gamma_{N,t}$ with respect to the topology $\tau_{\text{prod}}$ means therefore that for every sequence $\{ M_j \}_{j\in \bN}$ there exists a subsequence $\{ N_j \}_{j \in \bN} \subset \{ M_j \}_{j \in \bN}$ and a limit point $\Gamma_{\infty,t}$ such that $\Gamma_{N_j,t} \to \Gamma_{\infty,t}$ in the sense (\ref{eq:prodtop}).

\bigskip

{\bf Convergence:} The second main step consists in characterizing the limit points of the (compact) sequence $\Gamma_{N,t} = \{ \gamma^{(k)}_{N,t} \}_{k \geq 1}$ as solutions to the infinite hierarchy of equations (\ref{eq:infhierbd}).
\begin{proposition}\label{prop:conv}
Suppose that $V \in L^{\infty} (\bR^d)$ such that $V(x) \to 0$ as $|x| \to \infty$. Assume moreover that $\Gamma_{\infty,t} = \{ \gamma^{(k)}_{\infty,t} \}_{k \geq 1} \in \bigoplus_{k \geq 1} C([0,T], \cL^1_k)$ is a limit point of the sequence $\Gamma_{N,t} = \{ \gamma^{(k)}_{N,t} \}_{k=1}^N$ with respect to the product topology $\tau_{\text{prod}}$. Then $\gamma^{(k)}_{\infty,0} = |\ph \rangle \langle \ph|^{\otimes k}$ and
\begin{equation}\label{eq:infhier}
\gamma^{(k)}_{\infty,t} = \cU^{(k)} (t) \gamma^{(k)}_{0,\infty} + \int_0^t \rd s \, \cU^{(k)} (t-s) B^{(k)} \gamma^{(k+1)}_{\infty,t}
\end{equation}
for all $k \geq 1$. Here $\cU^{(k)} (t)$, and $B^{(k)}$ are defined as in (\ref{eq:free}) and, respectively, in (\ref{eq:Bk}).
\end{proposition}
Note that in Proposition \ref{prop:conv} we assume the potential to vanish at infinity. This condition, which was not required in Section \ref{sec:bd}, is not essential but it simplifies the proof and it is also satisfied for the singular potentials (like the Coulomb potential) that we are going to study in the next sections.

\begin{proof}
Passing to a subsequence we can assume that $\Gamma_{N,t} \to \Gamma_{\infty,t}$ as $N \to \infty$, with respect to the product topology $\tau_{\text{prod}}$; this implies immediately that $\gamma_{\infty,0} = |\ph \rangle \langle \ph|^{\otimes k}$. To prove (\ref{eq:infhier}), on the other hand, it is enough to show that for every fixed $k \geq 1$, and for every fixed $J^{(k)}$ from a dense subset of $\cK_k$,
\begin{equation}\label{eq:Jinf}
\begin{split}
\tr \; J^{(k)} \gamma^{(k)}_{\infty,t} = \tr J^{(k)} \cU^{(k)} (t) \gamma^{(k)}_{\infty,0} + \int_0^t \rd s \, \cU^{(k)} (t-s) \tr \, J^{(k)} B^{(k)} \gamma^{(k+1)}_s \,.
\end{split}
\end{equation}
To demonstrate (\ref{eq:Jinf}), we start from the BBGKY hierarchy (\ref{eq:BBGKY}) which leads to the relations
\begin{equation}\label{eq:BBGKYint2}
\begin{split}
\tr\, J^{(k)} \gamma^{(k)}_{N,t} = \; & \tr\, J^{(k)} \cU^{(k)} (t) \gamma^{(k)}_{N,0} + \frac{1}{N} \sum_{j=1}^k \int_0^t \rd s \, \tr J^{(k)} \cU^{(k)} (t-s) \left[ V (x_i - x_j) , \gamma^{(k)}_{N,s} \right] \\ &+ \frac{N-k}{N} \int_0^t \rd s \, \tr J^{(k)} \cU^{(k)} (t-s) B^{(k)}  \gamma^{(k+1)}_{N,s}\, .
\end{split}
\end{equation}
Since, by assumption, the l.h.s. and the first term on the r.h.s. of the last equation converge, as $N \to \infty$, to the l.h.s. and, respectively, to the first term on the r.h.s. of (\ref{eq:Jinf}) (for every compact operator $J^{(k)}$), (\ref{eq:infhier}) follows if we can prove that
\begin{equation}\label{eq:conver1}
\frac{1}{N} \sum_{j=1}^k \int_0^t \rd s \, \tr J^{(k)} \cU^{(k)} (t-s) \left[ V (x_i - x_j) , \gamma^{(k)}_{N,s} \right] \to 0
\end{equation}
and that
\begin{equation}\label{eq:conver2}
\frac{N-k}{N} \int_0^t \rd s \, \tr J^{(k)} \cU^{(k)} (t-s) B^{(k)}  \gamma^{(k+1)}_{N,s} \to \int_0^t \rd s \tr J^{(k)} \cU^{(k)} (t-s) B^{(k)} \gamma^{(k+1)}_{\infty,s}
\end{equation}
as $N \to \infty$. Eq. (\ref{eq:conver1}) follows because \[ \left| \tr J^{(k)} \cU^{(k)} (t-s) \left[ V (x_i -x_j), \gamma^{(k)}_{N,s} \right] \right| \leq 2 \| J^{(k)} \| \| V \| \tr \, \left| \gamma^{(k)}_{N,s} \right| \leq 2 \| J^{(k)} \| \| V \| \] is finite, uniformly in $N$. To prove Eq. (\ref{eq:conver2}) one can use a similar argument, combined with the observation that
\begin{equation*}
\begin{split}
\tr  \; &J^{(k)} \cU^{(k)} (t-s) B^{(k)} \left( \gamma^{(k+1)}_{N,s} - \gamma^{(k+1)}_{\infty,s} \right) \\ &= k \tr\;\left[ \left( \cU^{(k)} (s-t) J^{(k)} \right) V (x_1 - x_{k+1}) - V (x_1 -x_{k+1}) \left( \cU^{(k)} (s-t) J^{(k)} \right) \right] \, \left( \gamma^{(k+1)}_{N,s} - \gamma^{(k+1)}_{\infty,s} \right)
\\ &\to 0
\end{split}
\end{equation*}
as $N \to \infty$. This does not follow directly from the assumption that $\Gamma_{N,t} \to \Gamma_{\infty,t}$ with respect to the topology $\tau_{\text{prod}}$ because the operators $(\cU^{(k)} (s-t) J^{(k)}) V (x_1 -x_{k+1})$ and $V (x_1 -x_{k+1})(\cU^{(k)} (s-t) J^{(k)})$ are not compact on $L^2 (\bR^{d(k+1)})$. Instead we have to apply an approximation argument, cutting off high momenta in the $x_{k+1}$-variable, and using the fact that, by energy conservation, $\tr \,  \nabla_{k+1}^* \gamma_{N,t}^{(k+1)} \nabla_{k+1}$ is bounded, uniformly in $N$ and in $t$ (and that, therefore, $\tr \, \nabla_{k+1}^* \gamma_{\infty,t}^{(k+1)} \nabla_{k+1}$ is bounded as well). Note that, because of the assumption that $V(x) \to 0$ as $|x| \to \infty$, we only need a cutoff in momentum, and no cutoff in position space is necessary. The details of this approximation argument can be found, for example, in Eq. (7.35) and Eq. (7.36) in the proof of Theorem 7.1 in \cite{ESY3} (after replacing $\delta_{\beta}$ through the bounded potential $V$).
\end{proof}

\bigskip

{\bf Uniqueness:} to conclude the proof of Theorem \ref{thm:bd}, we still have to prove the uniqueness of the solution to the infinite hierarchy (\ref{eq:infhier}).
\begin{proposition}\label{prop:unique}
Fix $\Gamma_{\infty,0} = \{ \gamma^{(k)}_{\infty,0} \}_{k \geq 1} \in \bigoplus_{k \geq 1} \cL^1_k$. Then there exists at most one solution $\Gamma_{\infty,t} = \{ \gamma^{(k)}_{\infty,t} \}_{k \geq 1} \in \bigoplus_{k \geq 1} C([0,T], \cL^k_1)$ to the infinite hierarchy (\ref{eq:infhier}) such that $\gamma^{(k)}_{\infty,t=0} = \gamma^{(k)}_{\infty,0}$ and $\tr \, | \gamma^{(k)}_{\infty,t} | \leq 1$ for all $k \geq 1$ and all $t \in [0,T]$.
\end{proposition}
\begin{proof}
Suppose that $\{ \gamma^{(k)}_{\infty,1,t} \}_{k\geq 1}$ and $\{ \gamma^{(k)}_{\infty,2,t}\}_{k\geq 1}$ are two solutions of (\ref{eq:infhier}) with the same initial data $\{ \gamma^{(k)}_{\infty,0} \}_{k \geq 1}$, such that $\tr \, | \gamma^{(k)}_{\infty,i,t}| \leq 1$, for all $k \geq 1$, $t \in [0,T]$, and for $i =1,2$. Then we can expand $\gamma^{(k)}_{\infty,1,t}$ and $\gamma^{(k)}_{\infty,2,t}$ in the Duhamel series (\ref{eq:duhbd3}). It follows that
\begin{equation*}
\begin{split}
\tr \; \left| \gamma^{(k)}_{\infty,1,t} - \gamma^{(k)}_{\infty,2,t} \right| \leq \int_0^t \rd s_1 \dots \int_0^{s_{n-1}} \rd s_n \, \tr \; \left| \cU^{(k)} (t-s_1) B^{(k)} \dots B^{(k+n-1)} \left( \gamma_{\infty,1,s_n}^{(k+n)} - \gamma_{\infty,2,s_n}^{(k+n)} \right) \right|\,.
\end{split}
\end{equation*}
Applying recursively the bounds (\ref{eq:cU}) and (\ref{eq:Bkbd}), we obtain
\begin{equation*}
\tr \; \left| \gamma^{(k)}_{\infty,1,t} - \gamma^{(k)}_{\infty,2,t} \right| \leq \frac{(k+n-1)!}{(k-1)! n!} (2 \| V \| t )^n \leq 2^k \, (4 \| V \| t )^n \end{equation*}
and thus, for $0< t < 1/ 8 \|V \|$,
\[ \tr  \; \left| \gamma^{(k)}_{\infty,1,t} - \gamma^{(k)}_{\infty,2,t} \right| \leq 2^{k-n} \,. \]
Since the l.h.s. is independent of $n \geq 1$, it has to vanish. This proves uniqueness for short time. Iterating the same argument, we obtain uniqueness for all times.
\end{proof}

\subsection{Derivation of the Hartree Equation for a Coulomb Potential}
\label{sec:cou}

The arguments presented in Section \ref{sec:bd} and in Section \ref{sec:3steps} required the interaction potential $V$ to be bounded. Unfortunately, several systems of physical interest are described by unbounded potential. For example, in a non-relativistic approximation, a system of gravitating bosons (a boson star) can be described by the Hamiltonian
\begin{equation}\label{eq:HNcou}
H_N = \sum_{j=1}^N -\Delta_j - \frac{\lambda}{N} \sum_{i<j}^N \frac{1}{|x_i -x_j|}\,
\end{equation}
with a singular Coulomb interaction among the particles. The factor of $1/N$ in front of the potential energy can be justified, when describing gravitating particles, by the smallness of the gravitational constant. As in the case of bounded potential, we are interested in the dynamics generated by the Hamiltonian (\ref{eq:HNcou}) on factorized initial $N$-particle wave functions. We specialize here in the physically most interesting case of particles moving in three dimensions; however, the theorem remains valid in all dimensions $d \geq 2$.
\begin{theorem}[Erd\H os-Yau, \cite{EY}] \label{thm:cou}
Let $\psi_N = \ph^{\otimes N}$ for some $\ph \in H^1 (\bR^3)$ and let $\psi_{N,t} = e^{-iH_N t} \psi_N$ where the Hamiltonian $H_N$ is defined as in (\ref{eq:HNcou}). Then, for arbitrary $k \geq 1$ and $t \in \bR$, we have
\begin{equation}\label{eq:claim-cou} \tr \; \left| \gamma^{(k)}_{N,t} - |\ph_t \rangle \langle \ph_t |^{\otimes k} \right| \to 0 \end{equation} as $N \to \infty$. Here $\ph_t$ is the solution to the nonlinear Hartree equation
\[ i\partial_t \ph_t = -\Delta \ph_t - \lambda \left( \frac{1}{| \cdot |} * |\ph_t|^2 \right) \ph_t \] with initial data $\ph_{t=0} = \ph$.
\end{theorem}
{\it Remark.} Although, physically, the value of the constant $\lambda$ is positive (corresponding to the Coulomb attraction among gravitating particles), the theorem remains valid also for negative values of $\lambda$ (corresponding to repulsive Coulomb interaction).

\bigskip

The general strategy used in \cite{EY} to prove Theorem \ref{thm:cou} is the same as the one outlined in Section~\ref{sec:3steps}. First one proves the compactness of the sequence of marginal $\{ \gamma^{(k)}_{N,t} \}_{k =1}^N$ with respect to an appropriate weak topology (the product topology $\tau_{\text{prod}}$ introduced after (\ref{eq:whetak})), then one shows that an arbitrary limit point $\{ \gamma^{(k)}_{\infty,t} \}_{k \geq 1}$ of the sequence $\{ \gamma^{(k)}_{N,t} \}_{k=1}^N$ is a solution to the infinite hierarchy of equations
\begin{equation}\label{eq:infhier-cou}
\gamma^{(k)}_t = \cU^{(k)} (t) \gamma^{(k)} + \int_0^t \rd s \, \cU^{(k)} (t-s) B^{(k)} \gamma^{(k+1)}_s
\end{equation}
where $\cU^{(k)}$ is the free evolution defined in (\ref{eq:free}), and the collision map $B^{(k)}$ is now given by
\begin{equation}\label{eq:Bk-cou} B^{(k)} \gamma^{(k+1)} = -i\lambda \sum_{j=1}^k \tr_{k+1} \, \left[ \frac{1}{|x_j - x_{k+1}|} , \gamma^{(k+1)} \right] \,. \end{equation}
Finally, one proves the uniqueness of the solution to (\ref{eq:infhier-cou}). Although the proof of the compactness and of the convergence also require several changes with respect to what we discussed in Section \ref{sec:3steps}, the main difficulty one has to face when the bounded potential is replaced by the Coulomb interaction is the proof of the uniqueness of the solution to the infinite hierarchy. The key idea introduced by Erd\H os and Yau in \cite{EY} was to restrict the class of densities for which uniqueness must be proven. In Proposition \ref{prop:unique}, uniqueness is proved in the class of densities with $\tr \, | \gamma^{(k)}_t| \leq 1$ for all $k \geq 1$, and all $t \in [0,T]$ (but the same argument works under the weaker assumption $\tr \, | \gamma^{(k)}_t | \leq C^k$, for some constant $C < \infty$). Following \cite{EY}, in the case of a Coulomb potential we are only going to show the uniqueness of (\ref{eq:infhier-cou}) in the class of densities $\Gamma_t = \{ \gamma^{(k)}_t \}_{k \geq 1}$ satisfying the a-priori bound
\begin{equation} \label{eq:aprik1}
\tr \, \left| (1-\Delta_1)^{1/2} \dots (1-\Delta_k)^{1/2} \gamma^{(k)}_t  (1-\Delta_k)^{1/2} \dots (1-\Delta_1)^{1/2} \right| \leq C^k
\end{equation}
for all $k \geq 1$ and for all $t \in [0,T]$. Note that, for non-negative densities $\gamma^{(k)}_t \geq 0$ (in the sense of operators, that is, in the sense that $\langle \psi^{(k)} , \gamma^{(k)}_t \psi^{(k)} \rangle \geq 0$ for all $\psi^{(k)} \in L^2 (\bR^{3k})$) we have
\begin{equation*}
\tr \, \left| (1-\Delta_1)^{1/2} \dots (1-\Delta_k)^{1/2} \gamma^{(k)}_t  (1-\Delta_k)^{1/2} \dots (1-\Delta_1)^{1/2} \right| = \tr \, (1-\Delta_1) \dots (1-\Delta_k) \gamma^{(k)}_t \, .
\end{equation*}

\medskip

There is, of course, a price to pay in order to restrict the proof of the uniqueness to this class of densities. In fact, to apply this uniqueness result to the proof of Theorem \ref{thm:cou}, one has to show that an arbitrary limit point $\Gamma_{\infty,t} = \{ \gamma^{(k)}_{\infty,t} \}_{k \geq 1}$ of the sequence of densities $\Gamma_{N,t} = \{ \gamma^{(k)}_{N,t} \}_{k=1}^N$ associated with $\psi_{N,t}$ satisfies the a-priori bound (\ref{eq:aprik1}). Due to the Coulomb singularity, this is actually not so simple and requires an additional approximation argument.

\bigskip

{\bf Approximation of the Coulomb singularity:} For a fixed $\eps >0$ we define the regularized Hamiltonian
\begin{equation}\label{eq:wtHN}
\wt H_N = \sum_{j=1}^N -\Delta_j - \frac{\lambda}{N} \sum_{i<j}^N \frac{1}{|x_i -x_j|+ \eps N^{-1}}.
\end{equation}
Moreover, for a fixed sufficiently small $\delta >0$, we introduce the regularized initial data \begin{equation}\label{eq:reginit} \wt \psi_N = \frac{\chi (\delta \wt H_N / N) \psi_N }{\| \chi (\delta \wt H_N/N) \psi_N \|} \qquad (\text{recall that } \psi_N = \ph^{\otimes N} )   \end{equation} where $\chi \in C_0^{\infty} (\bR)$ is a monotone decreasing function such that $\chi (s) = 1$ for all $s \leq 1$ and $\chi (s) = 0$ for all $s \geq 2$. We consider then the regularized evolution of the regularized initial wave function \[ \wt \psi_{N,t} = e^{-i \wt H_N t} \wt \psi_N \, . \] The advantage of working with the regularized wave function $\wt \psi_{N,t}$ instead of $\psi_{N,t}$ is that it satisfies the following strong a-priori bounds.
\begin{proposition}\label{prop:apricou}
Let $\wt \psi_{N,t} = e^{-i \wt H_N t} \wt \psi_N$, for some fixed $\eps,\delta >0$. Then there exists a constant $C > 0$ (depending on $\eps, \delta$) and, for all $k \geq 1$, there exists $N_0 = N_0 (k) > k$ such that
\begin{equation}\label{eq:apricou0} \langle \wt \psi_{N,t}, (1-\Delta_1) \dots (1-\Delta_k) \, \wt \psi_{N,t} \rangle \leq C^k \,  \end{equation}
for all $N \geq N_0$.
\end{proposition}
{\it Remark.} Expressed in terms of the $k$-particle marginal $\wt \gamma_{N,t}^{(k)}$ associated with $\wt \psi_{N,t}$, the bound (\ref{eq:apricou0}) reads
\begin{equation}\label{eq:apricou} \tr \; (1-\Delta_1) \dots (1-\Delta_k) \, \wt \gamma^{(k)}_{N,t} \leq C^k \, .\end{equation}

\medskip

We will show Proposition \ref{prop:apricou} below, making use of Proposition \ref{prop:enest}; we will see there that the regularization of the Coulomb singularity and of the initial wave function both play an important role. For the solution $\psi_{N,t}$ of the original Schr\"odinger equation with the original factorized initial data $\psi_N = \ph^{\otimes N}$ it is not known whether bounds like (\ref{eq:apricou0}) hold true.

\medskip

In order for the regularized wave function $\wt \psi_{N,t}$ to be useful, one needs to prove that it approximates, in an appropriate sense, the wave function $\psi_{N,t}$. To compare the two $N$-particle wave function, we introduce a third wave function $\wh \psi_{N,t} = e^{-i \wt H_N t} \psi_N$, and we use the triangle inequality
\begin{equation}\label{eq:trian} \| \psi_{N,t} - \wt \psi_{N,t} \| \leq \| \psi_{N,t} - \wh\psi_{N,t} \| + \| \wh \psi_{N,t} - \wt \psi_{N,t} \| \,. \end{equation}
The second term is actually independent of time because of the unitarity of the evolution. Using the definition of the regularized initial data $\wt \psi_N$, one can prove that
\[ \| \wh \psi_{N,t} - \wt \psi_{N,t} \| = \| \psi_N - \wt \psi_N \| \leq C \delta^{1/2} \] uniformly in $N$. To control the first term on the r.h.s. of (\ref{eq:trian}), we observe that
\begin{equation*}
\frac{\rd}{\rd t} \, \left\| \psi_{N,t} - \wh \psi_{N,t} \right\|^2 = 2 \text{Im} \, \left\langle \left( H_N - \wt H_N \right) \wh \psi_{N,t} , \psi_{N,t} - \wh \psi_{N,t} \right\rangle
\end{equation*}
and thus that
\begin{equation}\label{eq:gronwa}
\left| \frac{\rd}{\rd t} \, \left\| \psi_{N,t} - \wh \psi_{N,t} \right\|^2 \right| \leq 2 \left\| \left( H_N - \wt H_N \right) \wh \psi_{N,t} \right\| \, \left\| \psi_{N,t} - \wh \psi_{N,t} \right\|\,.
\end{equation}
We have
\begin{equation*}
\| ( H_N - \wt H_N ) \wh \psi_{N,t} \| =  \left\| \frac{\eps}{N^2} \sum_{i<j}^N  \frac{1}{|x_i -x_j| \left( |x_i -x_j| + \eps N^{-1}\right)} \, \wh \psi_{N,t} \right\| \,.
\end{equation*}
Using the permutation symmetry of $\wh \psi_{N,t}$, it follows that
\begin{equation*}
\begin{split}
\| ( H_N - &\wt H_N ) \wh \psi_{N,t} \|^2  \\
\leq &\; \eps^2 \Big\langle \wh \psi_{N,t}, \frac{1}{|x_1 -x_2| \left( |x_1 -x_2|+\eps N^{-1} \right)} \, \frac{1}{|x_3 -x_4| \left( |x_3 -x_4|+ \eps N^{-1} \right)} \, \wh \psi_{N,t} \Big\rangle \\ &+ \eps^2 N^{-1} \Big\langle \wh \psi_{N,t}, \frac{1}{|x_1 -x_2| \left( |x_1 -x_2|+\eps N^{-1} \right)} \, \frac{1}{|x_2 -x_3| \left( |x_2 -x_3|+ \eps N^{-1} \right)} \, \wh \psi_{N,t} \Big\rangle \\ & + \eps^2 N^{-2} \Big\langle \wh \psi_{N,t}, \frac{1}{|x_1 -x_2|^2 \left( |x_1 -x_2|+\eps N^{-1} \right)^2} \, \wh \psi_{N,t} \Big\rangle
\end{split}
\end{equation*}
and thus
\begin{equation*}
\begin{split}
\| ( H_N - \wt H_N ) \wh \psi_{N,t} \|^2
\leq & \;  \eps^2 \Big\langle \wh \psi_{N,t}, \frac{1}{|x_1 -x_2|^2 |x_3 -x_4|^2} \, \wh \psi_{N,t} \Big\rangle \\ &+ \eps^2 N^{-1} \Big\langle \wh \psi_{N,t}, \frac{1}{|x_1 -x_2|^2 |x_2 -x_3|^2} \, \wh \psi_{N,t} \Big\rangle \\ &+ \eps^{1/2} N^{-1/2} \Big\langle \wh \psi_{N,t}, \frac{1}{|x_1 -x_2|^{5/2}} \, \wh \psi_{N,t} \Big\rangle\,.
\end{split}
\end{equation*}
Applying Hardy inequalities in the form
\begin{equation}\label{eq:hardy} \frac{1}{|x_i -x_j|^{\alpha}} \leq C (1-\Delta_i)^{\beta/2} (1-\Delta_j)^{\gamma/2} \end{equation}
for every $0 \leq \alpha < 3$, if $\beta+ \gamma \geq \alpha$ (see \cite{ES}[Lemma 9.1] for a proof) we find that
\begin{equation*}
\begin{split}
\| ( H_N - \wt H_N ) \wh \psi_{N,t} \|^2  \leq \; & \eps^2 (1  + N^{-1}) \langle \wh\psi_{N,t}, (1-\Delta_1) (1-\Delta_3) \wh \psi_{N,t} \rangle \\ &+ \eps^{1/2} N^{-1/2} \langle \wh \psi_{N,t}, (1-\Delta_1) ( 1-\Delta_2) \wh \psi_{N,t} \rangle
\end{split}
\end{equation*}
and thus, from (\ref{eq:apricou}),
\[ \| ( H_N - \wt H_N ) \wh \psi_{N,t} \|  \leq   C \eps^{1/4} \]
uniformly in $N$. {F}rom (\ref{eq:gronwa}), applying Gronwall's Lemma, it follows that
\[ \| \psi_{N,t} - \wh \psi_{N,t} \| \leq  C \eps^{1/4} t \, .  \]
{F}rom (\ref{eq:trian}), we obtain that
\[  \| \psi_{N,t} - \wt \psi_{N,t} \| \leq  C \left( \eps^{1/4} t + \delta^{1/2} \right) \]
and thus
\begin{equation}\label{eq:compgamma} \tr \; \left| \gamma^{(k)}_{N,t} - \wt \gamma^{(k)}_{N,t} \right| \leq C \left( \eps^{1/4} t +\delta^{1/2} \right) \end{equation}
for every $k \in \bN$, uniformly in $N \geq k$. Because of (\ref{eq:compgamma}), it suffices to prove (\ref{eq:claim-cou}) with $\gamma^{(k)}_{N,t}$ (the $k$-particle marginal associated with $\psi_{N,t}$) replaced by $\wt \gamma^{(k)}_{N,t}$ (the $k$-particle marginal associated with the regularized wave function $\wt \psi_{N,t}$) for fixed $\eps,\delta >0$; at the end (\ref{eq:claim-cou}) follows by letting $\eps,\delta \to 0$. Note that in \cite{EY} a slightly different approximation of the initial data was used; the details of the approximation presented above can be found (for a different model) in \cite{ESY2}[Section 5].

\bigskip

{\bf Energy estimates:} To prove the a-priori bounds of Proposition \ref{prop:apricou}, one can use so called energy estimates; these are estimates that compare the expectation of high powers of the Hamiltonian with corresponding powers of the kinetic energy.
\begin{proposition}\label{prop:enest}
Suppose that $\wt H_N$ is defined as in (\ref{eq:wtHN}) with $\lambda >0$ (the case $\lambda <0$ is simpler). Then there exist constants $C_1 >1$ and $C_2 >0$ and, for every $k \geq 1$, there exists an $N_0 = N_0 (k) \in \bN$ such that
\begin{equation}\label{eq:enest}
\langle \psi_N , ( \wt H_N + C_1 N )^k \, \psi_N \rangle \geq C_2^k N^k \, \langle \psi_N , (- \Delta_1 + C_1) \dots (- \Delta_k + C_1) \psi_N \rangle
\end{equation}
for every $\psi_N \in L^2_s (\bR^{3N})$ (symmetric with respect to permutations).
\end{proposition}
\begin{proof}
Using the operator inequality \[ \frac{1}{|x_i -x_j|} \leq \frac{\pi}{4} \, |\nabla_j| \, \] we can find a constant $C_1 > 1$ (depending on the coupling constant $\lambda$) such that \begin{equation}\label{eq:C1} \frac{\lambda}{|x_i -x_j|} \leq \frac{1}{2} \, (- \Delta_j + C_1) = \frac{1}{2} S_j^2 , \end{equation}
where we defined $S_j = (-\Delta_j + C_1)^{1/2}$. Note also that, for every $0<\alpha<3$ there exists a constant $C_{\alpha} <\infty$ such that \begin{equation}\label{eq:alpha1}
\frac{1}{|x_i -x_j|^{\alpha}} \leq C_{\alpha} S_j^{\alpha} \,.
\end{equation}

\medskip

We are going to prove (\ref{eq:enest}) for $C_1$ fixed as in (\ref{eq:C1}), and for an arbitrary $0< C_2 < 1/2$. The proof is by a two-step induction over $k \geq 0$. For $k = 0$, the claim is trivial. For $k =1$, it follows from (\ref{eq:C1}) because, as an operator inequality on the permutation symmetric space $L^2_s (\bR^{3N})$, we have
\begin{equation}\label{eq:k1}
\begin{split}
(\wt H_N + C_1 N)  \geq \; & N S_1^2 - \frac{N}{2} \,  \frac{\lambda}{|x_1 - x_2|}  \geq C_2 N S_1^2 \, .
\end{split}
\end{equation}
Next we assume that (\ref{eq:enest}) holds for all $k \leq n$ and we prove it for $k = n+2$, for an arbitrary $n \in \bN$. To this end, we observe that, because of the induction assumption,
\begin{equation*}
\begin{split}
(\wt H_N + C_1 N)^{n+2} &= (\wt H_N + C_1 N) (\wt H_N + C_1 N)^n (\wt H_N +C_1 N) \\ &\geq C_2^n N^n   (\wt H_N + C_1 N) S_1^2 \dots S_n^2 (\wt H_N + C_1 N) ,
\end{split}
\end{equation*}
for all $N \geq N_0 (n)$. Writing
\[ (\wt H_N + C_1 N) = \sum_{j \geq n+1} S_j^2 + h_N, \quad \text{with } \quad h_N = \sum_{j=1}^n S_j^2 - \frac{\lambda}{N} \sum_{i<j}^N \frac{1}{|x_i -x_j|+\eps N^{-1}} \] it follows that
\begin{equation}\label{eq:enest0}
\begin{split}
(\wt H_N + C_1 N)^{n+2} \geq \; &C_2^n N^n (N-n)(N-n-1) S_1^2 \dots S_{n+2}^2 \\ &+ C_2^n N^n (N-n) S_1^4 S_2^2 \dots S_{n+1}^2 \\ &+ C_2^n N^n (N-n) \left( S_1^2 \dots S_{n+1}^2 \, h_N + \text{h.c.} \right) \,.
\end{split}
\end{equation}
The first two terms are positive. As for the third term, by the definition of $h_N$,
we find that
\begin{equation}\label{eq:enest1}
\begin{split}
\left( S_1^2 \dots S_{n+1}^2 \, h_N  + \text{h.c.} \right) \geq &\; -\frac{(N-n)(N-n-1)}{2N} \,  \left( S_1^2 \dots S_{n+1}^2 \, \frac{\lambda}{|x_{n+2} - x_{n+3}|+\eps N^{-1}} +\text{h.c.} \right) \\ &-\frac{(N-n) n}{N} \left( S_1^2 \dots S_{n+1}^2 \, \frac{\lambda}{|x_{1} - x_{n+2}|+\eps N^{-1}} + \text{h.c.} \right) \\ & - \frac{n(n-1)}{2N} \left( S_1^2 \dots S_{n+1}^2 \frac{\lambda}{|x_1 -x_2|+\eps N^{-1}} + \text{h.c.} \right)\,.
\end{split}
\end{equation}
The first term on the r.h.s. of (\ref{eq:enest1}) can be bounded by
\begin{equation}\label{eq:enest2}
\begin{split}
\left( S_1^2 \dots S_{n+1}^2 \, \frac{\lambda}{|x_{n+2} - x_{n+3}| + \eps N^{-1}} + \text{h.c.} \right) & \leq 2 S_1 \dots S_{n+1} \frac{\lambda}{|x_{n+2} - x_{n +3}|} S_{n+1} \dots S_1 \\ &\leq S_1^2 \dots S_{n+2}^2\,.
\end{split}\end{equation}
As for the second term on the r.h.s. of (\ref{eq:enest1}), we remark that
\begin{equation*}
\begin{split}
\Big( S_1^2 \dots S_{n+1}^2 \, &\frac{\lambda}{|x_{1} - x_{n+2}|+\eps N^{-1}} + \text{h.c} \Big) \\ = \; & S_{n+1} \dots S_{2} \,\left( (-\Delta_1 +C_1) \frac{\lambda}{|x_{1} - x_{n+2}|+\eps N^{-1}} + \text{h.c.} \right) S_2 \dots S_{n+1} \\ \leq \; &  2C_1 \,  S_{n+1} \dots S_2 \frac{\lambda}{|x_{1} - x_{n+2}|+\eps N^{-1}} S_2 \dots S_{n+1} \\ &+ 2 S_{n+1} \dots S_{2} \, \nabla_1^* \frac{\lambda}{|x_{1} - x_{n+2}|+\eps N^{-1}} \nabla_1 S_2 \dots S_{n+1} \\ &+ \lambda \, S_{n+1} \dots S_{2} \, \left( \nabla_1^* \frac{(x_1 -x_{n+2})}{|x_1 -x_{n+2}| \left( |x_{1} - x_{n+2}|+\eps N^{-1}\right)^2} + \text{h.c.} \right) \, S_2 \dots S_{n+1}\,.
\end{split}\end{equation*}
Applying a Schwarz inequality in the last term, we conclude that there exists a constant $D>0$ (depending on $\lambda$) such that
\begin{equation}\label{eq:enest3}
\Big( S_1^2 \dots S_{n+1}^2 \,\frac{\lambda}{|x_{1} - x_{n+2}|+\eps N^{-1}} + \text{h.c.} \Big)
\leq  D \, S_1^2 \dots S_{n+2}^2 \,.
\end{equation}
Similarly, using the operator inequalities (\ref{eq:alpha1}), the last term on the r.h.s. of (\ref{eq:enest1}) can be bounded by
\begin{equation}\label{eq:enest4}
\Big( S_1^2 \dots S_{n+1}^2 \, \frac{1}{|x_{1} - x_{2}|+\eps N^{-1}} + \text{h.c.} \Big) \leq  D  \eps^{-1} N  \, S_1^2 \dots S^2_{n+1}  +  D S_1^4 S_2^2 \dots S_{n+1}^2
\end{equation}
for all $0< \eps <1$ and for a constant $D$ depending only on $\lambda$ (it is at this point that the condition $\eps>0$ is needed).
Inserting (\ref{eq:enest2}), (\ref{eq:enest3}), and (\ref{eq:enest4}) in the r.h.s. of (\ref{eq:enest1}), and the resulting bound in the r.h.s. of (\ref{eq:enest0}), we obtain that there exists $N_0 > 0$ (depending on $n$) such that
\begin{equation*}
(\wt H_N + C_1 N)^{n+2} \geq C_2^{n+2} S_1^2 \dots S_{n+2}^2
\end{equation*}
for all $N > N_0$. Note that the value of $N_0$ also depends on the parameter $\eps >0$.
\end{proof}
Using the result of Proposition \ref{prop:enest}, it is simple to complete the proof of the a-priori bounds for $\wt \psi_{N,t} = e^{-i\wt H_N t} \wt \psi_N$ (recall the definition of the regularized initial data $\wt \psi_N$ in (\ref{eq:reginit})).

\begin{proof}[Proof of Proposition \ref{prop:apricou}]
{F}rom (\ref{eq:enest}), and since $C_1 >1$, we have
\begin{equation*}
\begin{split}
\langle \wt \psi_{N,t} , (1-\Delta_1) \dots (1-\Delta_k) \wt \psi_{N,t} \rangle &\leq \langle \wt \psi_{N,t} , (C_1-\Delta_1) \dots (C_1-\Delta_k) \wt \psi_{N,t} \rangle \\ &\leq \frac{1}{C_2^k N^k} \langle \wt \psi_{N,t} , (\wt H_N + C_1 N)^k \wt \psi_{N,t} \rangle \\ &= \frac{1}{C_2^k N^k} \langle \wt \psi_{N} , (\wt H_N + C_1 N)^k \wt \psi_{N} \rangle
\end{split}
\end{equation*}
where in the last line we used the fact that the expectation of any power of $\wt H_N$ is preserved by the time-evolution. {F}rom the definition (\ref{eq:reginit}) of $\wt\psi_N$, we immediately obtain (\ref{eq:apricou0}).
\end{proof}

Since the a-priori bounds for $\wt\gamma_{N,t}^{(k)}$ obtained in Proposition \ref{prop:apricou} hold uniformly in $N$, they can also be used to derive a-priori bounds on the limit points $\{ \gamma^{(k)}_{\infty,t} \}_{k \geq 1}$ of the sequence $\{ \gamma^{(k)}_{N,t} \}_{k=1}^N$.
\begin{corollary}\label{cor:apriinfty}
Suppose that $\Gamma_{\infty,t} = \{ \gamma^{(k)}_{\infty,t} \}_{k\geq 1} \in \bigoplus_{k\geq 1} C([0,T], \cL^1_k)$ is a limit point of the sequence $\wt \Gamma_{N,t} = \{ \wt \gamma^{(k)}_{N,t} \}_{k=1}^N$ with respect to the product topology $\tau_{\text{prod}}$ defined after (\ref{eq:whetak}). Then $\gamma^{(k)}_{\infty,t} \geq 0$ and there exists a constant $C$ such that
\begin{equation}\label{eq:apricouinf}
\tr \; (1-\Delta_1) \dots (1-\Delta_k) \gamma^{(k)}_{\infty,t} \leq C^k \,
\end{equation}
for all $k \geq 1$.
\end{corollary}

\bigskip

{\bf Uniqueness:} The bounds of Corollary \ref{cor:apriinfty} are crucial; from (\ref{eq:apricouinf}) it follows that it is enough to show the uniqueness of the infinite hierarchy (\ref{eq:infhier}) in the class of densities satisfying (\ref{eq:apricouinf}), a much simpler task than proving uniqueness for all densities with $\tr \; |\gamma^{(k)}_t| \leq C^k$.
\begin{theorem}\label{thm:unique-cou}
Fix $\{ \gamma^{(k)} \}_{k\geq 1} \in \bigoplus_{k \geq 1} \cL^1_k$. Then there exists at most one solution $\{ \gamma^{(k)}_t \}_{k\geq 1} \in \bigoplus_{k\geq 1} C ([0,T] , \cL^1_k)$ to the infinite hierarchy (\ref{eq:infhier-cou}),
such that \begin{equation} \label{eq:aprik-cou}  \tr\; \left| (1-\Delta_1)^{1/2} \dots (1-\Delta_k)^{1/2} \gamma^{(k)} (1-\Delta_k)^{1/2} \dots (1-\Delta_1)^{1/2} \right| \leq C^k \end{equation} for all $k \geq 1$, and all $t \in [0,T]$.
\end{theorem}

\begin{proof}
We define the norm
\[ \| \gamma^{(k)} \|_{\cH_k} = \tr\; \left| (1-\Delta_1)^{1/2} \dots (1-\Delta_k)^{1/2} \gamma^{(k)} (1-\Delta_k)^{1/2} \dots (1-\Delta_1)^{1/2} \right| \]
and we observe that there exists a constant $C >0$ with (recall the definition (\ref{eq:Bk-cou}) for the collision map $B^{(k)}$)
\begin{equation}\label{eq:Bcou}
\| B^{(k)} \gamma^{(k+1)} \|_{\cH_k} \leq C k \| \gamma^{(k+1)} \|_{\cH_{k+1}}\,.
\end{equation}
To prove (\ref{eq:Bcou}), we write
\begin{equation*}\begin{split} \| B^{(k)} \gamma^{(k+1)} \|_{\cH_k} \leq \; &\sum_{j=1}^k \tr \, \left| S_1 \dots S_k \left( \tr_{k+1} \, \frac{1}{|x_j -x_{k+1}|}  \gamma^{(k+1)} \right) S_k \dots S_1 \right| \\ &+\sum_{j=1}^k \tr \, \left| S_1 \dots S_k \left( \tr_{k+1} \, \gamma^{(k+1)} \frac{1}{|x_j -x_{k+1}|}\right) S_k \dots S_1 \right|\,.
\end{split}
\end{equation*}
All terms can be handled similarly. We show how to bound the summand with $j=1$ on the first line.
\begin{equation*}
\begin{split}
\tr \, \Big| S_1 \dots S_k & \left( \tr_{k+1} \, \frac{1}{|x_1 -x_{k+1}|}  \gamma^{(k+1)} \right) S_k \dots S_1 \Big| \\ = \; & \tr \, \left| S_1 \dots S_k \left( \tr_{k+1} S_{k+1}^{-1} \, \frac{1}{|x_1 -x_{k+1}|} S_{k+1}^{-1} S_{k+1} \gamma^{(k+1)} S_{k+1} \right) S_k \dots S_1 \right|  \\
\leq \; & \tr \, \left| S_1 S_{k+1}^{-1} \frac{1}{|x_j -x_{k+1}|} S_{k+1}^{-1} S_1^{-1} S_1 \dots S_k  S_{k+1} \gamma^{(k+1)} S_{k+1} S_k \dots S_1 \right|  \\
\leq \; &  \| S_1 S_{k+1}^{-1} \frac{1}{|x_1 -x_{k+1}|} S_{k+1}^{-1} S_1^{-1} \| \, \| \gamma^{(k+1)} \|_{\cH_{k+1}} \\ \leq \; &C \| \gamma^{(k+1)} \|_{\cH_{k+1}}
\end{split}
\end{equation*}
where in the second line we used the cyclicity of the partial trace, in the third line we used (\ref{eq:223}) and, in the last line, we used the bound  \begin{equation}\label{eq:star} \| S_1 S_{k+1}^{-1} \frac{1}{|x_1 -x_{k+1}|} S_{k+1}^{-1} S_1^{-1} \| < \infty \, .\end{equation} To prove (\ref{eq:star}) we write, assuming for example that $k=1$,
\begin{equation*}
\begin{split}
S_1 S_2^{-1} \frac{1}{|x_1 -x_2|} S_2^{-1} S_1^{-1} = \; & S_1^{-1} S_2^{-1} (1-\Delta_1) \frac{1}{|x_1 -x_2|} S_2^{-1} S_1^{-1} \\ = \; & S_1^{-1} S_2^{-1} \frac{1}{|x_1 -x_2|} S_2^{-1} S_1^{-1} \\ & + S_1^{-1} S_2^{-1} \nabla_1^* \frac{1}{|x_1- x_2|} \nabla_1 S_2^{-1} S_1^{-1} + S_1^{-1} S_2^{-1} \nabla_1^* \frac{(x_1 -x_2)}{|x_1 -x_2|^3} S_1^{-1} S_2^{-1}\,,
\end{split}
\end{equation*}
and we use the norm-estimates $\| \nabla_1 S_1^{-1} \| < \infty$ and $\| S_2^{-1} |x_1 - x_2|^{-\alpha} S_2^{-1}\| < \infty$ for all $0\leq \alpha \leq 2$ (by (\ref{eq:hardy})).

\medskip

Suppose now that $\{ \gamma_{i,t}^{(k)} \}_{ k \geq 1}$, for $i =1,2$ are two solutions to the infinite hierarchy (\ref{eq:infhier-cou}). Using (\ref{eq:duhbd2}), we can expand both $\gamma^{(k)}_{1,t}$ and $\gamma^{(k)}_{2,t}$ in a Duhamel series. {F}rom (\ref{eq:Bcou}), and from the fact that $\| \cU^{(k)} \gamma^{(k)} \|_{\cH_k} = \| \gamma^{(k)} \|_{\cH_k}$, we obtain that
\begin{equation*}
\begin{split}
\left\| \gamma^{(k)}_{1,t} - \gamma^{(k)}_{2,t} \right\|_{\cH_k} \leq \; &C^n \frac{(k+n)!}{k!} \int_0^t \rd s_1 \dots \int_0^{s_{n-1}} \rd s_n \, \| \gamma^{(k+n)}_{1,s_n} - \gamma^{(k+n)}_{2,s_n} \|_{\cH_{k+n}} \\ \leq \; &C^k (Ct)^n \end{split}\end{equation*}
for any $n$. Here we used the a-priori bounds (\ref{eq:aprik-cou}). For $t \leq 1/ (2C)$, the l.h.s. must vanish. This shows uniqueness for short time, and thus, by iteration, for all times.
\end{proof}

\section{Dynamics of Bose-Einstein Condensates: the Gross-Pitaevskii Equation}
\setcounter{equation}{0}\label{sec:GP}

Dilute Bose gases at very low temperature are characterized by the macroscopic
occupancy of a single one-particle state; a non-vanishing fraction of the total
number of particles $N$ is described by the same one-particle orbital. Although
this phenomenon, known as Bose-Einstein condensation, has been predicted in the
early days of quantum mechanics, the first experimental evidence for its existence
was only obtained in 1995, in experiments performed by groups led by
Cornell and Wieman at the University of Colorado at Boulder and by Ketterle at MIT
(see \cite{CW,Kett}). In these important experiments, atomic gases were initially
trapped by magnetic fields and cooled down at very low temperatures. Then the magnetic
traps were switched off and the consequent time evolution of the gas was observed;
for sufficiently small temperatures, it was observed that the gas coherently moves
as a single particle, a clear sign for the existence of condensation.

\medskip

To describe these experiments from a theoretical point of view, we have, first of all,
to give a precise definition of Bose-Einstein condensation. It is simple to understand the meaning of condensation if one considers factorized wave functions, given by the (symmetrization of the) product of one-particle orbitals. In this case, to decide whether we have condensation, we only have to count the number of particles occupying every orbital; if there is a single orbital with macroscopic occupancy the wave function exhibits Bose-Einstein condensation, otherwise it does not. In particular, wave functions of the form $\psi_N (\bx) = \prod_{j=1}^N \ph (x_j)$, for some $\ph \in L^2 (\bR^3)$ (we consider in this section three dimensional systems only), exhibit Bose-Einstein condensation; since in these examples all particles occupy the same one-particle orbital, we say that $\psi_N$ exhibits complete Bose-Einstein condensation in the state $\ph$.

\medskip

Although factorized wave functions were used as initial data in Theorem \ref{thm:bd} and Theorem \ref{thm:cou}, they are, from a physical point of view, not very satisfactory, because they do not allow for any correlation among the particles. Since we would like to consider systems of interacting particles, the complete absence of correlations is not a realistic assumption. For this reason, we want to give a definition of Bose-Einstein condensation, in particular of complete Bose-Einstein condensation, that applies also to wave functions which are not factorized. To this end, we will make use of the one-particle density $\gamma^{(1)}_N$ associated with an $N$-particle wave function $\psi_N$. By definition (see (\ref{eq:gammakN})), the one-particle density is a non-negative trace class operator on $L^2 (\bR^3)$ with trace equal to one. It is simple to verify that the eigenvalues of $\gamma^{(1)}_N$ (which are all non-negative and sum up to one) can be interpreted as probabilities for finding particles in the state described by the corresponding eigenvector (a one-particle orbital). This observation justifies the following definition of Bose-Einstein condensation. We will say that a sequence $\{ \psi_N \}_{N \in \bN}$ with $\psi_N \in L^2_s (\bR^{3N})$ exhibits complete Bose-Einstein condensation in the one-particle state with orbital $\ph \in L^2 (\bR^3)$ if \begin{equation}\label{eq:BEC} \tr \; \left| \gamma^{(1)}_N - |\ph \rangle \langle \ph| \right| \to 0 \end{equation} as $N \to \infty$. In particular, complete Bose-Einstein condensation implies that the largest eigenvalue of $\gamma^{(1)}_N$ converges to one, as $N \to \infty$. More generally, we say that a sequence $\{ \psi_N \}_{N \in \bN}$ exhibits (not necessarily complete) Bose-Einstein condensation if the largest eigenvalue of $\gamma^{(1)}_N$ remains strictly positive in the limit $N \to \infty$. Note that condensation is not a property of a single $N$-particle wave function $\psi_N$, but it is a property characterizing a sequence $\{ \psi_N \}_{N \in \bN}$ in the limit $N \to \infty$.

\medskip

It is in general very difficult to verify that Bose-Einstein condensation occurs in physically interesting wave functions of interacting systems. There exists, however, a class of interacting systems for which complete condensation of the ground state has been recently established.

\medskip

In \cite{LSY}, Lieb, Yngvason, and Seiringer considered a trapped Bose
gas consisting of $N$ three-dimensional particles described by the
Hamiltonian
\begin{equation}\label{eq:ham1} H^{\text{trap}}_N
= \sum_{j=1}^N \left(-\Delta_j +
V_{\text{ext}} (x_j) \right) + \sum_{i<j}^N V_N (x_i -x_j),
\end{equation} where $V_{\text{ext}}$ is an external confining
potential with $\lim_{|x|\to \infty} V_{\text{ext}} (x) = \infty$, and $V_N (x)= N^2 V (Nx)$, where $V$ is pointwise positive, spherically symmetric, and rapidly decaying (for simplicity, $V$ can be thought of as being compactly supported). Note that the potential $V_N$ scales with $N$ so that its scattering length is of the order $1/N$ (Gross-Pitaevskii scaling). The scattering length of $V$ is a physical quantity measuring the effective range of the potential; two particles interacting through $V$ see each others, when they are far apart, as hard spheres with radius given by the scattering length of $V$. More precisely, if $f$ denotes the spherical symmetric solution to the zero-energy scattering equation \begin{equation}\label{eq:0en} \left( -\Delta + \frac{1}{2} V (x) \right) f = 0 \qquad \text{with boundary condition } f (x) \to 1 \quad \text{as } |x| \to \infty, \end{equation} the
scattering length of $V$ is defined by \[ a_0 = \lim_{|x| \to
\infty} |x| - |x| f(x)\,. \] This limit can be proven to exist if
$V$ decays sufficiently fast at infinity. Another equivalent characterization of
the scattering length is given by \begin{equation}\label{eq:8pia0} 8
\pi a_0 = \int \rd x \, V(x) f(x) \, .\end{equation}
It is simple to verify that, if $f$ solves (\ref{eq:0en}), the rescaled function $f_N (x) = f(Nx)$ solves the zero energy scattering equation with rescaled potential $V_N$, that is \begin{equation}\label{eq:0enN} \left( - \Delta + \frac{1}{2} V_N \right) f_N = 0 \qquad \text{with } \quad f_N (x) \to 1 \quad \text{as } |x| \to \infty \, .\end{equation} This implies immediately that the scattering length of $V_N$ is given by $a= a_0/N$, where $a_0$ is the scattering length of the unscaled potential $V$. Note that, for $|x| \gg a$, $f_N (x) \simeq 1- a/|x|$. For $|x| < a$, $f_N$ remains bounded; for practical purposes, we can think of this function as $f_N (x) \simeq 1 - a/(|x|+a)$.

\medskip

Letting $N \to \infty$,  Lieb, Yngvason, and Seiringer showed that the ground state energy $E (N)$ of (\ref{eq:ham1}) divided by the number of particle $N$ converges to
\[ \lim_{N \to \infty}\frac{E(N)}{N} =
\min_{\ph \in L^2 (\bR^3): \; \| \ph\| =1} \cE_{\text{GP}} (\ph) \]
where $\cE_{\text{GP}}$ is the Gross-Pitaevskii energy functional
\begin{equation}\label{eq:GPfunc} \cE_{\text{GP}} (\ph) = \int \rd x \; \left( |\nabla \ph
(x)|^2 + V_{\text{ext}} (x) |\ph (x)|^2 + 4 \pi a_0 |\ph (x)|^4
\right) \, .\end{equation}

\medskip

Later, in \cite{LS}, Lieb and Seiringer also proved that the ground state of the Hamiltonian (\ref{eq:ham1}) exhibits complete Bose-Einstein condensation into the minimizer of the Gross-Pitaevskii energy functional $\cE_{\text{GP}}$. More precisely they showed that, if $\psi_N$ is the ground state wave function of the Hamiltonian
(\ref{eq:ham1}) and if $\gamma^{(1)}_N$ denotes the corresponding
one-particle marginal, then
\begin{equation}\label{eq:conden}
\gamma_N^{(1)} \to |\phi_{\text{GP}} \rangle \langle
\phi_{\text{GP}}| \qquad \text{as } N \to \infty \, ,
\end{equation}
where $\phi_{\text{GP}} \in L^2 (\bR^3)$ is the minimizer of the
Gross-Pitaevskii energy functional (\ref{eq:GPfunc}).

\medskip

To describe the experiments mentioned above, it is important to understand the time-evolution of the Bose-Einstein condensate after removing the external traps. We define therefore the translation invariant Hamiltonian
\begin{equation}\label{eq:hamGP}
H_N = \sum_{j=1}^N -\Delta_j + \sum_{i<j}^N V_N (x_i -x_j) \,
\end{equation}
and we consider solutions to the $N$-particle Schr\"odinger equation
\begin{equation}\label{eq:schrGP}
i \partial_t \psi_{N,t} = H_N \psi_{N,t} \quad \Rightarrow \quad \psi_{N,t} = e^{-iH_N t} \psi_N \end{equation}
with initial data $\psi_N$ exhibiting complete Bose-Einstein condensation. In a series of joint articles with L. Erd\H os and H.-T. Yau, see \cite{ESY2,ESY3,ESY4,ESY5}, we prove that, for every fixed time $t \in \bR$, the evolved $N$-particle wave function $\psi_{N,t}$ still exhibits complete Bose-Einstein condensation. Moreover we show that the time evolution of the condensate wave function evolves according to the one-particle time-dependent Gross-Pitaevskii equation associated with the energy functional $\cE_{\text{GP}}$. Our main result is the following theorem.
\begin{theorem}\label{thm:mainGP}
Suppose that $V \geq 0$ is spherically symmetric and $V(x) \leq C \langle x \rangle^{-\sigma}$,
for some $\sigma >5$, and for all $x \in \bR^3$. Assume that the family $\{ \psi_N \}_{N \in \bN}$ with $\psi_N \in L_s^2 (\bR^{3N})$ and $\| \psi_N \| =1$ for all $N$, has finite energy per particle, that is \begin{equation}\label{eq:assH1} \langle \psi_N, H_N \psi_N \rangle
\leq C N \end{equation} for all $N \in \bN$, and that it exhibits complete Bose-Einstein condensation in
the sense that
\begin{equation}\label{eq:asscond} \tr \; \left| \gamma_N^{(1)} - |\ph \rangle \langle \ph| \right| \to 0 \end{equation} as $N \to \infty$ for some $\ph \in L^2 (\bR^3)$. Then, for every $k \geq 1$ and $t \in \bR$, we have
\begin{equation}\label{eq:claimGP} \tr \; \left| \gamma^{(k)}_{N,t} - |\ph_t \rangle \langle \ph_t|^{\otimes k} \right| \to 0 \end{equation} as $N \to \infty$. Here $\ph_t$ is the solution of the nonlinear Gross-Pitaevskii
equation \begin{equation}\label{eq:GPthm} i \partial_t \ph_t = -\Delta\ph_t +
 8\pi a_0 |\ph_t|^2 \ph_t \end{equation} with initial data $\ph_{t=0} = \ph$.
\end{theorem}

Making use of an approximation of the initial $N$-particle wave function (similarly to (\ref{eq:reginit})), it is possible to replace the assumption (\ref{eq:assH1}) with the much more stringent condition
\begin{equation}\label{eq:asskGP}
\langle \psi_N , H_N^k \psi_N \rangle \leq C^k N^k
\end{equation}
for all $k \in \bN$. In the following we will illustrate the main ideas involved in the proof of Theorem \ref{thm:mainGP}, assuming the initial wave function $\psi_N$ to satisfy (\ref{eq:asskGP}).

\subsection{Comparison with Mean-Field Systems}
\label{sec:comp-mf}

The formal relation of the Hamiltonian (\ref{eq:hamGP}) with the mean-field Hamiltonian (\ref{eq:ham-mf}) considered in Section \ref{sec:mf} is evident; (\ref{eq:hamGP}) can in fact be rewritten as
\begin{equation}\label{eq:GP-mf} H_N = \sum_{j=1}^N -\Delta_j + \frac{1}{N} \sum_{i<j} v_N (x_i -x_j) \,  \end{equation} with $v_N (x) = N^3 V(Nx)$. Since we are considering three dimensional systems, $v_N$ converges to a delta-function in the limit of large $N$; at least formally, as $N \to \infty$, we have $v_N (x) \to b_0 \delta (x)$, where $b_0 = \int V(x) \rd x$. In other words, the Hamiltonian (\ref{eq:hamGP}) in the Gross-Pitaevskii scaling can be formally interpreted as a mean field Hamiltonian with an $N$-dependent potential which converges, as $N \to \infty$, to a $\delta$-function. Despite the formal similarity, it should be stressed that the physics described by the Gross-Pitaevskii Hamiltonian is completely different from the physics described by the mean field Hamiltonian (\ref{eq:ham-mf}). In a mean-field system, each particle typically interacts with all other particles through a very weak potential. The Gross-Pitaevskii Hamiltonian (\ref{eq:hamGP}), on the other hand, describes a very dilute gas, where interactions are very rare and at the same time very strong. Although the physics described by (\ref{eq:ham-mf}) and (\ref{eq:hamGP}) are completely different, due to the formal similarity of the two models, we may try to apply the strategy discussed in Section \ref{sec:3steps} to prove Theorem \ref{thm:mainGP}. In other words, we may try to prove Theorem \ref{thm:mainGP} by showing the compactness of the sequence $\Gamma_{N,t} = \{ \gamma^{(k)}_{N,t} \}_{k=1}^N$ with respect to an appropriate weak topology (it is going to be the same topology introduced in Section \ref{sec:3steps}), the convergence to an infinite hierarchy similar to (\ref{eq:infhier}), and the uniqueness of the solution to the infinite hierarchy. It turns out that it is indeed possible to extend the general strategy introduced in Section \ref{sec:3steps} to prove Theorem \ref{thm:mainGP}; however, as we will see, many important modifications of the arguments used for bounded or for Coulomb potential are required. We discuss next the main changes.

\medskip

First of all, the simple observation that, formally, $v_N (x) \to b_0 \delta(x)$ as $N \to \infty$ may lead to the conclusion that the evolution of the condensate wave function $\ph_t$ should be described by the nonlinear Hartree equation (\ref{eq:hartree}) with $V$ replaced by $b_0 \delta$, that is by the equation \begin{equation}\label{eq:wrong} i\partial_t \ph_t = -\Delta \ph_t + b_0 ( \delta * |\ph_t|^2) \ph_t  = -\Delta \ph_t + b_0 |\ph_t|^2 \ph_t \, . \end{equation}
Comparing with (\ref{eq:GPthm}), we note that (\ref{eq:wrong}) is characterized by a different coupling constant in front of the nonlinearity. The emergence of the scattering length in the Gross-Pitaevskii equation (\ref{eq:GPthm}) is the consequence of a subtle interplay between the $N$-dependent interaction potential and the short scale correlation structure developed by the solution of the $N$-particle Schr\"odinger equation $\psi_{N,t}$ (as we will see, the correlation structure varies on lengths of the order $1/N$, the same lengthscale characterizing the interaction potential). This remark implies that, in order to prove the convergence to the infinite hierarchy (where the coupling constant $8\pi a_0$ already appears), we will need to identify the singular correlation structure of $\psi_{N,t}$ (which is then inherited by the marginal densities $\gamma^{(k)}_{N,t}$). This is one of the main difficulties in the proof of the convergence, which was completely absent in the analysis of mean-field systems presented in Section \ref{sec:mf}; we will discuss it in more details in Section~\ref{sec:convGP} and in Section \ref{sec:largeV}.

\medskip

The presence of the correlation structure in $\psi_{N,t}$ also affects the proof of a-priori bounds
\begin{equation}\label{eq:apriGP0}
\tr \, (1-\Delta_1) \dots (1-\Delta_k) \gamma^{(k)}_{\infty,t} \leq C^k
\end{equation}
for all $k\geq 1$, $t \in\bR$, for the limit points $\Gamma_{\infty,t} = \{ \gamma^{(k)}_{\infty,t} \}_{k \geq 1}$ of the marginal densities $\Gamma_{N,t} = \{ \gamma^{(k)}_{N,t} \}_{k=1}^N$ associated with $\psi_{N,t}$. As in the case of a Coulomb potential discussed in Section \ref{sec:cou}, these a-priori bounds play a fundamental role because they allow us to restrict the proof of the uniqueness to a smaller class of densities. In Section \ref{sec:cou}, we derived the a-priori bounds for $\gamma_{\infty,t}^{(k)}$ making use of the estimates (\ref{eq:apricou}) which hold uniformly in $N$, for $N$ sufficiently large (in turns, the bounds (\ref{eq:apricou}) were obtained through the energy estimates of Proposition \ref{prop:enest}). In the present setting, however, bounds of the form \begin{equation}\label{eq:nottrue}
\tr \, (1-\Delta_1) \dots (1-\Delta_k) \gamma^{(k)}_{N,t} \leq C^k \end{equation} cannot hold uniformly in $N$,  because of the short scale correlation structure developed by the solution of the Schr\"odinger equation. Remember in fact that the short scale structure varies on the length scale $1/N$; therefore, when we take derivatives of $\gamma^{(k)}_{N,t}$ as in (\ref{eq:nottrue}), we cannot expect to obtain bounds uniform in $N$ (unless we take only one derivative, because of energy conservation). Although the marginals $\gamma^{(k)}_{N,t}$, for large but finite $N$, do not satisfy the strong estimates (\ref{eq:nottrue}), it turns out that one can still prove the a-priori bound (\ref{eq:apriGP0}) on the limit point $\gamma^{(k)}_{\infty,t}$. This is indeed possible because in the weak limit $N \to \infty$, the singular short scale correlation structure characterizing the marginal densities $\gamma^{(k)}_{N,t}$ disappears, producing limit points $\gamma^{(k)}_{\infty,t}$ which are much more regular than the densities $\gamma_{N,t}^{(k)}$. Because of the absence of estimates of the form (\ref{eq:nottrue}) for $\gamma^{(k)}_{N,t}$, the proof of the a-priori bounds for the limit points $\gamma^{(k)}_{\infty,t}$ requires completely new ideas with respect to what has been discussed in Section \ref{sec:cou}; we briefly discuss the most important ones in Section \ref{sec:apriGP}.

\medskip

Finally, the singularity of the interaction potential strongly affects the proof of the uniqueness of the solution to the infinite hierarchy. In Section \ref{sec:cou}, the main idea to prove the uniqueness of the infinite hierarchy was to expand the solution in a Duhamel series and to control all Coulomb potentials appearing in the expansion through Laplacians acting on appropriate variables and at the end to control the expectation of the Laplacians through the a-priori bounds (\ref{eq:apricou}) on the densities $\gamma^{(k)}_{\infty,t}$. In this argument, it was very important that the Coulomb potential can be controlled by the kinetic energy, in the sense of the operator inequality
\begin{equation}\label{eq:opine}
\frac{1}{|x|} \leq C (1 - \Delta ) \, .
\end{equation}
In the present setting, the Coulomb potential has to be replaced by a $\delta$-function. In three dimensions, the $\delta$-potential cannot be controlled by the kinetic energy. In other words, the bound \[ \delta (x) \leq C (1-\Delta)^{\alpha} \] is not true for $\alpha =1$; it only holds if $\alpha > 3/2$ (in three dimensions, the $L^{\infty}$ norm of a function can be controlled by the $H^{\alpha}$-norm, only if $\alpha >3/2$). This observation implies that the a-priori bounds (\ref{eq:apriGP0}) are not sufficient to conclude the proof of the uniqueness of the infinite hierarchy with delta-interaction (while similar bounds were enough to prove the uniqueness of the infinite hierarchy with Coulomb potential). Since it does not seem possible to improve the a-priori bounds to gain control of higher derivatives (one would need more than $3/2$ derivatives per particle), we need new techniques to prove the uniqueness of the infinite hierarchy. We will briefly discuss these new methods in Section~\ref{sec:uniqueGP}.

\subsection{Convergence to the Infinite Hierarchy}\label{sec:convGP}

The goal of this section is to discuss the main ideas used to prove the next proposition which identifies limit points of the sequence $\Gamma_{N,t} = \{ \gamma^{(k)}_{N,t} \}_{k=1}^N$ as solutions to a certain infinite hierarchy of equations (this proposition replaces Proposition \ref{prop:conv}, which was stated for mean-field systems with bounded interaction potential).
\begin{proposition}\label{prop:conv-GP}
Suppose that $V \geq 0$, with $V(x) \leq C \langle x \rangle^{-\sigma}$, for some $\sigma >5$, and for all $x \in \bR^3$. Assume that the sequence $\psi_N$ satisfies (\ref{eq:asscond}) and the additional assumption (\ref{eq:asskGP}). Fix $T >0$ and let $\Gamma_{\infty,t} = \{ \gamma^{(k)}_{\infty,t} \}_{k \geq 1} \in
\bigoplus_{k \geq 1} C([0,T] , \cL_k^1)$ be a limit point of $\Gamma_{N,t} =
\{ \gamma_{N,t}^{(k)} \}_{k =1}^N$ (with respect to the product topology $\tau_{\text{prod}}$ defined in Section \ref{sec:3steps}). Then $\Gamma_{\infty,t}$ is a solution to the infinite hierarchy
\begin{equation}\label{eq:infhier-GP}
\gamma^{(k)}_{\infty,t} = \cU^{(k)} (t) \gamma^{(k)}_{\infty,0} -
8 \pi a_0 i \sum_{j=1}^k \int_0^t \rd s \, \cU^{(k)} (t-s) \tr_{k+1}
\left[ \delta (x_j - x_{k+1}), \gamma_{\infty,s}^{(k+1)} \right]
\end{equation}
with initial data $\gamma_{\infty,0}^{(k)} = |\ph \rangle \langle \ph|^{\otimes k}$ (see (\ref{eq:free}) for the definition of $\cU^{(k)}$).
\end{proposition}
The detailed proof of this proposition can be found in \cite[Theorem 8.1]{ESY5} (for small interaction potential, see also \cite[Theorem 7.1]{ESY3}).

\medskip

To prove the proposition, we start by studying the time-evolution of the marginal densities $\gamma^{(k)}_{N,t}$, which is governed by the BBGKY hierarchy. In integral form, the BBGKY hierarchy is given by
\begin{equation}\label{eq:BBGKY-int}
\begin{split}
\gamma^{(k)}_{N,t} = \; &\cU^{(k)} (t) \gamma^{(k)}_{N,0} - i \sum_{i<j}^k \int_0^t \rd s \, \cU^{(k)} (t-s) \, \left[ V_N (x_i -x_j) , \gamma^{(k)}_{N,s} \right] \\ &-i (N-k) \sum_{j=1}^k \int_0^t \rd s \; \cU^{(k)} (t-s) \, \tr_{k+1} \, \left[ V_N (x_j -x_{k+1}), \gamma^{(k+1)}_{N,s} \right]\,.
\end{split}
\end{equation}
Assuming (by passing to an appropriate subsequence) that $\Gamma_{N,t} \to \Gamma_{\infty,t}$ as $N \to \infty$ with respect to the product topology $\tau_{\text{prod}}$ introduced in Section \ref{sec:3steps}, it is simple to prove that the l.h.s. and the first term on the r.h.s. of (\ref{eq:BBGKY-int}) converge, as $N \to \infty$, to the l.h.s. and, respectively, to the first term on the r.h.s. of (\ref{eq:infhier-GP}). The second term on the r.h.s. of (\ref{eq:BBGKY-int}), on the other hand, can be proven to vanish in the limit $N \to \infty$ (at least formally, this follows by the observation that the second term is smaller by a factor of $N$ w.r.t. the third term). The fact that the second term on the r.h.s. of (\ref{eq:BBGKY-int}) is negligible in the limit $N \to \infty$ (compared with the third term) corresponds to the physical intuition that the interactions among the first $k$ particles affect their time-evolution less than their interaction with the other $(N-k)$ particles.

\medskip

To conclude the proof of Proposition \ref{prop:conv-GP}, we only need to show that the third term on the r.h.s. of (\ref{eq:BBGKY-int}) converges, as $N \to \infty$, to the last term on the r.h.s. of (\ref{eq:infhier-GP}). As already remarked in Section \ref{sec:comp-mf}, this convergence relies critically on the correlation structure characterizing the $(k+1)$-particle density $\gamma^{(k+1)}_{N,t}$. A naive approach, based on the observation that $(N-k) V_N (x_j -x_{k+1}) \simeq N^3 V (N (x_j - x_{k+1})) \simeq b_0 \delta (x_j -x_{k+1})$ for large $N$, fails to explain the coupling constant in front of the last term on the r.h.s. of (\ref{eq:infhier-GP}). The emergence of the scattering length can only be understood by taking into account the correlation structure of $\gamma^{(k+1)}_{N,t}$. Assuming for a moment that the correlations can be described, in good approximation, by the solution $f_N$ to the zero-energy scattering equation (\ref{eq:0enN}), we can expect that, for large $N$, \begin{equation}\label{eq:stru} \gamma^{(k+1)}_{N,t} (\bx_{k+1}; \bx'_{k+1}) \simeq f_N (x_j -x_{k+1}) \gamma^{(k+1)}_{\infty,t} (\bx_{k+1};\bx'_{k+1}) \end{equation} in the region where $x_j -x_{k+1}$ is of the order $1/N$ (and all other variables are at larger distances). Assuming some regularity of the limit point $\gamma^{(k+1)}_{\infty,t}$, and using (\ref{eq:8pia0}), the approximation (\ref{eq:stru}) immediately leads to
\begin{equation}\label{eq:stru2}
\begin{split}
\Big( \tr_{k+1} (N-k) V_N &(x_j - x_{k+1}) \gamma^{(k+1)}_{N,t} \Big) (\bx_k ; \bx'_k)  \\ \simeq \; & \int \rd x_{k+1} \, N^3 V (N (x_j -x_{k+1})) f (N (x_j -x_{k+1})) \gamma^{(k+1)}_{\infty,t} (\bx_{k}, x_{k+1} ; \bx'_{k}, x_{k+1} ) \\ = \; &\int \rd y \, V (y) f (y) \gamma^{(k+1)}_{\infty,t} \left(\bx_{k}, x_j + \frac{y}{N} ; \bx'_{k}, x_j + \frac{y}{N} \right) \\ \simeq \; &\left( \int \rd y \, V (y) f (y) \right) \gamma^{(k+1)}_{\infty,t} (\bx_k, x_j ; \bx'_k, x_j) \\ =\; & 8\pi a_0 \int \rd x_{k+1} \, \delta (x_j -x_{k+1})  \gamma^{(k+1)}_{\infty,t} (\bx_k, x_{k+1} ; \bx'_k, x_{k+1})
\end{split}
\end{equation}
and thus explains the emergence of the scattering length on the r.h.s. of (\ref{eq:infhier-GP}) (note that the third term on r.h.s. of (\ref{eq:BBGKY-int}) is a commutator and thus produces two summands; in (\ref{eq:stru2}) we only consider one of these terms, the other can be handled analogously). This heuristic argument shows that in order to prove Proposition \ref{prop:conv-GP} we need to identify the short scale structure of the marginal densities and prove that it can be described by the function $f_N$ as in (\ref{eq:stru}). To this end we are going to use energy estimates. In \cite{ESY3} and \cite{ESY5}, we developed two different approaches to this problem. The first approach is simpler, but it only works for sufficiently small interaction potentials. The second approach is a little bit more involved, but it can be used for all potentials satisfying the assumptions of Theorem \ref{thm:mainGP}. In the following we will focus on the first, simpler, approach; in the next subsection, we present the main ideas of the second approach.

\medskip

To measure the strength of the interaction potential $V$, we define the dimensionless constant
\begin{equation}\label{eq:rho} \rho = \sup_{x \in \bR^3} |x|^2 V (x) + \int \frac{\rd x}{|x|} V(x). \end{equation}
\begin{proposition}\label{prop:enestGP}
Assume that the potential $V$ satisfies the conditions of Theorem \ref{thm:mainGP}, and suppose that $\rho>0$ is sufficiently small. Then there exists $C > 0$ such that
\begin{equation}\label{eq:enesti}
\langle \psi, H_N^2 \psi \rangle \geq C N^2 \, \int \rd \bx \;
\left| \nabla_{i} \nabla_{j} \frac{\psi (\bx)}{f_N (x_i -x_j)}
\right|^2 \,
\end{equation}
for all $i \neq j$ and for all $\psi \in L^2_s (\bR^{3N}, \rd \bx)$.
\end{proposition}

This energy estimate, combined with the assumption (\ref{eq:asskGP}) on the initial wave function $\psi_N$, leads to the following a-priori bounds on the solution $\psi_{N,t}=e^{-iH_N t} \psi_N$ of the Schr\"odinger equation (\ref{eq:schrGP}).
\begin{corollary}\label{cor:apri}
Assume that $V$ satisfies the conditions of Theorem \ref{thm:mainGP}, and suppose that $\rho>0$ is sufficiently small. Suppose that $\psi_N$ satisfies (\ref{eq:assH1}) and (\ref{eq:asskGP}). Then we have
\begin{equation}\label{eq:aprior}
\int \rd \bx \, \left| \nabla_{i} \nabla_{j} \frac{\psi_{N,t}
(\bx)}{f_N (x_i -x_j)} \right|^2 \leq C
\end{equation}
for all $i \neq j$, uniformly in $N\in \bN$ and in $t \in \bR$. Therefore, if $\gamma^{(k)}_{N,t}$ denotes the $k$-particle marginal associated with $\psi_{N,t}$, we have, for every $1\leq i,j \leq k$ with $i \neq j$,
\[ \tr \; (1-\Delta_i) (1-\Delta_j) \frac{1}{f_N (x_i -x_j)} \gamma^{(k)}_{N,t} \frac{1}{f_N (x_i -x_j)} \leq C  \] uniformly in $N \in \bN$ and in $t \in \bR$.
\end{corollary}
\begin{proof}
Using (\ref{eq:enesti}), the conservation of the energy along the time evolution, and the assumption (\ref{eq:asskGP}) on the initial wave function $\psi_N$, we find
\begin{equation*} \int \rd \bx \, \left| \nabla_{i} \nabla_{j}
\frac{\psi_{N,t} (\bx)}{f_N (x_i -x_j)} \right|^2 \leq C N^{-2}
\langle \psi_{N,t}, H_N^2 \psi_{N,t} \rangle = C N^{-2} \langle
\psi_{N}, H_N^2 \psi_{N} \rangle \, \leq C .\end{equation*}
\end{proof}

Remark that the a-priori bounds (\ref{eq:aprior}) cannot hold true
if we do not divide the solution $\psi_{N,t}$ of the Schr\"odinger
equation by $f_N (x_i -x_j)$. In fact, using that $f_N (x) \simeq 1-a_0/(N|x| + 1)$, it is simple to check that
\[ \int \rd x \, |\nabla^2 f_N (x)|^2 \simeq N \, .\]
This implies that, if we replace $\psi_{N,t} (\bx)/ f_N (x_i
-x_j)$ by $\psi_N (\bx)$ the integral in (\ref{eq:aprior}) would be
of order $N$. Only after removing the singular factor $f_N
(x_i -x_j)$ from $\psi_{N,t} (\bx)$ we can obtain useful bounds on the regular part of the wave function (regular in the variable $(x_i-x_j)$). These a-priori bounds allow us to identify the correlation structure of the wave function $\psi_{N,t}$ and to show that, when $x_i$ and $x_j$ are close to each other,
$\psi_{N,t} (\bx)$ can be approximated by the time independent
correlation factor $f_N (x_i -x_j)$, which varies on the length scale
$1/N$, multiplied with a regular part (which only varies on scales of order one). In other words, the bounds (\ref{eq:aprior}) establish a strong separation of scales for the solution $\psi_{N,t}$ of the $N$-particle Schr\"odinger equation, and for its marginal densities; on length scales of order $1/N$, $\psi_{N,t}$ is characterized by a singular, time independent, short scale correlation structure described by the  the solution $f_N$ to the zero-energy scattering equation. On scales of order one, on the other hand, the wave function $\psi_{N,t}$ is regular, and, as it follows from Theorem \ref{thm:mainGP}, it can be approximated, in an appropriate sense, by products of the solution to the time-dependent Gross-Pitaevskii equation. Remark that although the short-scale correlation structure is time independent, it still affects, in a non-trivial way, the time-evolution on length scales of order one (because it produces the scattering length in the Gross-Pitaevskii equation).

\medskip

\begin{proof}[Proof of Proposition \ref{prop:enestGP}]
We decompose the Hamiltonian (\ref{eq:hamGP}) as
\[ H_N = \sum_{j=1}^N \, h_j \qquad \text{with} \qquad h_j = -\Delta_j
+ \frac{1}{2} \sum_{i \neq j} \, V_N (x_i - x_j)\,. \] For an
arbitrary permutation symmetric wave function $\psi$ and for any
fixed $i \neq j$, we have
\[ \langle \psi, H^2_N \psi \rangle = N \langle \psi, h_i^2
\psi \rangle + N (N-1) \langle \psi, h_i h_j \psi \rangle \geq N
(N-1) \langle \psi, h_i h_j \psi \rangle \,. \] Using the positivity
of the potential, we find
\begin{equation}\label{eq:pro1}\langle \psi, H^2_N \psi \rangle \geq N
(N-1) \left\langle \psi, \left(-\Delta_i + \frac{1}{2} V_N (x_i
-x_j) \right) \left(-\Delta_j + \frac{1}{2} V_N (x_i -x_j) \right)
\psi \right\rangle \,. \end{equation} Next, we define $\phi (\bx)$
by $\psi (\bx) = f_N (x_i -x_j) \, \phi (\bx)$ ($\phi$ is well
defined because $f_N (x) >0$ for all $x \in \bR^3$); note that the
definition of the function $\phi$ depends on the choice of $i,j$.
Then
\[ \frac{1}{f_N (x_i -x_j)} \Delta_i \, \left( f_N (x_i-x_j)
\phi (\bx) \right) = \Delta_i \phi (\bx) + \frac{(\Delta f_N )(x_i
-x_j)}{f_N (x_i -x_j)} \phi (\bx) + \frac{\nabla f_N (x_i -x_j)}{f_N
(x_i -x_j)} \nabla_i \phi (\bx) \, . \] {F}rom (\ref{eq:0en}) it
follows that
\[ \frac{1}{f_N (x_i -x_j)} \left( -\Delta_i + \frac{1}{2} V_N (x_i
-x_j) \right) f_N (x_i -x_j) \phi (\bx) = L_i \phi (\bx)
\] and analogously
\[
\frac{1}{f_N (x_i -x_j)} \left(-\Delta_j + \frac{1}{2} V_N (x_i
-x_j) \right) f_N (x_i -x_j) \phi (\bx) = L_j \phi (\bx)
\]
where we defined
\[ L_{\ell}  = -\Delta_{\ell} + 2 \frac{\nabla_{\ell} \, f_N (x_i
-x_j)}{f_N (x_i -x_j)} \, \nabla_{\ell}, \qquad \text{for } \quad
\ell=i,j \,.
\] Remark that, for $\ell =i,j$, the operator $L_{\ell}$ satisfies
\[ \int \rd \bx \, f_N^2 (x_i -x_j) \; L_{\ell} \, \overline{\phi} (\bx) \;
\psi (\bx) =\int \rd \bx \, f_N^2 (x_i -x_j) \; \overline{\phi}
(\bx) \; L_{\ell} \, \psi (\bx) = \int \rd \bx \, f_N^2 (x_i -x_j)
\; \nabla_{\ell} \, \overline{\phi} (\bx) \; \nabla_{\ell} \, \psi
(\bx) \, .\] Therefore, from (\ref{eq:pro1}), we obtain
\begin{equation}\label{eq:pro2}
\begin{split}
\langle \psi, H^2_N \psi \rangle \geq \; &N (N-1) \int \rd \bx \;
f^2_N (x_i -x_j) \; L_i\, \overline{\phi} (\bx)\, L_j\, \phi (\bx)
\\ = \;&N (N-1) \int \rd \bx \; f^2_N (x_i -x_j) \; \nabla_i
\overline{\phi} (\bx)\, \nabla_i L_j \, \phi (\bx)
\\= \;&N (N-1) \int \rd \bx \; f^2_N (x_i -x_j) \;
\nabla_i \overline{\phi} (\bx)\, L_j \, \nabla_i \phi (\bx)
\\ &+N (N-1) \int \rd \bx \; f^2_N (x_i -x_j) \;
\nabla_i \overline{\phi} (\bx)\, [\nabla_i, L_j] \phi (\bx)
\\= \;&N (N-1) \int \rd \bx \; f^2_N (x_i -x_j) \; \left| \nabla_j
\nabla_i \phi (\bx) \right|^2
\\ &+N (N-1) \int \rd \bx \; f^2_N (x_i -x_j) \; \left( \nabla_i
\frac{\nabla f_N (x_i -x_j)}{f_N (x_i -x_j)} \right) \;  \nabla_i
\overline{\phi} (\bx)\, \nabla_j \phi (\bx)\,.
\end{split}
\end{equation}
To control the second term on the right hand side of the last
equation we use bounds on the function $f_N$, which can be derived
from the zero energy scattering equation (\ref{eq:0en}):
\begin{equation}\label{eq:pro3} 1 - C \rho \leq f_N (x) \leq 1, \quad |\nabla f_N
(x)| \leq C \frac{ \rho }{|x|}, \quad |\nabla^2 f_N (x) | \leq C
\frac{\rho}{|x|^2} \end{equation} for constants $C$ independent of
$N$ and of the potential $V$ (recall the definition of the
dimensionless constant $\rho$ from (\ref{eq:rho})). Therefore,
for $\rho <1$,
\begin{equation*}
\begin{split}
\Big| \int \rd \bx \; f^2_N (x_i -x_j) \; &\left( \nabla_i
\frac{\nabla f_N (x_i -x_j)}{f_N (x_i -x_j)} \right) \;  \nabla_i
\overline{\phi} (\bx)\, \nabla_j \phi (\bx) \Big| \\
\leq \; & C \rho \int \rd \bx \; \frac{1}{|x_i -x_j|^2} \,
|\nabla_i \phi (\bx)| \, |\nabla_j \phi (\bx)| \\ \leq \; & C \rho
\int \rd \bx \; \frac{1}{|x_i -x_j|^2} \, \left( |\nabla_i \phi
(\bx)|^2 + |\nabla_j \phi (\bx)|^2 \right)
\\ \leq \; &C \rho \int \rd \bx \; |\nabla_i \nabla_j \phi
(\bx)|^2
\end{split}
\end{equation*}
where we used Hardy inequality. Thus, from (\ref{eq:pro2}), and
using again the first bound in (\ref{eq:pro3}), we obtain
\[\langle \psi, H^2_N \psi \rangle \geq \; N (N-1) (1- C \rho)
\int \rd \bx \left| \nabla_i \nabla_j \phi (\bx) \right|^2 \] which
implies (\ref{eq:enesti}).
\end{proof}

\bigskip

Equipped with the a-priori bounds of Corollary \ref{cor:apri}, we can now come back to the problem of proving the convergence of the last term on the r.h.s. of (\ref{eq:BBGKY-int}) to the last term on the r.h.s. of (\ref{eq:conv}). For simplicity, we consider the case $k=1$, and we only discuss the term with the interaction potential on the left of the density (the commutator also has a term with the interaction on the right of the density, which can be handled analogously). After multiplying with a smooth one-particle observable $J^{(1)}$ (a compact operator on $L^2 (\bR^3)$, with sufficiently smooth kernel), we need to prove that
\begin{equation*}
\tr \;  \left(\cU^{(1)} (s-t) J^{(1)} \right) \left( N^3 V (N(x_1 -x_2)) \gamma^{(2)}_{N,t} - 8\pi a_0 \delta (x_1 -x_2) \gamma^{(2)}_{\infty,t} \right) \to 0 \end{equation*}
as $N \to \infty$. To this end we decompose the difference in several terms. We use the notation $J^{(1)}_t = \cU^{(1)} (t) J^{(1)}$, and, for a bounded function $h(x) \geq 0$ with $\int \rd x \, h (x)  =1$, we define $h_{\alpha} (x) = \alpha^{-3} h (\alpha^{-1} x)$ for all $\alpha >0$. Then we have
\begin{equation}\label{eq:deco}
\begin{split}
\tr &\;\left(\cU^{(1)} (s-t) J^{(1)} \right) \left( N^3 V (N(x_1 -x_2)) \gamma^{(2)}_{N,t} - 8\pi a_0 \delta (x_1 -x_2) \gamma^{(2)}_{\infty,t} \right) \\
=\; & \tr \; J^{(1)}_{s-t} \,  N^3 V (N(x_1 -x_2)) f (N(x_1 -x_2)) \frac{1}{f(N(x_1 -x_2))} \gamma^{(2)}_{N,t}\frac{1}{f(N(x_1-x_2))} (f (N (x_1 -x_2)) - 1)  \\ &+
\tr \; J^{(1)}_{s-t}  \left(N^3 V (N(x_1 -x_2)) f (N(x_1 -x_2)) -8\pi a_0 \delta (x_1 -x_2) \right) \frac{1}{f(N(x_1 -x_2))} \gamma^{(2)}_{N,t}\frac{1}{f(N(x_1-x_2))} \\
&+ 8\pi a_0 \, \tr \; J^{(1)}_{s-t}  \left(\delta (x_1 -x_2) - h_{\alpha} (x_1 -x_2) \right) \frac{1}{f(N(x_1 -x_2))} \gamma^{(2)}_{N,t}\frac{1}{f(N(x_1-x_2))} \\
&+ 8\pi a_0 \, \tr \; J^{(1)}_{s-t} h_{\alpha} (x_1 -x_2) \left( \frac{1}{f(N(x_1 -x_2))} \gamma^{(2)}_{N,t}\frac{1}{f(N(x_1-x_2))} - \gamma^{(2)}_{N,t} \right) \\
&+ 8\pi a_0 \, \tr \; J^{(1)}_{s-t} h_{\alpha} (x_1 -x_2) \left(\gamma^{(2)}_{N,t} - \gamma^{(2)}_{\infty,t} \right) \\
&+ 8\pi a_0 \, \tr \; J^{(1)}_{s-t} \left( h_{\alpha} (x_1 -x_2) -\delta (x_1 -x_2) \right) \gamma^{(2)}_{\infty,t}\,.
\end{split}
\end{equation}
The idea here is that in order to compare the $N$-dependent potential $N^3 V(N(x_1-x_2))$ with the limiting $\delta$-potential, we have to test it against a regular density (using an appropriate Poincar{\'e} inequality). For this reason, we first regularize the density $\gamma^{(2)}_{N,t}$ in the variable $(x_1 -x_2)$ dividing it by the correlation function $f_N (x_1-x_2)$ on the left and the right (first term on the r.h.s. of the last equation). Using the regularity of $f_N^{-1} (x_1 -x_2) \gamma^{(2)}_{N,t} f_N^{-1} (x_1-x_2)$ from Corollary \ref{cor:apri}, we can then compare, in the regime of large $N$, the interaction potential with the delta-function (second term on the r.h.s.). At this point we are still not done, because, in order to remove the regularizing factors $f_N^{-1} (x_1 -x_2)$ (fourth term on the r.h.s. of (\ref{eq:deco})) and in order to replace the density $\gamma^{(2)}_{N,t}$ by its limit point $\gamma^{(2)}_{\infty,t}$ (fifth term on the r.h.s. of (\ref{eq:deco})), we need to test the density against a compact observable. For this reason, in the third term on the r.h.s. of (\ref{eq:deco}), we replace the $\delta$-function (which is of course not bounded) by the function $h_{\alpha}$ which approximate the delta-function on the length scale $\alpha$; it is important here that $\alpha$ is now decoupled from $N$. In the last term, after removing all the $N$ dependence, we go back to the $\delta$-potential using the regularity of the limiting density $\gamma^{(2)}_{\infty,t}$.

\medskip

To control the first and fourth term on the r.h.s. of (\ref{eq:deco}), we use the fact that $1- f_N (x_1 -x_2)\simeq 1/(N|x_1 -x_2|+1)$ varies on a length scale of order $1/N$. It follows that the first term converges to zero as $N \to\infty$, as well as the fourth term, for every fixed $\alpha >0$. To estimate the second, the third and the last term, we make use of appropriate Poincar{\'e} inequalities, combined with the result of Corollary \ref{cor:apri} and, for the last term, of Proposition \ref{prop:aprikGP} (we present an example of a Poincar{\'e} inequality, which can be used to estimate these terms in Appendix \ref{app:poin}). It follows that the second term converges to zero as $N \to \infty$, and that the third and the fifth terms converge to zero as $\alpha \to 0$, uniformly in $N$. Finally, the fifth term on the r.h.s. of (\ref{eq:deco}) converges to zero as $N \to \infty$, for every fixed $\alpha$; this follows from the assumption that $\gamma^{(2)}_{N,t} \to \gamma^{(2)}_{\infty,t}$ as $N \to \infty$ with respect to the weak* topology (some additional work has to be done here, because the operator $J^{(1)}_{s-t} h_{\alpha} (x_1-x_2)$ is not compact). Therefore, if we first fix $\alpha >0$ and let $N \to \infty$ and then we let $\alpha \to 0$ all terms on the r.h.s. of (\ref{eq:deco}) converge to zero; this concludes the proof of Proposition \ref{prop:conv-GP}.

\subsection{Convergence for Large Interaction Potentials}
\label{sec:largeV}

As pointed out in Section \ref{sec:convGP}, the energy estimate given in Proposition \ref{prop:enestGP}, which was a crucial ingredient for the proof of Proposition   \ref{prop:conv-GP}, only holds for sufficiently small potentials (for sufficiently small values of the parameter $\rho$ defined in (\ref{eq:rho})). For large potentials, we need a different approach. The new technique, developed in \cite{ESY5}, is based on the use of the wave operator associated with the one-particle Hamiltonian $\fh_N = -\Delta + (1/2) V_N$, defined through the strong limit
\begin{equation}\label{eq:waveop}
W_N = s-\lim_{t\to \infty} e^{i\fh_N t} e^{i\Delta t}\,.
\end{equation}
Under the assumptions of Theorem \ref{thm:mainGP} on the potential $V$, it is simple to show that the limit (\ref{eq:waveop}) exists, that the wave operator $W_N$ is complete, in the sense that
\begin{equation*}
W_N^{-1} = W_N^* = s-\lim_{t\to \infty} e^{-i\Delta t} e^{-i\fh_N t}\, ,
\end{equation*}
and that it satisfies the intertwining relation
\begin{equation}\label{eq:inter}
W^*_N \, \fh \, W_N = -\Delta \, .
\end{equation}
It is also important to observe that the wave operator $W_N$ is related by simple scaling to the wave operator $W$ associated with the one-particle Hamiltonian $\fh = -\Delta +(1/2) V$ (and defined analogously to (\ref{eq:waveop})). In fact, if $W_N (x;x')$ and $W (x;x')$ denote the kernels of $W_N$ and, respectively, of $W$, we have
\[ W_N (x;x') = N^3 W (Nx;Nx') \, \qquad \text{and } \quad W^*_N (x;x') = N^3 W^* (Nx ; Nx') \, .\]
In particular this implies that the norm of $W_N$, as an operator from $L^p (\bR^3)$ to $L^p (\bR^3)$, for arbitrary $1\leq p \leq \infty$, is independent of $N$. From the work of Yajima, see \cite{Ya,Ya0}, we know that, under the conditions on $V$ assumed in Theorem \ref{thm:mainGP}, $W$ is a bounded operator from $L^p (\bR^3)$ to $L^p (\bR^3)$, for all $1\leq p \leq \infty$. Therefore
\[ \| W_N \|_{L^p \to L^p} = \| W \|_{L^p \to L^p} < \infty \qquad \text{for all } 1 \leq p \leq \infty \, . \]
In the following we will denote by $W_{N, (i,j)}$ the wave operator $W_N$ acting only on the relative variable $x_j - x_i$. In other words, the action of $W_{N,(i,j)}$ on a $N$-particle wave function $\psi_N \in L^2 (\bR^{3N})$ is given by \begin{equation}\label{eq:Wij} \left( W_{N,(i,j)}
\psi_N \right) (\bx) = \int \rd v \; W_N (x_j - x_i; v) \, \psi_N \left( x_1, \dots ,
 \frac{x_i + x_j}{2} + \frac{v}{2}, \dots , \frac{x_i + x_j}{2} - \frac{v}{2},
\dots, x_N \right) \end{equation} if $j <i$ (the formula for $i >j$ is similar). Similarly, we define $W^*_{N,(i,j)}$. Using the wave operator we have the following energy estimate, which replaces Proposition~\ref{prop:enestGP}, and whose proof can be found in \cite[Proposition 5.2]{ESY5}.
\begin{proposition}\label{prop:energ2}
Suppose $V \geq 0$, $V \in L^1 (\bR^3) \cap L^2 (\bR^3)$ and $V(x) = V(-x)$ for all $x\in \bR^3$. Then we have, for every $i \neq j$,
\begin{equation}\label{eq:energ2}
\langle \psi_N , H_N^2 \psi_N \rangle \geq C N^2 \int \rd \bx \; \left|
\left(\nabla_i \cdot \nabla_j \right) \, W^*_{N,(i,j)} \psi_N \right|^2\, .
\end{equation}
\end{proposition}

{F}rom Proposition \ref{prop:energ2}, we obtain immediately an a-priori bound on $\psi_{N,t}$ and on its marginal densities.
\begin{corollary}\label{cor:apriV}
Assume that $V$ satisfies the conditions of Theorem \ref{thm:mainGP}. Suppose that $\psi_N$ satisfies (\ref{eq:assH1}) and (\ref{eq:asskGP}). Then we have, for all $i \neq j$,
\begin{equation}\label{eq:apriorV}
\int \rd \bx \, \left| \left( \nabla_{i} \cdot \nabla_{j} \right)  W^*_{N,(i,j)} \psi_{N,t} (\bx) \right|^2 \leq C
\end{equation}
uniformly in $N \in \bN$ and $t \in \bR$. Therefore, if $\gamma^{(k)}_{N,t}$ denote the $k$-particle marginal associated with $\psi_{N,t}$, we have, for every $1\leq i,j \leq k$ with $i \neq j$,
\[ \tr \; \left( (\nabla_i \cdot \nabla_j)^2 -\Delta_i - \Delta_j +1 \right) W_{N,(i,j)}^*  \gamma^{(k)}_{N,t} W_{N,(i,j)} \leq C  \] uniformly in $N \in \bN$ and in $t \in \bR$.
\end{corollary}

The philosophy of the bounds (\ref{eq:apriorV}) and (\ref{eq:aprior}) is the same; first we have to regularize the wave function $\psi_{N,t}$, and then we can prove useful bounds on its derivatives. There are however important differences. In (\ref{eq:aprior}) we regularized $\psi_{N,t}$ in position space, by factoring out the short scale correlation structure $f_N (x_i -x_j)$. In (\ref{eq:apriorV}), instead, we regularize $\psi_{N,t}$ applying the wave operator $W_{N,(i,j)}^*$. Another important difference is that (\ref{eq:apriorV}) is weaker than (\ref{eq:aprior}); in fact, (\ref{eq:apriorV}) only gives a control on the combination $\sum_{\alpha=1}^3 \partial_{x_{i,\alpha}} \partial_{x_{j,\alpha}}$, while (\ref{eq:aprior}) controls $\partial_{x_{i,\alpha}}\partial_{x_{j,\beta}}$ for all $1 \leq \alpha,\beta \leq 3$. The weakness of the bound (\ref{eq:apriorV}) makes the proof of the convergence more difficult. In particular we have to establish new Poincar{\'e} inequalities, which only require control of the inner product $\nabla_i \cdot \nabla_j$. It turns out that the weaker control provided by (\ref{eq:apriorV}) is still enough to conclude the proof of convergence to the infinite hierarchy (Proposition \ref{prop:conv-GP}). For more details, see \cite[Section 8]{ESY5}.

\subsection{A-Priori Estimates on Limit Points $\Gamma_{\infty,t}$}\label{sec:apriGP}

In this section we present some of the arguments involved in the proof of the a-priori bounds (\ref{eq:apriGP0}).
\begin{proposition} \label{prop:aprikGP}
Assume that $V$ satisfies the conditions of Theorem \ref{thm:mainGP}. Suppose that $\psi_N$ satisfies (\ref{eq:assH1}) and (\ref{eq:asskGP}). Let $\Gamma_{\infty,t}= \{\gamma_{\infty,t}^{(k)} \}_{k\geq 1} \in \bigoplus_{k \geq 1} C([0,T], \cL^1_k)$ be a limit point of the sequence $\Gamma_{N,t} = \{ \gamma^{(k)}_{N,t} \}_{k=1}^N$ with respect to the product topology $\tau_{\text{prod}}$ defined in Section \ref{sec:3steps}. Then $\gamma^{(k)}_{\infty,t} \geq 0$ and there exists a constant $C$ such that
\begin{equation}\label{eq:aprik3}
\tr \; (1- \Delta_1) \dots (1-\Delta_k) \gamma^{(k)}_{\infty,t} \leq C^k
\end{equation}
for all $k \geq 1$ and $t \in [0,T]$.
\end{proposition}
The main difficulty in proving Proposition \ref{prop:aprikGP} is the
fact that the estimate (\ref{eq:aprik3}) does not hold true if we
replace $\gamma^{(k)}_{\infty,t}$ with the marginal density
$\gamma^{(k)}_{N,t}$. More precisely,
\begin{equation}\label{eq:aprwrong} \tr \; (1-\Delta_1) \dots
(1-\Delta_k) \gamma^{(k)}_{N,t} \leq C^k \end{equation} cannot hold
true with a constant $C$ independent of $N$. In fact, for finite $N$
and $k >1$, the $k$-particle density $\gamma_{N,t}^{(k)}$ still
contains the singular short scale correlation structure. For example, when particle one and particle two are very close to each other (at distances of order $1/N$), we can expect the two-particle density to be approximately given by
\[ \gamma^{(2)}_{N,t} (\bx_2, \bx'_2) \simeq  \const \, f_N (x_1 -x_2) f_N (x'_1 -x'_2) \]
(the constant part takes into account factors which vary on larger scales). It is then simple to check that \[ \tr\; (1-\Delta_1) (1-\Delta_2) \gamma^{(2)}_{N,t} \simeq N \, . \]
Only after taking the weak limit $N \to \infty$, the short scale correlation structure disappears (because it varies on a length scale of order $1/N$), and one can hope to prove bounds like (\ref{eq:aprik3}).

\medskip

To overcome this problem, we cutoff the wave function $\psi_{N,t}$
when two or more particles come at distances smaller than some
intermediate length scale $\ell$, with $N^{-1} \ll \ell \ll 1$ (more
precisely, the cutoff will be effective only when one or more
particles come close to one of the variable $x_j$ over which we want
to take derivatives). For fixed $j= 1, \dots ,N$, we define
$\theta_j \in C^{\infty} (\bR^{3N})$ such that
\[ \theta_j (\bx) \simeq \left\{ \begin{array}{ll} 1 & \quad
\text{if} \quad |x_i -x_j| \gg \ell \quad \text{for all } i \neq j
\\
0 & \quad \text{if there exists } i \neq j \quad \text{with} \quad
|x_i -x_j| \lesssim \ell \end{array} \right. \; . \] It is
important, for our analysis, that $\theta_j$ controls its
derivatives (in the sense that, for example, $|\nabla_i \theta_j|
\leq C \ell^{-1} \theta^{1/2}_j$); for this reason we cannot use
standard compactly supported cutoffs. Instead we have to
construct appropriate functions which decay exponentially when
particles come close together (the prototype of such function is $\theta (x) = \exp (- \sqrt{(x/\ell)^2 + 1})$). Making use of the functions $\theta_j
(\bx)$, we prove the following higher order energy estimates.

\begin{proposition}\label{prop:highen}
Choose $\ell \ll 1$ such that $N \ell^2 \gg 1$. Suppose that
$\alpha$ is small enough. Then there exist constants $C_1$ and $C_2$
such that, for any $\psi \in L^2_s (\bR^{3N})$,
\begin{equation}\label{eq:highen}
\langle \psi, (H_N +C_1 N)^k \psi \rangle \geq C_2 N^k \int \rd \bx
\; \theta_1 (\bx) \dots \theta_{k-1} (\bx) \, |\nabla_1 \dots
\nabla_k \psi (\bx)|^2 \, .
\end{equation}
\end{proposition}

The meaning of the bound (\ref{eq:highen}) is clear. The
$L^2$-norm of the $k$-th derivative $\nabla_1 \dots \nabla_k
\psi$ can be controlled by the expectation of the $k$-th power of the energy per
particle, if we restrict the integration domain to regions where the first
$(k-1)$ particles are ``isolated'' (in the sense that there is no
particle at distances smaller than $\ell$ from $x_1, x_2, \dots,
x_{k-1}$).

\medskip

Note that we can allow one ``free derivative''; in (\ref{eq:highen})
we take the derivative over $x_k$ although there is no cutoff
$\theta_k (\bx)$. The reason is that the correlation structure
becomes singular, in the $L^2$-sense, only when we derive it twice
(if one uses the zero energy solution $f_N$ introduced in
(\ref{eq:0en}) to describe the correlations, this can be seen by
observing that $|\nabla f_N (x)| \leq 1/ |x|$, which is locally
square integrable). Remark that the condition $N\ell^2 \gg 1$ is necessary to control the error due to the localization of the kinetic energy on distances of order
$\ell$. The proof of Proposition
\ref{prop:highen} is based on induction over $k$; for details see
Section 7 in \cite{ESY5}.

\medskip

{F}rom the estimates (\ref{eq:highen}), using the preservation of
the expectation of $H_N^k$ along the time evolution and the condition (\ref{eq:asskGP}), we obtain the following bounds for the solution $\psi_{N,t} =
e^{-iH_N t} \psi_N$ of the Schr\"odinger equation (\ref{eq:schrGP}).
\begin{equation*}
\int \rd \bx \; \theta_1 (\bx) \dots \theta_{k-1} (\bx) \; \left|
\nabla_1 \dots \nabla_k \psi_{N,t} (\bx) \right|^2 \leq C^k
\end{equation*}
uniformly in $N$ and $t$, and for all $k \geq 1$. Translating these
bounds in the language of the density matrix $\gamma_{N,t}$, we
obtain \begin{equation}\label{eq:cutoffs} \tr \; \theta_1 \dots
\theta_{k-1} \nabla_1 \dots \nabla_k \gamma_{N,t} \nabla_1^* \dots
\nabla_k^* \leq C^k \,. \end{equation} The idea now is to use the
freedom in the choice of the cutoff length $\ell$. If we fix the
position of all particles but $x_j$, it is clear that the cutoff
$\theta_j$ is effective at most in a volume of the order $N\ell^3$.
If we choose $\ell$ such that $N\ell^3 \to 0$ as $N \to \infty$
(which is of course compatible with the condition that $N \ell^2 \gg
1$), we can expect that, in the limit of large $N$, the cutoff
becomes negligible. This approach yields in fact the desired
results; starting from (\ref{eq:cutoffs}), and choosing $\ell$ such
that $N \ell^3 \ll 1$, we can complete the proof of Proposition
\ref{prop:aprikGP} (see Proposition 6.3 in \cite{ESY3} for more
details).

\subsection{Uniqueness of the Solution to the Infinite Hierarchy}
\label{sec:uniqueGP}

To complete the proof of Theorem \ref{thm:mainGP} we have to prove the uniqueness of the solution to the infinite hierarchy (\ref{eq:infhier-GP}) in the class of densities satisfying the a-priori bounds (\ref{eq:aprik3}). Remark that the uniqueness of the infinite hierarchy (\ref{eq:infhier-GP}), in a different class of densities, was recently proven by Klainerman and Machedon in \cite{KM}. The proof proposed by Klainerman and Machedon is simpler than the proof of Proposition \ref{prop:uniqueGP} which we discuss below. Unfortunately, the result of \cite{KM} cannot be applied to the proof of Theorem \ref{thm:mainGP}, because it is not yet clear whether limit points of the sequence of marginal densities $\Gamma_{N,t} = \{ \gamma^{(k)}_{N,t} \}_{k=1}^N$ fit into the class of densities for which uniqueness is proven.
\begin{proposition}\label{prop:uniqueGP} Fix $T >0$ and $\Gamma = \{ \gamma^{(k)} \}_{k \geq
1} \in \bigoplus_{k \geq 1} \cL^1_k$. Then there exists at most one solution $\Gamma_{t} = \{ \gamma^{(k)}_{t} \}_{k \geq 1} \in \bigoplus C ([0,T], \cL_k)$ of the infinite hierarchy (\ref{eq:infhier-GP}) with $\Gamma_{t=0} = \Gamma$, such that $\gamma_t^{(k)} \geq 0$ is symmetric with respect to permutations, and
\begin{equation}\label{eq:aprit0}
\tr  \, (1-\Delta_1)  \dots (1-\Delta_k)  \, \gamma^{(k)}  \leq C^k
\end{equation} for all $k \geq 1$ and all $t \in [0,T]$.
\end{proposition}

In this section we briefly explain some of the main steps involved in the proof of Proposition \ref{prop:uniqueGP}; the details can be found in \cite{ESY2}[Section 9].

\medskip

To shorten the notation, we write the infinite hierarchy (\ref{eq:infhier-GP}) in the form
\begin{equation}\label{eq:GPhier-int}
\gamma_{t} = \cU^{(k)} (t) \gamma_{0} +
\int_0^t \rd s \; \cU^{(k)} (t-s) \, B^{(k)} \gamma^{(k+1)}_{s} \, ,
\end{equation}
where $\cU^{(k)} (t)$ denotes the free evolution of $k$ particles
\[ \cU^{(k)} (t) \gamma^{(k)} = e^{it\sum_{j=1}^k \Delta_j}
\gamma^{(k)} e^{-it\sum_{j=1}^k \Delta_j} \] and the collision
operator $B^{(k)}$ maps $(k+1)$-particle operators into $k$-particle
operators according to
\begin{equation}\label{eq:B}
B^{(k)} \gamma^{(k+1)} = -8i \pi a_0  \sum_{j=1}^k \tr_{k+1}\, \left[
\delta (x_j -x_{k+1}), \gamma^{(k+1)} \right] \,.
\end{equation}
The map $B^{(k)}$ is defined in analogy to Section \ref{sec:mf}; in particular the kernel of $B^{(k)} \gamma^{(k+1)}$ is given by the expression on the r.h.s. of (\ref{eq:trk+1}), with $V (x)$ replaced by $8\pi a_0 \delta (x)$.

\medskip

Iterating (\ref{eq:GPhier-int}) $n$ times we obtain the Duhamel type
series
\begin{equation}\label{eq:Duhamel}
\gamma^{(k)}_{t} = \cU^{(k)} (t) \gamma^{(k)}_0 + \sum_{m=1}^{n-1}
\xi^{(k)}_{m,t} + \eta^{(k)}_{n,t}
\end{equation}
with
\begin{equation}\label{eq:xi}
\begin{split}
\xi^{(k)}_{m,t} &= \int_0^t \rd s_1 \dots \int_0^{s_{m-1}} \rd s_m
\; \cU^{(k)} (t-s_1) B^{(k)} \cU^{(k+1)} (s_1 -s_2) B^{(k+1)} \dots
B^{(k+m-1)} \cU^{(k+m)} (s_m) \gamma^{(k+m)}_{0} \\&= \sum_{j_1=1}^k
\sum_{j_2=1}^{k+1} \dots \sum_{j_m=1}^{k+m} \int_0^t \rd s_1 \dots
\int_0^{s_{m-1}} \rd s_m \; \cU^{(k)} (t-s_1)  \tr_{k+1} \Big[
\delta (x_{j_1} - x_{k+1}),
\\  &\hspace{.3cm}  \left. \cU^{(k+1)} (s_1 -s_2) \tr_{k+2} \left[ \delta
(x_{j_2} - x_{k+2}), \dots \tr_{k+m} \left[ \delta ( x_{j_m}-
x_{k+m}), \cU^{(k+m)} (s_m) \gamma_0^{(k+m)} \right] \dots\right]
\right]
\end{split}
\end{equation} and the error term
\begin{equation}\label{eq:eta}
\eta^{(k)}_{n,t} =\int_0^t \rd s_1 \int_0^{s_1} \rd s_2 \dots
\int_0^{s_{n-1}} \rd s_n \; \cU^{(k)} (t-s_1) B^{(k)} \cU^{(k+1)}
(s_1 -s_2) B^{(k+1)} \dots B^{(k+n-1)} \gamma^{(k+m)}_{s_n}\,.
\end{equation}
Note that the error term (\ref{eq:eta}) has exactly the same form as
the terms in (\ref{eq:xi}), with the only difference that the last
free evolution is replaced by the full evolution
$\gamma^{(k+m)}_{s_n}$.

\medskip

To prove the uniqueness of the infinite hierarchy, it is enough to
prove that the fully expanded terms (\ref{eq:xi}) are well-defined and that
the error term (\ref{eq:eta}) converges to zero as $n \to
\infty$ (in some norm, or even after testing it against a
sufficiently large class of smooth observables). The main problem
here is that the delta function in the collision operator $B^{(k)}$
cannot be controlled by the kinetic energy (in the sense that, in
three dimensions, the operator inequality $\delta (x) \leq C(1-
\Delta)$ does not hold true). For this reason, the a-priori
estimates (\ref{eq:aprit0}) are not sufficient
to show that (\ref{eq:eta}) converges to zero, as $n \to \infty$.
Instead, we have to make use of the smoothing effects of the free
evolutions $\cU^{(k+j)} (s_j -s_{j+1})$ in (\ref{eq:eta}) (in a
similar way, Stricharzt estimates are used to prove the
well-posedness of nonlinear Schr\"odinger equations). To this end,
we rewrite each term in the series (\ref{eq:Duhamel}) as a sum of
contributions associated with certain Feynman graphs, and then we
prove the convergence of the Duhamel expansion by controlling each
contribution separately.

\bigskip

The details of the diagrammatic expansion can be found in \cite[Section 9]{ESY2}. Here we only sketch the main ideas. We start by
considering the term $\xi_{m,t}^{(k)}$ in (\ref{eq:xi}). After
multiplying it with a compact $k$-particle observable $J^{(k)}$ and
taking the trace, we expand the result as
\begin{equation}\label{eq:expa}
\tr \; J^{(k)} \xi^{(k)}_{m,t} = \sum_{\Lambda \in \cF_{m,k}}
K_{\Lambda,t}
\end{equation}
where $K_{\Lambda,t}$ is the contribution associated with the
Feynman graph $\Lambda$. Here $\cF_{m,k}$ denotes the set of all
graphs consisting of $2k$ disjoint, paired, oriented, and rooted
trees with $m$ vertices. An example of a graph in $\cF_{m,k}$ is
drawn in Figure \ref{fig:feynman}. Each vertex has one of the two
forms drawn in Figure \ref{fig:feynman}, with one ``father''-edge on
the left (closer to the root of the tree) and three ``son''-edges on
the right. One of the son edge is marked (the one drawn on the same
level as the father edge; the other two son edges are drawn below).
Graphs in $\cF_{m,k}$ have $2k+3m$ edges, $2k$ roots (the edges on
the very left), and $2k+2m$ leaves (the edges on the very right). It
is possible to show that the number of different graphs in
$\cF_{m,k}$ is bounded by $2^{4m+k}$.
\begin{figure}
\begin{center}
\epsfig{file=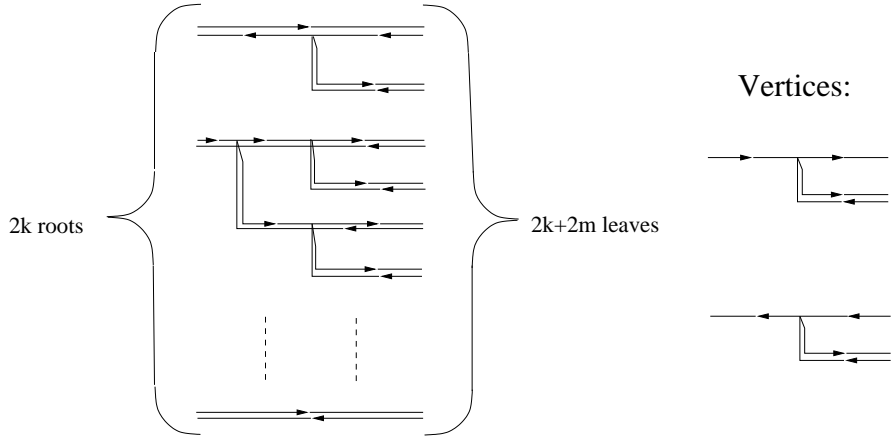,scale=.60}
\end{center}
\caption{A Feynman graph in $\cF_{m,k}$ and its two types of
vertices}\label{fig:feynman}
\end{figure}

\medskip

The particular form of the graphs in $\cF_{m,k}$ is due to the
quantum mechanical nature of the expansion; the presence of a
commutator in the collision operator (\ref{eq:B}) implies that, for
every $B^{(k+j)}$ in (\ref{eq:xi}), we can choose whether to write
the interaction on the left or on the right of the density. When we
draw the corresponding vertex in a graph in $\cF_{m,k}$, we have to
choose whether to attach it on the incoming or on the outgoing edge.

\medskip

Graphs in $\cF_{m,k}$ are characterized by a natural partial
ordering among the vertices ($v \prec v'$ if the vertex $v$ is on
the path from $v'$ to the roots); there is, however, no total
ordering. The absence of total ordering among the vertices is the
consequence of a rearrangement of the summands on the r.h.s. of
(\ref{eq:xi}); by removing the order between times associated with
non-ordered vertices we substantially reduce the number of terms in
the expansion. In fact, while (\ref{eq:xi}) contains $(m+k)!/k!$
summands, in (\ref{eq:expa}) we are only summing over at most $2^{4m+k}$
contributions. The price we have to pay is that the apparent gain of
a factor $1/m!$ due to the ordering of the time integrals in
(\ref{eq:xi}) is lost in the new expansion (\ref{eq:expa}). However,
since we want to use the time integrations to smooth out
singularities it seems quite difficult to make use of this factor $1/m!$.
In fact,
we find that the expansion (\ref{eq:expa}) is better suited for
analyzing the cumulative space-time smoothing effects of the
multiple free evolutions than (\ref{eq:xi}).

\medskip

Because of the pairing of the $2k$ trees, there is a natural pairing
between the $2k$ roots of the graph. Moreover, it is also possible
to define a natural pairing of the leaves of the graph (this is
evident in Figure \ref{fig:feynman}); two leaves $\ell_1$ and
$\ell_2$ are paired if there exists an edge $e_1$ on the path from
$\ell_1$ back to the roots, and an edge $e_2$ on the path from
$\ell_2$ to the roots, such that $e_1$ and $e_2$ are the two
unmarked son-edges of the same vertex (or, in case there is no unmarked
sons in the path from $\ell_1$ and $\ell_2$ to the roots, if the two
roots connected to $\ell_1$ and $\ell_2$ are paired).

\medskip

For $\Lambda \in \cF_{m,k}$, we denote by $E(\Lambda)$,
$V(\Lambda)$, $R(\Lambda)$ and $L(\Lambda)$ the set of all edges,
vertices, roots and, respectively, leaves in the graph $\Lambda$.
For every edge $e \in E(\Lambda)$, we introduce a three-dimensional
momentum variable $p_e$ and a one-dimensional frequency variable
$\a_e$. Then, denoting by $\wh\gamma^{(k+m)}_0$ and by $\wh J^{(k)}$
the kernels of the density $\gamma^{(k+m)}_0$ and of the observable
$J^{(k)}$ in Fourier space, the contribution $K_{\Lambda,t}$ in
(\ref{eq:expa}) is given by
\begin{equation}\label{eq:contriLam}
\begin{split}
K_{\Lambda,t} = &\; \int \prod_{e \in E(\Lambda)} \frac{\rd p_e \rd
\a_e}{ \a_e - p_e^2 + i \tau_e \mu_e} \, \prod_{v \in V(\Lambda)}
\delta \left(\sum_{e \in v} \pm \a_e \right) \delta \left(\sum_{e
\in v} \pm p_e \right)  \\
&\hspace{2cm} \times  \exp \left( -it \sum_{e \in R(\Lambda)} \tau_e
(\a_e +i\tau_e \mu_e) \right) \, \wh J^{(k)} \left( \{ p_e \}_{e \in
R(\Lambda)} \right) \, \wh\gamma^{(k+m)}_0 \left( \{ p_e \}_{e \in
L(\Lambda)} \right)\,.
\end{split}
\end{equation}
Here $\tau_e = \pm 1$, according to the orientation of the edge $e$.
We observe from (\ref{eq:contriLam}) that the momenta of the roots
of $\Lambda$ are the variables of the kernel of $J^{(k)}$, while the
momenta of the leaves of $\Lambda$ are the variables of the kernel
of $\gamma^{(k+m)}_0$ (this also explain why roots and leaves of
$\Lambda$ need to be paired).

\medskip

The denominators $(\a_e - p_e^2 + i \tau_e \mu_e)^{-1}$ are called
propagators; they correspond to the free evolutions in the expansion
(\ref{eq:xi}) and they enter the expression (\ref{eq:contriLam})
through the formula
\[ e^{itp^2} = \int^{\infty}_{-\infty} \rd \a \; \frac{e^{it(\a+i\mu)}}{\a- p^2
+i\mu} \] (here and in (\ref{eq:contriLam}) the measure $\rd \a$ is
defined by $\rd \a = \rd' \a / (2\pi i)$ where $\rd'\a$ is the
Lebesgue measure on $\bR$).

\medskip

The regularization factors $\mu_e$ in (\ref{eq:contriLam}) have to
be chosen such that $\mu_{\text{father}} = \sum_{e=\text{ son}}
\mu_{e}$ at every vertex. The delta-functions in
(\ref{eq:contriLam}) express momentum and frequency conservation
(the sum over $e \in v$ denotes the sum over all edges adjacent to
the vertex $v$; here $\pm \a_e = \a_e$ if the edge points towards
the vertex, while $\pm \a_e=-\a_e$ if the edge points out of the
vertex, and analogously for $\pm p_e$).

\medskip

An analogous expansion can be obtained for the error term
$\eta^{(k)}_{n,t}$ in (\ref{eq:eta}). The problem now is to analyze
the integral (\ref{eq:contriLam}) (and the corresponding integral
for the error term). Through an appropriate choice of the
regularization factors $\mu_e$ one can extract the time dependence
of $K_{\Lambda,t}$ and show that
\begin{equation}\label{eq:bound1}
\begin{split}
|K_{\Lambda,t}| \leq C^{k+m} \; t^{m/4}  \int & \prod_{e \in
E(\Gamma)} \frac{\rd \a_e \rd p_e}{\langle \a_e - p_e^2 \rangle} \;
\prod_{v \in V(\Gamma)} \delta \left( \sum_{e \in v} \pm \a_e
\right) \delta \left( \sum_{e \in v} \pm p_e \right) \, \\ &\times
\Big|\wh J^{(k)} \left( \{ p_e \}_{e \in R(\Gamma)} \right) \Big| \;
\Big| \wh \gamma^{(k+m)}_0 \left( \{ p_e \}_{e\in L
(\Gamma)}\right)\Big|
\end{split}
\end{equation}
where we introduced the notation $\langle x \rangle =
(1+x^2)^{1/2}$.

\medskip

Because of the singularity of the interaction at zero, we may be
faced here with an ultraviolet problem; we have to show that all
integrations in (\ref{eq:bound1}) are finite in the regime of large
momenta and large frequency. Because of (\ref{eq:aprit0}), we know
that the kernel $\wh \gamma^{(k+m)}_0 (\{ p_e \}_{e \in
L(\Lambda)})$ in (\ref{eq:bound1}) provides decay in the momenta of
the leaves. {F}rom (\ref{eq:aprit0}) we have, in momentum space,
\[ \int \rd p_1 \dots \rd p_{n} \; (p_1^2+1) \dots (p_{n}^2 + 1) \,
\wh \gamma^{(n)}_0 (p_1, \dots ,p_{n}; p_1, \dots ,p_{n}) \leq C^n
\] for all $n \geq 1$. Heuristically, this suggests that
\begin{equation}\label{eq:bound2} |\wh \gamma_{0}^{(k+m)} (\{ p_e \}_{e
\in L(\Lambda)})| \lesssim \prod_{e \in L(\Lambda)} \langle p_e
\rangle^{-5/2} \, ,\end{equation} where $\langle p \rangle = (1+p^2)^{1/2}$.
Using this decay in the momenta of
the leaves and the decay of the propagators $\langle \a_e - p_e^2
\rangle^{-1}, e \in E(\Lambda)$, we can prove the finiteness of all
the momentum and frequency integrals in (\ref{eq:contriLam}).
On the heuristic level, this can be seen using a simple power counting
argument. Fix $\kappa \gg 1$, and cutoff all momenta $|p_e| \geq
\kappa$ and all frequencies $|\a_e| \geq \kappa^2$. Each
$p_e$-integral scales then as $\kappa^3$, and each $\a_e$-integral
scales as $\kappa^2$. Since we have $2k+3m$ edges in $\Lambda$, we
have $2k+3m$ momentum- and frequency integrations. However, because
of the $m$ delta functions (due to momentum and frequency
conservation), we effectively only have to perform $2k+2m$ momentum-
and frequency-integrations. Therefore the whole integral in
(\ref{eq:contriLam}) carries a volume factor of the order
$\kappa^{5(2k+2m)} = \kappa^{10k + 10m}$. Now, since there are
$2k+2m$ leaves in the graph $\Lambda$, the estimate
(\ref{eq:bound2}) guarantees a decay of the order $\kappa^{-5/2 (2k
+ 2m)} = \kappa^{-5k-5m}$. The $2k+3m$ propagators, on the other
hand, provide a decay of the order $\kappa^{-2(2k+3m)} =
\kappa^{-4k-6m}$. Choosing the observable $J^{(k)}$ so that $\wh
J^{(k)}$ decays sufficiently fast at infinity, we can also gain an
additional decay $\kappa^{-6k}$. Since
\[ \kappa^{10k+10m} \cdot \kappa^{-5k-5m -4k -6m -6k} =
\kappa^{-m-5k} \ll 1 \] for $\kappa \gg 1$, we can expect
(\ref{eq:contriLam}) to converge in the large momentum and large
frequency regime. Remark the importance of the decay provided by the
free evolution (through the propagators); without making use of it,
we would not be able to prove the uniqueness of the infinite
hierarchy.

\medskip

This heuristic argument is clearly far from rigorous. To obtain a
rigorous proof, we use an integration scheme dictated by the
structure of the graph $\Lambda$; we start by integrating the
momenta and the frequency of the leaves (for which (\ref{eq:bound2})
provides sufficient decay). The point here is that when we perform
the integrations over the momenta of the leaves we have to propagate
the decay to the next edges on the left. We move iteratively from
the right to the left of the graph, until we reach the roots; at
every step we integrate the frequencies and momenta of the son edges
of a fixed vertex and as a result we obtain decay in the momentum of
the father edge. When we reach the roots, we use the decay of the
kernel $\wh J^{(k)}$ to complete the integration scheme. In a
typical step, we consider a vertex as the one drawn in Figure
\ref{fig:2} and we assume to have decay in the momenta of the three
son-edges, in the form $|p_e|^{-\lambda}$, $e=u,d,w$ (for some $2<
\lambda < 5/2$). Then we integrate over the frequencies $\a_u, \a_d,
\a_w$ and the momenta $p_u, p_d, p_w$ of the son-edges and as a
result we obtain a decaying factor $|p_r|^{-\lambda}$ in the
momentum of the father edge.
\begin{figure}
\begin{center}
\epsfig{file=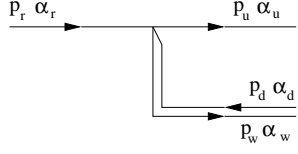,scale=.75}
\end{center}
\caption{Integration scheme: a typical vertex}\label{fig:2}
\end{figure}
In other words, we prove bounds of the form
\begin{equation}\label{eq:scheme}
\int  \frac{\rd \a_u \rd \a_d \rd \a_w \rd p_u \rd p_d \rd
p_w}{|p_u|^{\lambda} |p_d|^{\lambda} |p_w|^{\lambda}} \,
\frac{\delta (\a_r = \a_u +\a_d - \a_w)\delta (p_r = p_u +p_d
-p_w)}{\langle \a_u - p_u^2 \rangle \langle \a_d - p_d^2 \rangle
\langle \a_w -p_w^2 \rangle} \leq \frac{\const}{|p_r|^{\lambda}}\,.
\end{equation}
Power counting implies that (\ref{eq:scheme}) can only be correct if
$\lambda >2$. On the other hand, to start the integration scheme we
need $\lambda < 5/2$ (from (\ref{eq:bound2}) this is the decay in
the momenta of the leaves, obtained from the a-priori estimates). It
turns out that, choosing $\lambda = 2 + \eps$ for a sufficiently
small~$\eps>0$, (\ref{eq:scheme}) can be made precise, and the
integration scheme can be completed.

\medskip

After integrating all the frequency and momentum variables, from
(\ref{eq:bound1}) we obtain that
\[ |K_{\Lambda,t}| \leq C^{k+m} \; t^{m/4} \] for every $\Lambda \in
\cF_{m,k}$. Since the number of diagrams in $\cF_{m,k}$ is bounded
by $C^{k+m}$, it follows immediately that
\[ \left| \tr \; J^{(k)} \, \xi_{m,t}^{(k)} \right| \leq C^{k+m}
t^{m/4} \,. \] Note that, from (\ref{eq:xi}), one may expect
$\xi_{m,t}^{(k)}$ to be proportional to $t^m$. The reason why we
only get a bound proportional to $t^{m/4}$ is that we effectively
use part of the time integration to control the singularity of the
potentials.

\medskip

The only property of $\gamma^{(k+m)}_0$ used in the
analysis of (\ref{eq:contriLam}) is the estimate (\ref{eq:aprit0}),
which provides the necessary decay in the momenta of the leaves.
Since the a-priori bound (\ref{eq:aprit0}) hold uniformly in
time, we can use a similar argument to bound the contribution
arising from the error term $\eta^{(k)}_{n,t}$ in (\ref{eq:eta}) (as
explained above, also $\eta^{(k)}_{n,t}$ can be expanded analogously
to (\ref{eq:expa}), with contributions associated to Feynman graphs
similar to (\ref{eq:contriLam}); the difference, of course, is that
these contributions will depend on $\gamma^{(k+n)}_s$ for all $s \in
[0,t]$, while (\ref{eq:contriLam}) only depends on the initial
data). We get \begin{equation}\label{eq:errorbound}
\left| \tr \; J^{(k)} \, \eta_{n,t}^{(k)} \right| \leq C^{k+n}\;
t^{n/4}\,. \end{equation} This bound immediately implies the
uniqueness. In fact, given two solutions $\Gamma_{1,t} = \{
\gamma_{1,t}^{(k)} \}_{k \geq 1}$ and $\Gamma_{2,t} = \{
\gamma_{2,t}^{(k)} \}_{ k \geq 1}$ of the infinite hierarchy
(\ref{eq:GPhier-int}), both satisfying the a-priori bounds
(\ref{eq:aprit0}) and with the same initial data, we can expand both
in a Duhamel series of order $n$ as in (\ref{eq:Duhamel}). If we fix
$k \geq 1$, and consider the difference between $\gamma^{(k)}_{1,t}$
and $\gamma^{(k)}_{2,t}$, all terms (\ref{eq:xi}) cancel out because
they only depend on the initial data. Therefore, from
(\ref{eq:errorbound}), we immediately obtain that, for arbitrary
(sufficiently smooth) compact $k$-particle operators $J^{(k)}$,
\[ \left|\tr J^{(k)} \left( \gamma_{1,t}^{(k)} - \gamma_{2,t}^{(k)}
\right)\right| \leq 2 \, C^{k+n} \; t^{n/4} \,. \] Since it is
independent of $n$, the left side has to vanish for all $t < 1/(2C)^4$.
This proves uniqueness for short times. But then, since the a-priori
bounds hold uniformly in time, the argument can be repeated to prove
uniqueness for all times.

\subsection{Other Microscopic Models Leading to the Nonlinear Schr\"odinger Equation}

As discussed in Section \ref{sec:comp-mf}, the strategy used to prove Theorem \ref{thm:mainGP} is dictated by the formal similarity with the mean-field systems discussed in Section \ref{sec:mf}; from (\ref{eq:GP-mf}) the Hamiltonian characterizing dilute Bose gases in the Gross-Pitaevskii scaling can be formally interpreted as a mean field Hamiltonian with an $N$-dependent potential converging to a delta-function as $N \to \infty$ (the physics described by the two models is however completely different). The choice of the $N$-dependent potential $V_N (x) = N^2 V (Nx)$ in the Gross-Pitaevskii scaling is, of course, not the only choice for which the formal identification with a mean-field model is possible. For arbitrary $\beta >0$, we can for example define the $N$-particle Hamiltonian
\begin{equation*}
H_{N,\beta} = \sum_{j=1}^N -\Delta_j + \frac{1}{N} \sum_{i<j} N^{3\beta} V (N^{\beta} (x_i -x_j)) \,
\end{equation*}
acting on the Hilbert space $L^2_s (\bR^{3N})$. The Hamiltonian (\ref{eq:hamGP}) is recovered by choosing $\beta = 1$. For $0 < \beta < 1$, the potential $N^{3\beta} V (N^{\beta} x)$ still converges to a delta-function as  $N \to \infty$, but the convergence is slower. This fact has important consequences for the macroscopic dynamics; it turns out, in fact, that for $0< \beta <1$ the correlation structure developed by the evolved wave function $\psi_{N,\beta,t} = e^{-iH_{N,\beta} t}$ varies on much shorter length scales compared with the length scale $N^{-\beta}$ characterizing the potential. Therefore, for $0< \beta <1$, the time evolution of the condensate wave function is still governed by a cubic nonlinear Schr\"odinger equation; this time, however, the coupling constant in front of the nonlinearity is given by $b_0 = \int V$ instead of $8 \pi a_0$ (recall that the emergence of the scattering length in the Gross-Pitaevskii equation was a consequence of the interplay between the correlation structure in the many body wave function and the interaction potential; since, for $0< \beta <1$, the potential and the correlation structure vary on different length scales, this interplay is suppressed). The following theorem can be proven using the techniques developed in \cite{ESY5} (the statement for $0< \beta <1/2$ was proven in \cite{ESY2}; in \cite{ESY3}, the whole range $0< \beta \leq 1$ was covered, but only for sufficiently small potentials). 
\begin{theorem}
Suppose that $V \geq 0$ satisfies the same assumption as in Theorem \ref{thm:mainGP}, and assume that $0 < \beta \leq 1$. Let $\psi_N (\bx) = \prod_{j=1}^N \ph (x_j)$, for some $\ph \in H^1 (\bR^3)$ and $\psi_{N,t}= e^{-iH_{\beta,N} t} \psi_N$ with the
mean-field Hamiltonian
\[ H_{\beta,N} = \sum_{j=1}^N -\Delta_j + \frac{1}{N} \sum_{i<j}^N N^{3\beta}
V(N^{\beta}(x_i -x_j))\,. \]
Then, for every fixed $k \geq 1$ and $t \in \bR$, we have \[ \gamma^{(k)}_{N,t} \to |\ph_t \rangle \langle \ph_t|^{\otimes k} \] as $N \to \infty$, where $\ph_t$ is the solution to the
nonlinear Schr\"odinger equation
\[ i\partial_t \ph_t = -\Delta \ph_t + \sigma |\ph_t|^2 \ph_t \]
with initial data $\ph_{t=0} = \ph$ and with \[ \sigma = \left\{
\begin{array}{ll} 8 \pi a_0 \quad & \text{if } \beta =1 \\
b_0  \quad & \text{if } 0 < \beta < 1
\end{array} \right. \, .\]
\end{theorem}

\section{Rate of Convergence towards Mean-Field Dynamics}
\setcounter{equation}{0}\label{sec:RS}

{F}rom the results of Section \ref{sec:mf}, we obtain that, for every fixed $t \in \bR$, and for every fixed $k \in \bN$, \[ \gamma^{(k)}_{N,t} \to |\ph_t \rangle \langle \ph_t|^{\otimes k} \] where $\gamma^{(k)}_{N,t}$ is the $k$-particle density associated with the solution $\psi_{N,t}$ of the $N$-particle Schr\"odinger equation, and $\ph_t$ is the solution of the Hartree equation. {F}rom Theorem \ref{thm:bd} and Theorem \ref{thm:cou}, we do not get any information about the rate of convergence of $\gamma^{(k)}_{N,t}$ to $|\ph_t \rangle \langle \ph_t|^{\otimes k}$. We only know that the difference $\gamma^{(k)}_{N,t} - |\ph_t \rangle \langle \ph_t|^{\otimes k}$ converges to zero, but we do not know how fast. Also in Section \ref{sec:GP}, we do not obtain any information about the rate of convergence of the $N$-particle Schr\"odinger evolution towards the Gross-Pitaevskii equation. This is not only a question of academic interest; in order to apply these results to physically relevant situations, bounds on the error are essential.

\medskip

For bounded potential, (\ref{eq:diffbd2}) implies (specializing to the case $k=1$) that
\begin{equation}\label{eq:err} \tr \; \left| \gamma^{(1)}_{N,t} - |\ph_t \rangle \langle \ph_t | \right| \leq C N^{-1} \end{equation}
for sufficiently short times $0 \leq t \leq t_0 = 1/(8\| V \|)$. Iterating the argument leading to (\ref{eq:err}) to obtain estimates valid for larger times, it is possible to derive bounds of the form
\begin{equation}\label{eq:errexp} \tr \; \left| \gamma^{(1)}_{N,t} - |\ph_t \rangle \langle \ph_t | \right| \leq \frac{C}{N^{2^{-t}}} \, . \end{equation} Although (\ref{eq:errexp}) shows that, for every fixed $t>0$, the difference $\gamma^{(1)}_{N,t} - |\ph_t \rangle \langle \ph_t|$ converges to zero, it is not very useful in applications because it deteriorates too fast; one would like to find bounds on the difference $\gamma^{(1)}_{N,t} - |\ph_t \rangle \langle \ph_t |$ which are of the same order in $N$ for every fixed time.

\medskip

In \cite{RS}, a joint work with I. Rodnianski, we obtain such bounds for mean-field systems with potential having at most a Coulomb singularity; the problem of obtaining error estimates for the Gross-Pitaevskii dynamics is still open. To prove such bounds, we do not make use of the BBGKY hierarchy. Instead, we use an approach, introduced by Hepp in \cite{Hepp} and extended by Ginibre and Velo in \cite{GV1}, based on embedding the $N$-body Schr\"odinger system into the second quantized Fock-space representation and on the use of coherent states as initial data (coherent states do not have a fixed number of particles; this is what makes the use of a Fock-space representation necessary). The Hartree dynamics emerges as the main component of the evolution of coherent states in the mean field limit. To obtain bounds on the difference $\gamma^{(1)}_{N,t} - |\ph_t \rangle \langle \ph_t |$  for initial data describing coherent states, it is therefore enough to study the fluctuations around the Hartree dynamics, and to prove that, in an appropriate sense, they are small. Since factorized $N$-particle wave functions can be written as appropriate linear combinations of coherent states, the estimates for coherent initial data can be translated into bounds for factorized initial data. Using these techniques, we prove in \cite{RS} the following theorem. We focus in this section on three dimensional systems, which are the most interesting from the point of view of physics; most of the results however can be extended to dimension $d \neq 3$.

\begin{theorem}\label{thm:fact}
Suppose that there exists $D>0$ such that the operator inequality
\begin{equation}\label{eq:assump} V^2 (x) \leq D \; (1-\Delta_x)  \end{equation} holds true.
Let
\begin{equation*}
\psi_N (\bx) = \prod_{j=1}^N \ph (x_j),\end{equation*} for some $\ph
\in H^1 (\bR^3)$ with $\| \ph \| =1$. Denote by $\psi_{N,t} = e^{-i
H_N t} \psi_N$ the solution to the Schr\"odinger equation
(\ref{eq:schr}) with initial data $\psi_{N,0}= \psi_N$, and let
$\gamma^{(1)}_{N,t}$ be the one-particle density associated with
$\psi_{N,t}$. Then there exist constants $C,K$, depending only on
the $H^1$ norm of $\ph$ and on the constant $D$ on the r.h.s. of
(\ref{eq:assump}) such that
\begin{equation}\label{eq:RS-fact}
\tr\; \Big| \gamma^{(1)}_{N,t} - |\ph_t \rangle \langle \ph_t| \Big|
\leq \frac{C}{N^{1/2}} \; e^{K |t|} \, ,
\end{equation}
for every $t \in \bR$ and every $N \in \bN$. Here $\ph_t$ is the solution to the nonlinear Hartree equation
\begin{equation}\label{eq:hartree2}
i\partial_t \ph_t = -\Delta \ph_t + (V *|\ph_t|^2 ) \ph_t
\end{equation}
with initial data $\ph_{t=0} = \ph$.
\end{theorem}
{\it Remarks.} Condition (\ref{eq:assump}) is in particular satisfied by bounded potentials and by potentials with an attractive or repulsive Coulomb singularity. Theorem \ref{thm:fact} implies therefore Theorem \ref{thm:bd} and Theorem \ref{thm:cou}. Note, moreover, that the decay of the order $N^{-1/2}$ on the r.h.s. of (\ref{eq:RS-fact}) is not expected to be optimal. In fact, for initially coherent states we obtain in Theorem \ref{thm:coh} the expected decay of the order $1/N$ for every fixed time $t \in \bR$; unfortunately, when factorized initial data are expressed as a superposition of coherent states, part of the decay is lost (note, however, that for a certain class of bounded potential a decay of the order $N^{-1}$ for factorized initial data has recently been established in \cite{ErS}).

\subsection{Fock-Space Representation}

We define the bosonic Fock space over $L^2 (\bR^3, \rd x)$ as the
Hilbert space
\[ \cF = \bigoplus_{n \geq 0} L^2 (\bR^3 , \rd x)^{\otimes_s n} =
\bC \oplus \bigoplus_{n \geq 1} L^2_s (\bR^{3n}, \rd x_1 \dots \rd
x_n)\, ,
\] where we put $L^2 (\bR^3)^{\otimes_s 0} = \bC$.
Vectors in $\cF$ are sequences $\psi = \{ \psi^{(n)} \}_{n \geq 0}$
of $n$-particle wave functions $\psi^{(n)} \in L^2_s (\bR^{3n})$.
The scalar product on $\cF$ is defined by
\[ \langle \psi_1 , \psi_2 \rangle = \sum_{n \geq 0} \langle
\psi_1^{(n)} , \psi_2^{(n)} \rangle_{L^2 (\bR^{3n})} =
\overline{\psi_1^{(0)}} \psi_2^{(0)} + \sum_{n \geq 1} \int \rd x_1
\dots \rd x_n \, \overline{\psi_1^{(n)}} (x_1 , \dots , x_n)
\psi_2^{(n)} (x_1, \dots ,x_n) \,. \] An $N$ particle state with
wave function $\psi_N$ is described on $\cF$ by the sequence $\{
\psi^{(n)} \}_{ n \geq 0}$ where $\psi^{(n)} = 0$ for all $n \neq N$
and $\psi^{(N)} = \psi_N$. The vector $\{1, 0, 0, \dots \} \in \cF$
is called the vacuum, and will be denoted by $\Omega$.

\medskip

On $\cF$, we define the number of particles operator $\cN$, by $(\cN
\psi)^{(n)} = n \psi^{(n)}$. Eigenvectors of $\cN$ are vectors of
the form $\{ 0, \dots, 0, \psi^{(m)}, 0,  \dots \}$ with a fixed
number of particles $m$. For $f \in L^2 (\bR^3)$ we also define the
creation operator $a^* (f)$ and the annihilation operator $a(f)$ on
$\cF$ by
\begin{equation*}
\begin{split}
\left(a^* (f) \psi \right)^{(n)} (x_1 , \dots ,x_n) &=
\frac{1}{\sqrt n} \sum_{j=1}^n f(x_j) \psi^{(n-1)} ( x_1, \dots,
x_{j-1}, x_{j+1},
\dots , x_n) \\
\left(a (f) \psi \right)^{(n)} (x_1 , \dots ,x_n) &= \sqrt{n+1} \int
\rd x \; \overline{f (x)} \, \psi^{(n+1)} (x, x_1, \dots ,x_n) \, .
\end{split}
\end{equation*}
The operators $a^* (f)$ and $a(f)$ are unbounded, densely defined and
closed; $a^* (f)$ creates a particle with wave function $f$, $a(f)$ annihilates it. It is simple to check that, for arbitrary $n \geq 1$, \[ \frac{ \left( a^*(f) \right)^n}{\sqrt{n!}} \,\Omega = \{ 0, \dots ,0, f^{\otimes n} , 0, \dots \} \,.\]
The creation operator $a^*(f)$ is the adjoint of
the annihilation operator $a(f)$ (note that by definition $a(f)$ is
anti-linear in $f$), and they satisfy the canonical commutation
relations \begin{equation}\label{eq:comm} [ a(f) , a^* (g) ] =
\langle f,g \rangle_{L^2 (\bR^3)}, \qquad [ a(f) , a(g)] = [ a^*
(f), a^* (g) ] = 0 \,. \end{equation} For every $f\in L^2 (\bR^3)$,
we introduce the self adjoint operator
\[ \phi (f) = a^* (f) + a(f) \,. \]
We will also make use of operator valued distributions $a^*_x$ and
$a_x$ ($x \in \bR^3$), defined so that \begin{equation*}\begin{split}
a^* (f) &= \int \rd x \, f(x) \, a_x^* \\ a(f) & = \int \rd x \,
\overline{f (x)} \, a_x \end{split}
\end{equation*}
for every $f \in L^2 (\bR^3)$. The canonical commutation relations
take the form \[ [ a_x , a^*_y ] = \delta (x-y) \qquad [ a_x, a_y
] = [ a^*_x , a^*_y] = 0 \, .\]
The number of particle operator, expressed through the distributions
$a_x,a^*_x$, is given by
\[ \cN = \int \rd x \, a_x^* a_x \,. \]

The following lemma provides some useful bounds to control creation
and annihilation operators in terms of the number of particle
operator $\cN$.
\begin{lemma}\label{lm:a-bd}
Let $f \in L^2 (\bR^3)$. Then
\begin{equation*}
\begin{split}
\| a(f) \psi \| &\leq \| f \| \, \| \cN^{1/2} \psi \| \\
\| a^* (f) \psi \| &\leq \| f \| \, \| \left( \cN + 1 \right)^{1/2}
\psi \| \\
\| \phi (f) \psi \| &\leq 2 \| f \| \| \left( \cN + 1 \right)^{1/2}
\psi \|\, .
\end{split}
\end{equation*}
\end{lemma}
\begin{proof}
The last inequality clearly follows from the first two. To prove the
first bound we note that
\begin{equation*}
\begin{split} \| a (f) \psi \| &\leq \int \rd x \, |f(x)| \, \| a_x
\psi \| \leq \left( \int \rd x \, |f(x)|^2 \right)^{1/2} \, \left(
\int \rd x \, \| a_x \psi \|^2 \right)^{1/2} \\ &= \| f \| \, \|
\cN^{1/2} \psi \|\, .
\end{split}
\end{equation*}
The second estimate follows by the canonical commutation relations
(\ref{eq:comm}) because
\begin{equation*}
\begin{split}
\| a^* (f) \psi \|^2 &= \langle \psi, a(f) a^* (f) \psi \rangle =
\langle \psi, a^* (f) a(f) \psi \rangle + \| f \|^2 \| \psi \|^2  \\
&= \| a(f) \psi \|^2 + \| f\|^2 \| \psi \|^2 \leq \| f\|^2 \, \left(
\|\cN^{1/2} \psi \| + \| \psi \|^2 \right) = \| f \|^2 \| \left( \cN
+1 \right)^{1/2} \psi \|^2 \, .
\end{split}
\end{equation*}
\end{proof}

Given $\psi \in \cF$, we define the one-particle density
$\gamma^{(1)}_{\psi}$ associated with $\psi$ as the
operator on $L^2 (\bR^3)$ with kernel given by
\begin{equation}\label{eq:margi} \gamma^{(1)}_{\psi} (x; y) = \frac{1}{\langle \psi,
\cN \psi \rangle} \, \langle \psi, a_y^* a_x \psi \rangle\, .
\end{equation} By definition, $\gamma_{\psi}^{(1)}$ is a positive trace
class operator on $L^2 (\bR^3)$ with $\tr \, \gamma_{\psi}^{(1)}
=1$. For every $N$-particle state with wave function $\psi_N \in
L^2_s (\bR^{3N})$ (described on $\cF$ by the sequence $\{ 0, 0,
\dots, \psi_N, 0,0, \dots \}$) it is simple to see that this
definition is equivalent to the definition (\ref{eq:gammakN}).

\medskip

We define the Hamiltonian $\cH_N$ on $\cF$ by $ (\cH_N \psi)^{(n)} =
\cH^{(n)}_N \psi^{(n)}$, with
\[ \cH^{(n)}_N = - \sum_{j=1}^n \Delta_j + \frac{1}{N} \sum_{i<j}^n
V(x_i -x_j) \, . \] Using the distributions $a_x, a^*_x$, $\cH_N$
can be rewritten as
\begin{equation}\label{eq:ham2} \cH_N = \int \rd x \nabla_x a^*_x
\nabla_x a_x + \frac{1}{2N} \int \rd x \rd y \, V(x-y) a_x^* a_y^*
a_y a_x \, . \end{equation} By definition the Hamiltonian $\cH_N$
leaves sectors of $\cF$ with a fixed number of particles invariant.
Moreover, it is clear that on the $N$-particle sector, $\cH_N$
agrees with the Hamiltonian $H_N$ (the subscript $N$ in $\cH_N$ is a
reminder of the scaling factor $1/N$ in front of the potential
energy). We will study the dynamics generated by the operator
$\cH_N$. In particular we will consider the time evolution of
coherent states, which we introduce next.

\medskip

For $f \in L^2 (\bR^3)$, we define the Weyl-operator
\begin{equation*}
W(f) = \exp \left( a^* (f) - a(f) \right) = \exp \left( \int \rd x
\, (f(x) a^*_x - \overline{f} (x) a_x) \right) \, .
\end{equation*}
Then the coherent state $\psi (f) \in \cF$ with one-particle wave
function $f$ is defined by
\[ \psi (f) = W(f) \Omega \, .\]
Notice that \begin{equation}\label{eq:coh} \psi (f)= W(f) \Omega =
e^{-\| f\|^2 /2} \sum_{n \geq 0} \frac{ (a^* (f))^n}{n!} \Omega  =
e^{-\| f\|^2 /2} \sum_{n \geq 0} \frac{1}{\sqrt{n!}} \, f^{\otimes n} \,,
\end{equation}
where $f^{\otimes n}$ indicates the Fock-vector $\{ 0, \dots , 0 ,f^{\otimes n}, 0, \dots \}$. This follows from
\[ \exp (a^* (f) - a (f)) = e^{-\|f \|^2/2} \exp (a^* (f)) \exp
(-a(f)) \] which is a consequence of the fact that the commutator $[
a (f) , a^* (f)] = \| f \|^2$ commutes with $a(f)$ and $a^* (f)$.
{F}rom (\ref{eq:coh}), it follows that coherent states are
superpositions of states with different number of particles (the
probability of having $n$ particles in $\psi (f)$ is given by
$e^{-\| f\|^2} \| f \|^{2n}/n!$).

\medskip

In the following lemma we collect some important and well known
properties of Weyl operators and coherent states.
\begin{lemma}\label{lm:coh}
Let $f,g \in L^2 (\bR^3)$.
\begin{itemize}
\item[i)] The Weyl operator satisfy the relations
\[ W(f) W(g) = W(g) W(f) e^{-2i \, \text{Im} \, \langle f,g \rangle} = W(f+g) e^{-i\, \text{Im} \, \langle f,g \rangle} \,. \]
\item[ii)] $W(f)$ is a unitary operator and
\[ W(f)^* = W(f)^{-1}  = W (-f). \]
\item[iii)] For every $x \in \bR^3$, we have \[ W^* (f) a_x W(f) = a_x + f(x), \qquad \text{and} \quad W^* (f) a^*_x
W(f) = a^*_x + \overline{f} (x) \, .\]
\item[iv)] {F}rom iii) it follows that coherent states are eigenvectors of annihilation operators
\[ a_x \psi (f) = f(x) \psi (f)  \qquad \Rightarrow \qquad a (g)
\psi (f) = \langle g, f \rangle_{L^2} \psi (f) \, .\]
\item[v)] The expectation of the number of particles in the coherent
state $\psi (f)$ is given by $\| f\|^2$, that is
\[ \langle \psi (f), \cN \psi (f) \rangle = \| f \|^2
\, . \] Also the variance of the number of particles in $\psi (f)$
is given by $\|f \|^2$ (the distribution of $\cN$ is Poisson), that
is
\[ \langle \psi (f), \cN^2 \psi (f) \rangle - \langle \psi (f) ,
\cN \psi (f) \rangle^2 = \| f \|^2 \, .\]
\item[vi)] Coherent states are normalized but not orthogonal to each
other. In fact
\[ \langle \psi (f) , \psi (g) \rangle = e^{-\frac{1}{2}\left( \| f
\|^2 + \| g \|^2 - 2 (f,g) \right)}  \quad \Rightarrow \quad
|\langle \psi (f) , \psi (g) \rangle| = e^{-\frac{1}{2} \| f- g
\|^2} \, .\]
\end{itemize}
\end{lemma}

\subsection{Time Evolution of Coherent States}\label{sec:coh}

Next we study the dynamics of coherent states with expected number
of particles $N$ in the limit $N \to \infty$. We choose the initial
data
\begin{equation*} \psi (\sqrt{N} \ph) = W(\sqrt{N} \ph) \Omega \qquad
\text{for $\ph \in H^1 (\bR^3)$ with } \| \ph \| =1 \end{equation*}
and we consider its time evolution $\psi (N,t) = e^{-i \cH_N t} \psi
(\sqrt{N} \ph)$ with the Hamiltonian $\cH_N$ defined in
(\ref{eq:ham2}).

\begin{theorem}\label{thm:coh}
Suppose that there exists $D >0$ such that the operator inequality
\begin{equation}\label{eq:assump-coh} V^2 (x) \leq D (1 - \Delta_x) \end{equation}
holds true. Let $\Gamma_{N,t}^{(1)}$ be the one-particle marginal
associated with $\psi(N,t)= e^{-i\cH_N t} W(\sqrt{N} \ph) \Omega$ (as defined
in (\ref{eq:margi})). Then there exist constants $C,K>0$ (only
depending on the $H^1$-norm of $\ph$ and on the constant $D$
appearing in (\ref{eq:assump-coh})) such that
\begin{equation}\label{eq:tr-1}
\tr \; \Big|\Gamma^{(1)}_{N,t} - |\ph_t \rangle \langle \ph_t| \Big|
\leq \frac{C}{N} \; e^{K |t|}
\end{equation}
for all $t \in \bR$.
\end{theorem}

We explain next the main steps in the proof of Theorem \ref{thm:coh}.
By (\ref{eq:margi}), the kernel of $\Gamma^{(1)}_{N,t}$ is given by
\begin{equation}\label{eq:gamma}
\begin{split}
\Gamma^{(1)}_{N,t} (x;y) = &\frac{1}{N} \left\langle \Omega, W^*
(\sqrt{N} \ph) e^{i\cH_N t} a_y^* a_x e^{-i\cH_N t} W(\sqrt{N} \ph)
\Omega \right\rangle \\ = \; &\ph_t (x) \overline{\ph}_t (y) +
\frac{\overline{\ph}_t (y)}{\sqrt{N}} \left\langle \Omega, W^*
(\sqrt{N} \ph) e^{i\cH_N t} (a_x - \sqrt{N} \ph_t(x)) e^{-i\cH_N t}
W(\sqrt{N} \ph) \Omega \right\rangle \\ &+ \frac{\ph_t
(x)}{\sqrt{N}} \left\langle \Omega, W^* (\sqrt{N} \ph) e^{i\cH_N t}
(a_y^* - \sqrt{N} \overline{\ph}_t (y)) e^{-i\cH_N t} W(\sqrt{N}
\ph) \Omega \right\rangle \\ &+ \frac{1}{N} \left\langle \Omega, W^*
(\sqrt{N} \ph) e^{i\cH_N t} ( a_y^* - \sqrt{N} \overline{\ph}_t (y))
(a_x - \sqrt{N} \ph_t (x)) e^{-i \cH_N t} W(\sqrt{N} \ph ) \Omega
\right\rangle \,.
\end{split}
\end{equation}
It was observed by Hepp in \cite{Hepp} (see also
Eqs. (1.17)-(1.28) in \cite{GV1}) that
\begin{equation}\label{eq:GV}
\begin{split} W^* (\sqrt{N} \ph_s) \; e^{i\cH_N (t-s)} (a_x - \sqrt{N} \ph_t
(x)) e^{-i\cH_N (t-s)} W (\sqrt{N} \ph_s) &= \cU_N  (t;s)^* \, a_x \,
\cU_N (t;s) \\ &= \cU_N (s;t) \, a_x \, \cU_N (t;s)
\end{split}\end{equation} where the unitary evolution $\cU_N (t;s)$
is determined by the equation
\begin{equation}\label{eq:cUN} i
\partial_t \cU_N (t;s) = \cL_N (t) \cU_N (t;s) \qquad \text{and}
\quad \cU_N (s;s)= 1 \end{equation} with the generator
\begin{equation}\label{eq:cLN}
\begin{split}
\cL_N (t) = & \int \rd x \, \nabla_x a^*_x \nabla_x a_x + \int \rd x
\, \left(V *|\ph_t|^2 \right) (x) \, a^*_x a_x +
\int \rd x \rd y \, V(x-y) \, \overline{\ph_t} (x) \ph_t (y) a^*_y a_x \\
&+ \frac{1}{2} \int \rd x \rd y \, V(x-y) \left( \ph_t (x) \ph_t (y)
a^*_x a^*_y +
\overline{\ph_t} (x) \overline{\ph_t} (y) a_x a_y \right) \\
&+\frac{1}{\sqrt{N}} \int \rd x \rd y \, V(x-y) \, a_x^* \left(
\ph_t (y)
a^*_y  + \overline{\ph_t} (y) a_y \right) a_x \\
&+\frac{1}{2N} \int \rd x \rd y \, V(x-y) \, a^*_x a^*_y a_y a_x \, .
\end{split}
\end{equation}
It follows from (\ref{eq:gamma}) that
\begin{equation}\label{eq:gamma2cou}
\begin{split}
\Gamma^{(1)}_{N,t} (x,y) - \ph_t (x) \overline{\ph}_t (y) = \;
&\frac{1}{N} \left\langle \Omega, \cU_N (t;0)^* a_y^* a_x \cU_N
(t;0) \Omega \right\rangle \\ &+ \frac{\ph_t (x)}{\sqrt{N}}
\left\langle \Omega,\cU_N (t;0)^* a^*_y \cU_N (t;0) \Omega \right\rangle \\
&+ \frac{\overline{\ph}_t (y)}{\sqrt{N}} \left\langle \Omega,\cU_N
(t;0)^* a_x \cU_N (t;0) \Omega \right\rangle\,.
\end{split}
\end{equation}
In order to produce another decaying factor $1/\sqrt{N}$ in the last
two term on the r.h.s. of the last equation, we compare the
evolution $\cU_N (t;0)$ with another evolution $\wt\cU_N (t;0)$
defined through the equation
\begin{equation}\label{eq:wtcUN} i
\partial_t \wt \cU_N (t;s) = \wt\cL_N (t)\, \wt\cU_N (t;s) \qquad
\text{with} \quad \wt\cU_N (s;s) = 1 \end{equation} and
\begin{equation*}
\begin{split}
\wt \cL_N (t) = & \int \rd x \, \nabla_x a^*_x \nabla_x a_x + \int
\rd x \, \left(V *|\ph_t|^2 \right) (x) \, a^*_x a_x + \int \rd x
\rd y \, V(x-y) \overline{\ph_t} (x) \ph_t (y) a^*_y a_x \\
&+ \frac{1}{2} \int \rd x \rd y \, V(x-y) \left( \ph_t (x) \ph_t (y)
a^*_x a^*_y +
\overline{\ph_t} (x) \overline{\ph_t} (y) a_x a_y \right) \\
&+\frac{1}{2N} \int \rd x \rd y \, V(x-y) \, a^*_x a^*_y a_y a_x \, .
\end{split}
\end{equation*}
{F}rom (\ref{eq:gamma2cou}) we find
\begin{equation}\label{eq:gamma3cou}
\begin{split}
&\Gamma^{(1)}_{N,t} (x;y) - \ph_t (x) \overline{\ph}_t (y) \\ & =
\frac{1}{N} \langle \Omega, \cU_N (t;0)^* a_y^* a_x \cU_N (t;0)
\Omega \rangle \\ &\;\; + \frac{\ph_t (x)}{\sqrt{N}} \left(
\left\langle \Omega,\cU_N (t;0)^* a_y^* \left(\cU_N (t;0) - \wt
\cU_N (t;0)\right) \Omega \right\rangle + \left\langle \Omega,
\left(\cU_N (t;0)^* - \wt \cU_N (t;0)^* \right) a_y^* \wt \cU_N
(t;0) \Omega
\right\rangle \right)\\
&\;\; + \frac{\overline{\ph}_t (y)}{\sqrt{N}} \left( \left\langle
\Omega,\cU_N (t;0)^* a_x \left( \cU_N (t;0) - \wt \cU_N (t;0)
\right) \Omega \right\rangle + \left\langle \Omega, \left(\cU_N
(t;0)^* -\wt \cU_N (t;0)^* \right) a_x \wt \cU_N (t;0) \Omega
\right\rangle \right),
\end{split}
\end{equation}
because
\[ \left\langle \Omega, \wt \cU_N (t;0)^*  a_y \, \wt \cU_N (t;0) \Omega
\right\rangle = \left\langle \Omega,\wt \cU_N (t;0)^* a_x^*\,  \wt
\cU_N (t;0) \Omega \right\rangle = 0 \, .\] This follows from the
observation that, although the evolution $\wt\cU_N (t)$ does not
preserve the number of particles, it preserves the parity (it
commutes with $(-1)^{\cN}$). {F}rom (\ref{eq:gamma3cou}), it easily follows that
\begin{equation}\label{eq:fin1}
\begin{split}
\Big\| \Gamma^{(1)}_{N,t} - |\ph_t \rangle
\langle \ph_t| \Big\|_{\text{HS}}
\leq \; &\frac{1}{N} \; \langle \cU_N (t;0) \Omega, \cN \cU_N (t;0) \Omega \rangle \\
&+ \frac{2}{\sqrt{N}} \| (\cU_N (t;0) - \wt
\cU_N (t;0)) \Omega \| \, \| (\cN+1)^{1/2} \cU_N (t;0) \Omega \| \\
&+ \frac{2}{\sqrt{N}} \|(\cU_N (t;0) - \wt \cU_N
(t;0))\Omega\| \, \| (\cN+1)^{1/2} \wt \cU_N (t;0) \Omega \|\,.
\end{split}
\end{equation}

To bound the r.h.s. of (\ref{eq:fin1}), we need to compare the dynamics $\cU_N (t;0)$ and $\wt \cU_N (t;0)$, and to control the growth of the number of particle $\cN$ with respect to the fluctuation dynamics $\cU_N (t;0)$ and $\wt \cU_N (t;0)$. We show, first of all, that
\begin{equation}\label{eq:lm0cou}
\langle \wt \cU_N (t;0) \, \Omega, \cN \, \wt \cU_N
(t;0) \Omega \rangle \leq C \, e^{K |t|} \,. \end{equation}
To prove this bound, we compute the time derivative
\begin{equation*}
\begin{split}
\frac{\rd}{\rd t} \langle \wt \cU_N (t;0) \, \Omega, (\cN+1) \, \wt \cU_N
(t;0) \Omega \rangle = \; & \langle \wt \cU_N (t;0) \Omega, [ \wt \cL_N (t) , \cN  ]
\wt \cU_N (t;0) \Omega \rangle \\ = \; & \langle \wt \cU_N (t;0) \Omega,
[ \wt \cL_N (t) , \cN ] \wt \cU_N (t;0) \Omega \rangle \\
= \; &2\text{Im} \int \rd x \rd y V(x-y) \ph_t (x) \ph_t (y) \langle \wt
\cU_N (t;0) \Omega, [ a_x^* a_y^* , \cN ] \wt \cU_N (t;0)
\Omega \rangle
\\ = \; &4\text{Im} \int
\rd x \rd y V(x-y) \ph_t (x) \ph_t (y) \langle \wt \cU_N (t;0)
\Omega,  a_x^* a_y^*  \wt \cU_N (t;0) \Omega \rangle\,.
\end{split}
\end{equation*}
Thus, from Lemma \ref{lm:a-bd}, we obtain
\begin{equation}\label{eq:gron}
\begin{split}
\Big| \frac{\rd}{\rd t} \langle \wt \cU_N (t;0) &\Omega, (\cN+1)
\wt \cU_N (t;0) \Omega \rangle \Big| \\ \leq \; &4 \int \rd x
|\ph_t (x)| \| a_x \wt \cU_N (t;0) \Omega \| \, \| a^*
(V(x-.)\ph_t) \wt \cU_N (t;0) \Omega \| \\ \leq \; &4 \sup_x \| V(x-.) \ph_t \|
\| (\cN+1)^{1/2} \wt \cU_N (t;0)\Omega \|^2 \\ \leq \; &C \, \langle
\wt \cU_N (t;0) \Omega, (\cN+1) \wt \cU_N (t;0) \Omega \rangle\, ,
\end{split}
\end{equation}
where we used the fact that
\[ \| V (x-. ) \ph_t \|^2 = \int \rd y \, |V(x-y)|^2 |\ph_t (y)|^2\leq C \| \ph_t \|^2_{H_1} \leq C \| \ph \|_{H^1}^2 \]
because of the assumption (\ref{eq:assump-coh}). {F}rom (\ref{eq:gron}), we obtain (\ref{eq:lm0cou}) applying Gronwall's Lemma.

\medskip

Making use of (\ref{eq:lm0cou}) (and of an analogous bound for the growth of the expectation of $\cN^4$ w.r.t. the evolution $\wt\cU_N (t;0)$; see \cite{RS}[Lemma 3.7]), we can derive the bound
\begin{equation}\label{eq:lm3cou} \left\| \left(\cU_N (t;0) - \wt \cU_N
(t;0) \right) \Omega \right\| \leq \frac{C}{\sqrt{N}} \, e^{K |t|} \,
\end{equation}
for the difference between the two time evolutions $\cU_N (t;0)$ and $\wt \cU_N (t;0)$ (note that, at least formally, the difference between the two generators $\cL_N (t)$ and $\wt \cL_N (t)$ is a term of the order $N^{-1/2}$; this explains the decay in $N$ on the r.h.s. of (\ref{eq:lm3cou})).

\medskip

In \cite{Hepp,GV1} the time evolution $\cU (t;s)$ was proven to converge, as $N \to \infty$, to a limiting dynamics $\cU_{\infty} (t;s)$ defined by \[ i \partial_t \cU_{\infty} (t;s) = \cL_{\infty} (t) \cU_{\infty} (t;s) \qquad \text{and } \quad \cU_{\infty} (s;s) = 1 \] with generator \begin{equation*}
\begin{split}
\cL_{\infty} (t) = & \int \rd x \, \nabla_x a^*_x \nabla_x a_x + \int
\rd x \, \left(V *|\ph_t|^2 \right) (x) \, a^*_x a_x + \int \rd x
\rd y \, V(x-y) \overline{\ph_t} (x) \ph_t (y) a^*_y a_x \\
&+ \frac{1}{2} \int \rd x \rd y \, V(x-y) \left( \ph_t (x) \ph_t (y)
a^*_x a^*_y + \overline{\ph_t} (x) \overline{\ph_t} (y) a_x a_y \right) \,.
\end{split}
\end{equation*}
In this sense, Hepp (in \cite{Hepp}, for smooth potentials) and Ginibre-Velo (in \cite{GV1}, for singular potentials) were able to identify the limiting time evolution of the fluctuations around the Hartree dynamics. Analogously to (\ref{eq:lm3cou}), the bound  (\ref{eq:lm0cou}) can also be used to show that, for a dense set of $\Psi \in \cF$, there exists constants $C,K$ such that $\| ( \cU_N (t;s) - \cU_{\infty} (t;s)) \Psi \| \leq C N^{-1/2} e^{K|t-s|}$, giving therefore a quantitative control on the convergence established in \cite{Hepp,GV1}.

\medskip

Note, however, that the convergence of $\cU_N (t;s)$ to the limiting dynamics $\cU_{\infty} (t;s)$ is still not enough to obtain estimates on the difference between $\Gamma^{(1)}_{N,t}$ and $|\ph_t \rangle \langle \ph_t|$. In fact, to reach this goal, we still need, by (\ref{eq:fin1}), to control the growth of $\cN$ with respect to the time evolution $\cU_N (t;s)$. We are going to prove that
\begin{equation}\label{eq:lm1cou}
\langle \cU_N (t;0)\, \Omega, \cN \, \cU_N (t;0) \Omega \rangle \leq C \, e^{K |t|}\,.
\end{equation}
Inserting (\ref{eq:lm0cou}), (\ref{eq:lm3cou}) and (\ref{eq:lm1cou}) on the r.h.s. of (\ref{eq:fin1}), it follows immediately that
\begin{equation}\label{eq:HS1} \Big\| \Gamma^{(1)}_{N,t} - |\ph_t \rangle
\langle \ph_t| \Big\|_{\text{HS}} \leq C \frac{e^{K|t|}}{N} \, , \end{equation}
which implies the claim (\ref{eq:tr-1}). In fact, since $|\ph_t \rangle \langle \ph_t|$ is a rank one projection, the operator $\delta \gamma = \gamma^{(1)}_{N,t} - |\ph_t \rangle \langle \ph_t|$ has at most one negative eigenvalue. Noticing that $\tr \, \delta \gamma = 0$, it follows that $\delta \gamma$ has a negative eigenvalue, and that the negative eigenvalue must equal, in absolute value, the sum of all its positive eigenvalues. Therefore, the trace norm of $\delta \gamma$ is twice as large as the operator norm of $\delta \gamma$. Since the operator norm is always controlled by the Hilbert Schmidt norm, we obtain (\ref{eq:tr-1}) (this nice argument was pointed out to us by R. Seiringer).

\medskip

The proof of (\ref{eq:lm1cou}) is much more involved than the proof of the analogous bound (\ref{eq:lm0cou}). This is a consequence of the presence, in the generator $\cL_N (t)$, of terms which are cubic in the creation and annihilation operators (these terms are absent from $\wt \cL_N (t)$). Because of these terms, also the commutator $[\cL_N (t) , \cN]$ is cubic in creation and annihilation operators, and thus its expectation (in absolute value) cannot be controlled by the expectation of $\cN$. For this reason, to prove (\ref{eq:lm1cou}), we have to introduce yet another approximate dynamics $\cW_N (t;s)$, defined by
\[ i\partial_t \cW_N (t;s) = \cK_N (t) \cW_N (t;s), \qquad \text{with } \quad \cW_N (s;s) = 1 \]
and with generator
\begin{equation}\label{eq:cKN}
\begin{split}
\cK_N (t) = & \int \rd x \, \nabla_x a^*_x \nabla_x a_x + \int \rd x
\, \left(V *|\ph_t|^2 \right) (x) \, a^*_x a_x +
\int \rd x \rd y \, V(x-y) \, \overline{\ph_t} (x) \ph_t (y) a^*_y a_x \\
&+ \frac{1}{2} \int \rd x \rd y \, V(x-y) \left( \ph_t (x) \ph_t (y)
a^*_x a^*_y +
\overline{\ph_t} (x) \overline{\ph_t} (y) a_x a_y \right) \\
&+\frac{1}{\sqrt{N}} \int \rd x \rd y \, V(x-y) \, a_x^* \left(
\ph_t (y)
\chi (\cN < N) a^*_y  + \overline{\ph_t} (y) a_y \chi (\cN < N) \right) a_x \\
&+\frac{1}{2N} \int \rd x \rd y \, V(x-y) \, a^*_x a^*_y a_y a_x \, .
\end{split}
\end{equation}
Observe, that the generator $\cK_N (t)$ has exactly the same form as the generator $\cL_N (t)$; the only difference is the presence of a cutoff in the number of particles $\cN$ inserted in the cubic term. Thanks to the cutoff in $\cN$ and to the factor $N^{-1/2}$ in front of the cubic term in $\cK_N (t)$, we can prove, making use of a Gronwall-type argument, that
\begin{equation}\label{eq:WNts1}
\langle \cW_N (t;s)\, \Omega, \cN \, \cW_N (t;s) \Omega \rangle \leq C \, e^{K |t-s|}\,.
\end{equation}
Actually, it is simple to see that the last inequality can be improved to
\begin{equation}\label{eq:WNts}
\langle \cW_N (t;s)\, \Omega, \cN^j \, \cW_N (t;s) \Omega \rangle \leq C_j \, e^{K_j |t-s|}\,.
\end{equation}
for every $j \in \bN$ and for appropriate $j$-dependent constants $C_j$ and $K_j$. To obtain (\ref{eq:lm1cou}), we still have to compare the dynamics $\cU_N (t;s)$ and $\cW_N (t;s)$. To this end, we first show weak a-priori bounds of the form
\begin{equation}\label{eq:weakbd}
\langle \cU_N (t;s) \, \psi , \cN^j \, \cU_N (t;s) \psi \rangle \leq C \, \langle \psi, (\cN +N +1)^j \psi \rangle
\end{equation}
for every $\psi \in \cF$ and for $j=1,2,3$ (these bounds hold uniformly in $t,s \in \bR$ and can be proven using the very definition of the unitary group $\cU_N (t;s)$; see \cite{RS}[Lemma 3.6]). Using (\ref{eq:weakbd}), we find
\begin{equation*}
\begin{split}
\langle \cU_N (t;0) &\Omega, \cN \left( \cU_N (t;0) - \cW_N
(t;0)\right) \Omega \rangle \\ = \; &\langle \cU_N (t;0) \Omega, \cN
\cU_N (t;0) \left( 1 - \cU_N (t;0)^* \cW_N (t;0)\right) \Omega
\rangle \\ = \; &-i\int_0^t \rd s \; \langle \cU_N (t;0) \Omega, \cN
\cU_N (t;s) \left(\cL_N (s) - \cK_N (s) \right) \cW_N
(s;0) \Omega \rangle \\
= \; &-\frac{i}{\sqrt{N}} \int_0^t \rd s \int \rd x \, \rd y  V(x-y)  \\
&\hspace{.2cm} \times \langle \cU_N (t;0) \Omega, \cN \cU_N (t;s)
a^*_x \left( \overline{\ph}_t (y) a_y \chi (\cN > N) + \ph_t (y)
\chi (\cN
>N) a_y^* \right) a_x \cW_N
(s;0) \Omega \rangle \\
= \; &-\frac{i}{\sqrt{N}}\int_0^t \rd s \int \rd x  \, \langle a_x
\cU_N (t;s)^* \cN \cU_N (t;0) \Omega, a (V(x-.)\ph_t) \chi (\cN
> N) a_x \cW_N (s;0) \Omega \rangle \\
&-\frac{i}{\sqrt{N}}\int_0^t \rd s \int \rd x  \langle a_x \cU_N
(t;s)^* \cN \cU_N (t;0) \Omega, \chi (\cN >N) a^* (V(x-.)\ph_t) a_x
\cW_N (s;0) \Omega \rangle\,.
\end{split}
\end{equation*}
Hence
\begin{equation*}
\begin{split}
\Big| \langle \cU_N (t;0) &\Omega, \cN \left( \cU_N (t;0) -
\cW_N (t;0)\right) \Omega \rangle \Big| \\ \leq
 \; &\frac{1}{\sqrt{N}} \int_0^t \rd s \int \rd x  \| a_x \cU_N (t;s)^* \cN
\cU_N (t;0) \Omega \| \, \| a (V(x-.)\ph_t)  a_x \chi (\cN > N+1)
\cW_N (s;0) \Omega \|  \\
&+\frac{1}{\sqrt{N}}\int_0^t \rd s \int \rd x  \| a_x \cU_N (t;s)^*
\cN \cU_N (t;0) \Omega \| \| a^* (V(x-.)\ph_t) a_x \chi (\cN
>N) \cW_N (s;0) \Omega \| \\
\leq \; &\frac{2 \sup_x \| V(x-.)\ph_t \|}{\sqrt{N}} \int_0^t \rd s  \|\cN^{1/2} \cU_N
(t;s)^* \cN \cU_N (t;0) \psi\| \, \|  \cN \chi (\cN > N)
\cW_N (s;0) \psi \| \, .
\end{split}
\end{equation*}
Therefore, using the inequality $\chi (\cN > N) \leq (\cN/N)^2$ and applying (\ref{eq:WNts}) (with $j=4$) and (\ref{eq:weakbd}) (first with $j=1$, and then with $j=3$) we can bound the two norms in the $s$-integral. It follows that
\[ \Big| \langle \cU_N (t;0) \Omega, \cN \left( \cU_N (t;0) -
\cW_N (t;0)\right) \Omega \rangle \Big|  \leq C e^{K t} \]
which, combined with (\ref{eq:WNts1}), implies (\ref{eq:lm1cou}).

\subsection{Time-evolution of Factorized Initial Data}

To prove Theorem \ref{thm:fact}, we express the factorized initial data as a linear combination of coherent states. Using the properties listed in Lemma \ref{lm:coh}, it is simple to check that \[ \{ 0, 0, \dots, 0, \ph^{\otimes N}, 0, 0, \dots \} =
\frac{(a^* (\ph))^N}{\sqrt{N!}} \Omega = d_N \int_0^{2\pi} \frac{\rd
\theta}{2\pi} \; e^{i\theta N} W ( e^{-i\theta} \sqrt{N} \ph) \Omega \]
with the constant
\begin{equation*}
d_N = \frac{\sqrt{N!}}{N^{N/2} e^{-N/2}} \simeq N^{1/4}\,.
\end{equation*}

The kernel of the one-particle density $\gamma_{N,t}^{(1)}$
associated with the solution of the Schr\"odinger equation
$\{ 0 , \dots , 0 , e^{-iH_N t} \ph^{\otimes N} , 0, \dots \}$
is thus given by (see (\ref{eq:margi}))
\begin{equation*}
\begin{split}
\gamma^{(1)}_{N,t} (x;y) =\; &\frac{1}{N} \left\langle \frac{(a^*
(\ph))^N}{\sqrt{N!}} \Omega, e^{i \cH_N t} a_y^* a_x e^{-i\cH_N t}
\frac{(a^* (\ph))^N}{\sqrt{N!}} \Omega \right\rangle
\\ = \; & \frac{d^2_N}{N} \int_0^{2\pi} \frac{\rd \theta_1}{2\pi} \int_0^{2\pi}
\frac{\rd \theta_2}{2\pi} \, e^{-i\theta_1 N} e^{i\theta_2 N}
\langle W(e^{-i\theta_1} \sqrt{N} \ph) \Omega, a_y^* (t) a_x (t)
W(e^{-i\theta_2} \sqrt{N} \ph) \Omega \rangle
\end{split}
\end{equation*}
where we introduced the notation $a_x (t) = e^{i\cH_N t} a_x
e^{-i\cH_N t}$. Next, we expand
\begin{equation}\label{eq:onepart}
\begin{split}
\gamma^{(1)}_{N,t} (x;y) = \; & \frac{d^2_N}{N} \int_0^{2\pi}
\frac{\rd \theta_1}{2\pi} \int_0^{2\pi} \frac{\rd \theta_2}{2\pi} \,
e^{-i\theta_1 N} e^{i\theta_2 N}  \left\langle W(e^{-i\theta_1}
\sqrt{N} \ph) \Omega, \, \left( a_y^*
(t) - e^{i\theta_1} \sqrt{N} \overline{\ph}_t (y) \right) \right. \\
&\hspace{3cm} \left. \times \left( a_x (t) - e^{-i\theta_2} \sqrt{N}
\ph_t (x) \right) \, W(e^{-i\theta_2} \sqrt{N} \ph) \Omega
\right\rangle
\\ & + \frac{d^2_N \, \overline{\ph}_t (y)}{\sqrt{N}} \int_0^{2\pi}
\frac{\rd \theta_1}{2\pi} \int_0^{2\pi} \frac{\rd \theta_2}{2\pi} \,
e^{-i\theta_1 (N-1)} e^{i\theta_2 N} \\ &\hspace{3cm} \times
\left\langle W(e^{-i\theta_1} \sqrt{N} \ph) \Omega, \left( a_x (t) -
e^{-i\theta_2} \sqrt{N} \ph_t (x) \right)W(e^{-i\theta_2} \sqrt{N}
\ph) \Omega \right\rangle
\\& + \frac{d^2_N \, \ph_t (x)}{\sqrt{N}} \int_0^{2\pi} \frac{\rd
\theta_1}{2\pi} \int_0^{2\pi} \frac{\rd \theta_2}{2\pi} \,
e^{-i\theta_1 N} e^{i\theta_2 (N-1)} \\ &\hspace{3cm} \times
\left\langle W(e^{-i\theta_1} \sqrt{N} \ph) \Omega, \left( a^*_y (t)
- e^{i\theta_1} \sqrt{N} \overline{\ph}_t (y) \right) \,
W(e^{-i\theta_2} \sqrt{N} \ph) \Omega \right\rangle \\ &+ d_N^2 \,
\ph_t (x) \overline{\ph}_t (y) \int_0^{2\pi} \frac{\rd
\theta_1}{2\pi} \int_0^{2\pi} \frac{\rd \theta_2}{2\pi} \,
e^{-i\theta_1 (N-1)} e^{i\theta_2 (N-1)} \\ &\hspace{3cm} \times
\left\langle W(e^{-i\theta_1} \sqrt{N} \ph) \Omega, \,
W(e^{-i\theta_2} \sqrt{N} \ph) \Omega \right\rangle \, .
\end{split}
\end{equation}
Since
\begin{equation*}
\begin{split} d_N \int_0^{2\pi} \frac{\rd \theta}{2\pi} e^{i\theta (N-1)}
W (e^{-i\theta} \sqrt{N} \ph) \Omega &= d_N e^{-N/2}
\sum_{j=0}^{\infty} \left( \int_0^{2\pi} \frac{\rd \theta}{2\pi}
e^{i\theta (N-1-j)} \right) N^{j/2} \frac{(a^* (\ph))^j}{j!} \Omega
\\ &= d_N \frac{e^{-N/2} N^{(N-1)/2}}{\sqrt{N-1!}} \frac{(a^*
(\ph))^{N-1}}{N-1!} \Omega \\ &= \ph^{\otimes N-1} \, ,
\end{split}
\end{equation*}
we find that
\begin{equation}\label{eq:gamma1-phi}
\begin{split}
\gamma^{(1)}_{N,t} & \, (x;y) - \ph_t (x) \overline{\ph}_t (y)
\\ = \; & \frac{d^2_N}{N} \int_0^{2\pi}
\frac{\rd \theta_1}{2\pi} \int_0^{2\pi} \frac{\rd \theta_2}{2\pi} \,
e^{-i\theta_1 N} e^{i\theta_2 N}  \\ &\hspace{.5cm} \times \left\langle W(e^{-i\theta_1}
\sqrt{N} \ph) \Omega, \, \left( a_y^*
(t) - e^{i\theta_1} \sqrt{N} \overline{\ph}_t (y) \right)
\left( a_x (t) - e^{-i\theta_2} \sqrt{N}
\ph_t (x) \right)  W(e^{-i\theta_2} \sqrt{N} \ph) \Omega
\right\rangle
\\ & + \frac{d_N \, \overline{\ph}_t (y)}{\sqrt{N}}
\int_0^{2\pi} \frac{\rd \theta_2}{2\pi} \, e^{i\theta_2 N} 
\left\langle \ph^{\otimes (N-1)}, \left( a_x (t) -
e^{-i\theta_2} \sqrt{N} \ph_t (x) \right) W(e^{-i\theta_2} \sqrt{N}
\ph) \Omega \right\rangle
\\& + \frac{d_N \, \ph_t (x)}{\sqrt{N}} \int_0^{2\pi} \frac{\rd
\theta_1}{2\pi} \,
e^{-i\theta_1 N} 
\left\langle W(e^{-i\theta_1} \sqrt{N} \ph) \Omega, \left( a^*_y (t)
- e^{i\theta_1} \sqrt{N} \overline{\ph}_t (y) \right) \,
\ph^{\otimes (N-1)} \right\rangle \,
\end{split}
\end{equation}
and thus
\begin{equation*}
\begin{split}
\left| \gamma^{(1)}_{N,t}  \, (x;y) - \ph_t (x) \overline{\ph}_t (y)\right| \leq \; & \frac{d^2_N}{N} \int_0^{2\pi}
\frac{\rd \theta_1}{2\pi} \int_0^{2\pi} \frac{\rd \theta_2}{2\pi} \, \| a_y \cU^{\theta_1}_N (t;0)   \Omega \| \, \| a_x \cU^{\theta_2}_N (t;0) \Omega \|
\\ & + \frac{d_N \, |\ph_t (y)|}{\sqrt{N}}
\int_0^{2\pi} \frac{\rd \theta}{2\pi}  \| a_x \cU_N^{\theta} (t;0) \Omega \| \\
& + \frac{d_N \, |\ph_t (x)|}{\sqrt{N}}
\int_0^{2\pi} \frac{\rd \theta}{2\pi}  \| a_y \cU_N^{\theta} (t;0) \Omega \|
\end{split}
\end{equation*}
where $\cU_N^{\theta} (t;s)$ is defined as (\ref{eq:cUN}), with $\ph_t$ replaced by $e^{i\theta} \ph_t$ (we are using, here, the fact that, although the Hartree equation is nonlinear, $e^{i\theta} \ph_t$ is always a solution if $\ph_t$ is). Since $d_N \simeq N^{1/4}$, it follows from the bound (\ref{eq:lm1cou}) for the growth of the expectation of $\cN$ with respect to the fluctuation evolution $\cU_N^{\theta} (t;s)$ (the bound clearly holds uniformly in $\theta$) that
\[ \left\| \gamma^{(1)}_{N,t} - |\ph_t \rangle \langle \ph_t| \right\|_{\text{HS}} \leq \frac{C}{N^{1/4}} e^{K t} \, \]
and therefore (using the argument presented after (\ref{eq:HS1}) that
\begin{equation}\label{eq:fintrace} \tr\, \left| \gamma^{(1)}_{N,t} - |\ph_t \rangle \langle \ph_t| \right| \leq \frac{C}{N^{1/4}} e^{K t} \, . \end{equation}

\medskip

To improve the decay in $N$ on the r.h.s. of (\ref{eq:fintrace}) from $N^{-1/4}$ to $N^{-1/2}$ (as claimed in Theorem~\ref{thm:fact}), it is necessary to study the second and third error terms on the r.h.s. of (\ref{eq:gamma1-phi}) more precisely; for the details, see \cite{RS}[Lemma 4.2].

\appendix

\section{Non-Standard Sobolev- and Poincar{\'e} Inequalities}
\label{app:poin}\setcounter{equation}{0}

In this section, we collect some non-standard Sobolev- and Poincar{\'e}-type inequalities which are very useful when dealing with singular potentials.

\begin{lemma}[Sobolev-type inequalities]\label{lm:sob1}
Let $\psi \in L^2 (\bR^6, \rd x_1 \rd x_2)$. If $V \in L^{3/2} (\bR^3)$, we have
\begin{equation}\label{eq:3/2}
\langle \psi, V(x_1 -x_2) \psi \rangle \leq C \| V \|_{3/2} \, \langle \psi, (1-\Delta_1) \psi \rangle \, .
\end{equation}
If $V \in L^1 (\bR^3)$, then \begin{equation}\label{eq:Vsob} \langle \psi, V(x_1 -x_2) \psi \rangle \leq C \| V \|_1 \, \langle \psi, (1-\Delta_1) (1-\Delta_2) \psi \rangle \end{equation}
\end{lemma}
The first bound follows from a H\"older inequality followed by a standard Sobolev inequality (in the variable $x_1$, with fixed $x_2$). The proof of (\ref{eq:Vsob}) can be obtained through the same arguments used in the proof of the next Poincar{\'e}-type inequality (see \cite{EY}).
\begin{lemma}[Poincar{\'e}-type inequality] \label{lm:poin}
Suppose that $h\in L^1 (\bR^3)$ is a probability density with $\int \rd x \, |x|^{1/2} \, h(x) < \infty$. For $\alpha >0$, let $h_{\alpha} (x) = \a^{-3} h (x/\a)$. Then we have, for every $0\leq \kappa<1/2$,
\[ \left| \langle \ph, \left( h_{\alpha} (x_1 -x_2) -\delta (x_1 -x_2) \right) \psi \rangle \right| \leq C \a^{\kappa} \langle \ph, (1-\Delta_1) (1-\Delta_2) \ph \rangle^{1/2} \langle \psi, (1-\Delta_1) (1-\Delta_2) \psi \rangle^{1/2} \,. \]
\end{lemma}
\begin{proof}
We rewrite the inner product in Fourier space.
\begin{equation*}
\begin{split}
\langle \ph, \Big( h_{\alpha} &(x_1 -x_2) -\delta (x_1 -x_2) \Big) \psi \ra\\ &= \int \rd p_1 \rd p_2 \rd q_1 \rd q_2 \rd x \, \delta (p_1 + p_2 - q_1 -q_2) \, \overline{\widehat{\ph}} (p_1, p_2) \, \widehat{\psi} (q_1 , q_2) \, h (x) \, \left( e^{i\alpha (p_1 -q_1) \cdot x} - 1 \right).
\end{split}
\end{equation*}
Using that $|e^{i\a (p_1 - q_1) \cdot x} -1| \leq \a^{\kappa} |x|^{\kappa} |p_1- q_1|^{\kappa}$, we obtain
\begin{equation*}
\begin{split}
\Big| \langle \ph, \Big( h_{\alpha} (x_1 -x_2) - &\delta (x_1 -x_2) \Big) \psi \ra \Big| \\ \leq \; \a^{\kappa} &\left( \int \rd x \, h(x) |x|^{\kappa} \right) \\ &\times \int \rd p_1 \rd p_2 \rd q_1 \rd q_2  \, \delta (p_1 + p_2 - q_1 -q_2) \left( |p_1|^{\kappa} + | q_1|^{\kappa} \right) |\widehat{\ph} (p_1, p_2)| \, | \widehat{\psi} (q_1 , q_2)| \,.
\end{split}\end{equation*}
We show how to control the term proportional to $|p_1|^{\kappa}$; the other term can be handled similarly.
\begin{equation*}
\begin{split}
\Big| \langle \ph, &\left( h_{\alpha} (x_1 -x_2) -\delta (x_1 -x_2) \right) \psi \ra \Big| \\ \leq \;& C \a^{\kappa} \, \int \rd p_1 \rd p_2 \rd q_1 \rd q_2  \, \delta (p_1 + p_2 - q_1 -q_2) \frac{|p_1|^{\kappa} (1+p_1^2)^{(1-\kappa)/2} (1+p_2^2)^{1/2}}{(1+q_1^2)^{1/2} (1+q_2^2)^{1/2}} | \widehat{\ph} (p_1, p_2)| \\ & \hspace{9cm} \times \frac{(1+q_1^2)^{1/2} (1+q_2^2)^{1/2}}{(1+p_1^2)^{(1-\kappa)/2} (1+p_2^2)^{1/2}} | \widehat{\psi} (q_1 , q_2)| \\ \leq \; & C \a^{\kappa}\left( \int \rd p_1 \rd p_2 \rd q_1 \rd q_2  \, \delta (p_1 + p_2 - q_1 -q_2) \frac{(1+p_1^2)(1+p_2^2)}{(1+q_1^2)(1+q_2^2)} | \widehat{\ph} (p_1, p_2)|^2 \right)^{1/2} \\ &\hspace{2cm} \times \left( \int \rd p_1 \rd p_2 \rd q_1 \rd q_2  \, \delta (p_1 + p_2 - q_1 -q_2)\frac{(1+q_1^2) (1+q_2^2)}{(1+p_1^2)^{1-\kappa} (1+p_2^2)} | \widehat{\psi} (q_1 , q_2)|^2 \right)^{1/2} \\
\leq \; &C \a^{1/2} \la \ph, (1-\Delta_1)(1-\Delta_2) \ph \ra^{1/2} \la \psi, (1-\Delta_1)(1-\Delta_2) \psi \ra^{1/2} \\ &\hspace{2cm} \times \left( \sup_p \int \rd q \frac{1}{(1+q^2)(1+(p-q)^2)} \right)^{1/2} \left( \sup_q \int \rd p \frac{1}{(1+p^2) (1+(q-p)^2)^{1-\kappa}} \right)^{1/2}\,.
\end{split}
\end{equation*}
The claim follows because
\begin{equation}\label{eq:int10}
\sup_{q\in \bR^3} \int \rd p \frac{1}{(1+p^2) (1+(q-p)^2)^{1-\kappa}} \leq C
\end{equation}
for all $\kappa <1/2$. To prove (\ref{eq:int10}) we consider the three regions $|p| > 2 |q|$, $|q|/2 \leq |p| \leq 2 |q|$ and $|p| < |q|/2$ separately. Since $|p-q|>|p|/2$ for $|p|>2|q|$, it follows that
\[ \int_{|p|> 2|q|} \frac{\rd p}{(1+p^2) (1+(q-p)^2)^{1-\kappa}} \leq \int_{|p| >2|q|} \frac{\rd p}{\left(1+\frac{p^2}{4}\right)^{2-\kappa}} < C \int \frac{\rd p}{\left(1+p^2\right)^{2-\kappa}} < \infty \] for $\kappa < 1/2$, uniformly in $q$. For $|p| < |q|/2$, we use the fact that $|q-p|>|q|/2$, and we obtain \[\int_{|p|<|q|/2} \frac{\rd p}{(1+p^2) (1+(q-p)^2)^{1-\kappa}} \leq \frac{C}{(1+q^2)^{1-\kappa}} \int_{|p| < |q|/2} \frac{\rd p}{1+p^2} \leq \frac{C |q|}{(1+q^2)^{1-\kappa}} \, \] which is bounded uniformly in $q$. Finally, in the region $|q|/2 \leq |p| \leq 2|q|$, we use that
\begin{equation*}
\begin{split}
\int_{|q|/2 < |p| < 2 |q|} \frac{\rd p}{(1+p^2)(1+(q-p)^2)^{1-\kappa}} \leq \frac{C}{(1+q^2)} \int_{|p| < 3|q|} \frac{\rd p}{(1+p^2)^{1-\kappa}} \leq C \frac{|q|^{2\kappa + 1}}{1+q^2} < \infty
\end{split}
\end{equation*}
uniformly in $q\in \bR^3$, for all $\kappa < 1/2$.
\end{proof}

In the approach developed in \cite{ESY5} for the case of large interaction potential we can only prove weaker estimates on the solution $\psi_{N,t}$ of the Schr\"odinger equation. As discussed in Section \ref{sec:largeV}, we can only prove that
\[ \la W_{N,(i,j)}^* \psi_{N,t}, \left( (\nabla_i \cdot \nabla_j)^2 -\Delta_i -\Delta_j + 1\right) W_{N,(i,j)}^* \psi_{N,t} \ra \leq C \] uniformly in $N$ and $t$. For this reason, we need estimates which only require the boundedness of the expectation of this particular combination of derivatives. The next lemma gives a Sobolev inequality of this type.
\begin{lemma}\label{lm:VL1}
Suppose $V \in L^1 (\bR^3)$. Then
\begin{equation*}
\begin{split}
\left|\langle \ph, V (x_1 -x_2) \psi \rangle \right| \leq C \| V \|_1 \, \langle
&\psi, \left( (\nabla_1 \cdot \nabla_2)^2 -\Delta_1 - \Delta_2 + 1 \right) \psi \rangle^{1/2} \; \\ &
 \times \langle \ph, \left( (\nabla_1 \cdot \nabla_2)^2 -\Delta_1 - \Delta_2 + 1 \right)
\ph \rangle^{1/2}\end{split}
\end{equation*}
for every $\psi,\ph \in L^2 ( \bR^6, \rd x_1  \rd x_2)$.
\end{lemma}

\begin{proof}
Switching to Fourier space, we find
\begin{equation*}
\begin{split}
\langle \ph, V (x_1 -x_2) \psi \rangle
=\; & \int \rd p_1 \rd p_2 \rd q_1 \rd q_2 \; \overline{\widehat\ph (p_1, p_2)}
\widehat{\psi} (q_1, q_2) \, \widehat{V} (q_1 -p_1)\, \delta (p_1 + p_2 - q_1 -q_2) \, .
\end{split}
\end{equation*}
Therefore, by a weighted Schwarz inequality,
\begin{equation*}
\begin{split}
\Big| \langle \ph, V &(x_1 -x_2) \psi \rangle \Big|  \\
\leq \; & \| \widehat{V} \|_{\infty} \left( \int \rd p_1 \rd p_2 \rd q_1 \rd q_2 \;
\frac{(p_1 \cdot p_2)^2 + p_1^2 + p^2_2 +1}{(q_1 \cdot q_2)^2 + q_1^2 + q^2_2 +1} \,
 |\widehat\ph (p_1, p_2)|^2 \delta (p_1 + p_2 - q_1 -q_2) \right)^{1/2} \\ &
\hspace{1cm} \times \left(\int \rd p_1 \rd p_2 \rd q_1 \rd q_2
\frac{(q_1 \cdot q_2)^2 + q_1^2 + q^2_2 +1}{(p_1 \cdot p_2)^2 + p_1^2 + p^2_2 +1}
 |\widehat{\psi} (q_1, q_2)|^2 \,\delta (p_1 + p_2 - q_1 -q_2)\right)^{1/2} \\
\leq \; & \| V \|_{1} \; \left( \sup_{p} \int \rd q \;
\frac{1}{(q \cdot (p -q))^2 + q^2 + (p - q)^2 +1}\right) \\ & \hspace{1cm} \times
\left\langle \psi, \left((\nabla_1 \cdot \nabla_2)^2 -\Delta_1 - \Delta_2 + 1 \right)
\psi \right\rangle^{1/2} \, \left\langle \ph, \left((\nabla_1 \cdot \nabla_2)^2 -\Delta_1 - \Delta_2 + 1 \right)
\ph \right\rangle^{1/2}\,.
\end{split}
\end{equation*}
The lemma follows from
\begin{equation}\label{eq:int1}
\sup_{p \in \bR^3} \int \rd q \; \frac{1}{(q \cdot (p -q))^2 + q^2 + (p - q)^2 +1} < \infty \, . \end{equation}
To prove (\ref{eq:int1}), we write
\begin{equation}\label{eq:int2}
\begin{split}
\int \rd q \; \frac{1}{(q \cdot (p -q))^2 + q^2 + (p - q)^2 +1} = \; & \int_{|q-\frac{p}{2}| > |p|}
\rd q \; \frac{1}{\left( \left(q-\frac{p}{2}\right)^2 - \frac{p^2}{4}\right)^2 + q^2 + (p - q)^2 +1} \\ &
+ \int_{|q-\frac{p}{2}| < |p|} \rd q \; \frac{1}{\left( \left(q-\frac{p}{2}\right)^2 -
\frac{p^2}{4}\right)^2 + q^2 + (p - q)^2 +1}\,.
\end{split}
\end{equation}
The first term on the r.h.s. of the last equation is bounded by
\begin{equation*}
\begin{split}
\int_{|q-\frac{p}{2}| > |p|} \rd q \; \frac{1}{\left( \left(q-\frac{p}{2}\right)^2 - \frac{p^2}{4}\right)^2
 + q^2 + (p - q)^2 +1} \leq \; & \int_{|q-\frac{p}{2}| > |p|} \rd q \; \frac{1}{\frac{9}{16}
\left|q-\frac{p}{2}\right|^4 +1}
\\ \leq \; & \frac{16}{9} \int_{\bR^3} \rd q \; \frac{1}{|q|^4  +1} < \infty\,,
\end{split}
\end{equation*}
uniformly in $p \in \bR^3$. As for the second term on the r.h.s. of (\ref{eq:int2}), we observe that
\begin{equation*}
\begin{split}
\int_{|q-\frac{p}{2}| < |p|} \rd q \; &\frac{1}{\left( \left(q-\frac{p}{2}\right)^2 -
\frac{p^2}{4}\right)^2 + q^2 + (p - q)^2 +1} \\ = \; & \int_{|x| < |p|} \rd x \;
\frac{1}{\left( x^2 - \frac{p^2}{4}\right)^2 + \left( x+\frac{p}{2} \right)^2 + \left(x-\frac{p}{2} \right)^2 +1} \\
= \; & 4\pi \int_{0}^{|p|} \rd r \frac{r^2}{\left( r^2 - \frac{|p|^2}{4}\right)^2 +  2r^2 + \frac{|p|^2}{2} +1} \\
\leq \; & C |p|^2 \int_{-|p|/2}^{|p|/2} \rd r \frac{1}{ r^2 \left( r + |p|\right)^2 +
 \left( r+\frac{|p|}{2} \right)^2 + \frac{|p|^2}{4} +1}  \\
\leq\; & C \int_{-|p|/2}^{|p|/2} \rd r \, \frac{1}{ r^2  + 1}  \leq C \int_{\bR} \rd r \, \frac{1}{r^2 + 1} < \infty ,
\end{split}
\end{equation*}
uniformly in $p$.
\end{proof}

\end{document}